\pdfoutput=1
\documentclass[11pt]{article}
\usepackage{authblk}
\usepackage{extarrows}
\usepackage[top=2cm, bottom=2cm, left=2cm, right=2cm]{geometry}
\usepackage{centernot}

\usepackage{color}
\usepackage{helvet}         
\usepackage{courier}        
\usepackage{type1cm}        

\usepackage{framed}
\usepackage{tikz}
\usepackage{makeidx}         
\usepackage{graphicx}        
\usepackage{multicol}        
\usepackage[bottom]{footmisc}

\usepackage{amsmath}
\usepackage{amssymb}
\usepackage{bbold}
\usepackage{amsthm}
\usepackage{subcaption}
\usepackage{sidecap}
\usepackage{floatrow}
\usepackage{pdflscape}
\usepackage{pdflscape}
\usepackage{comment}
\usepackage[font=small]{caption}
\usepackage{enumitem}
\usepackage{esint}

\usepackage{scalerel}

\newtheorem{theorem}{Theorem}
\newtheorem{corollary}[theorem]{Corollary}
\newtheorem{lemma}[theorem]{Lemma}

\theoremstyle{remark}

\newtheorem{remark}[theorem]{\bf Remark}

\numberwithin{theorem}{section}
\numberwithin{question}{section}
\numberwithin{figure}{section}
\numberwithin{equation}{section}

\begin{document}

\title{Conformally covariant probabilities, operator product expansions, \\ and logarithmic correlations in two-dimensional critical percolation}
\bigskip{}
\author[1]{Federico Camia\thanks{federico.camia@nyu.edu}}
\author[2]{Yu Feng\thanks{yufeng\_proba@163.com. This work was conducted while Y.F. was at Tsinghua University, China.}}

\affil[1]{NYU Abu Dhabi, UAE \& Courant Institute, USA}
\affil[2]{University of Michigan, USA}

\date{}

%

%

\global\long\def\CR{\mathrm{CR}}
\global\long\def\ST{\mathrm{ST}}
\global\long\def\SF{\mathrm{SF}}
\global\long\def\cov{\mathrm{cov}}
\global\long\def\dist{\mathrm{dist}}
\global\long\def\SLE{\mathrm{SLE}}
\global\long\def\unitD{\mathbb{D}}
\global\long\def\hSLE{\mathrm{hSLE}}
\global\long\def\CLE{\mathrm{CLE}}
\global\long\def\GFF{\mathrm{GFF}}
\global\long\def\inte{\mathrm{int}}
\global\long\def\ext{\mathrm{ext}}
\global\long\def\opi{\overline{\pi}}
\global\long\def\orho{\overline{\rho}}
\global\long\def\inrad{\mathrm{inrad}}
\global\long\def\outrad{\mathrm{outrad}}
\global\long\def\dimH{\mathrm{dim}}
\global\long\def\capa{\mathrm{cap}}
\global\long\def\diam{\mathrm{diam}}
\global\long\def\free{\mathrm{free}}
\global\long\def\Dist{\mathrm{Dist}}
\global\long\def\hF{{}_2\mathrm{F}_1}
\global\long\def\simple{\mathrm{simple}}
\global\long\def\even{\mathrm{even}}
\global\long\def\odd{\mathrm{odd}}
\global\long\def\st{\mathrm{ST}}
\global\long\def\usf{\mathrm{USF}}
\global\long\def\Leb{\mathrm{Leb}}
\global\long\def\LP{\mathrm{LP}}
\global\long\def\coulomb{\LH}
\global\long\def\coulombnew{\LG}
\global\long\def\kfunc{p}
\global\long\def\OO{\mathcal{O}}
\global\long\def\parti{\mathbf{Q}}
\global\long\def\rad{\mathrm{rad}}

\global\long\def\eps{\epsilon}
\global\long\def\ov{\overline}
\global\long\def\U{\mathbb{U}}
\global\long\def\T{\mathbb{T}}
\global\long\def\HH{\mathbb{H}}
\global\long\def\LA{\mathcal{A}}
\global\long\def\LB{\mathcal{B}}
\global\long\def\LC{\mathcal{C}}
\global\long\def\LD{\mathcal{D}}
\global\long\def\LF{\mathcal{F}}
\global\long\def\LK{\mathcal{K}}
\global\long\def\LE{\mathcal{E}}
\global\long\def\LG{\mathcal{G}}
\global\long\def\LI{\mathcal{I}}
\global\long\def\LJ{\mathcal{J}}
\global\long\def\LL{\mathcal{L}}
\global\long\def\LM{\mathcal{M}}
\global\long\def\LN{\mathcal{N}}
\global\long\def\LQ{\mathcal{Q}}
\global\long\def\LR{\mathcal{R}}
\global\long\def\LT{\mathcal{T}}
\global\long\def\LS{\mathcal{S}}
\global\long\def\LU{\mathcal{U}}
\global\long\def\LV{\mathcal{V}}
\global\long\def\LW{\mathcal{W}}
\global\long\def\LX{\mathcal{X}}
\global\long\def\LY{\mathcal{Y}}
\global\long\def\PartF{\mathcal{Z}}
\global\long\def\LH{\mathcal{H}}
\global\long\def\LJ{\mathcal{J}}
\global\long\def\R{\mathbb{R}}
\global\long\def\C{\mathbb{C}}
\global\long\def\N{\mathbb{N}}
\global\long\def\Z{\mathbb{Z}}
\global\long\def\E{\mathbb{E}}
\global\long\def\PP{\mathbb{P}}
\global\long\def\QQ{\mathbb{Q}}
\global\long\def\A{\mathbb{A}}
\global\long\def\one{\mathbb{1}}
\global\long\def\bn{\mathbf{n}}
\global\long\def\MR{MR}
\global\long\def\cond{\,|\,}
\global\long\def\la{\langle}
\global\long\def\ra{\rangle}
\global\long\def\tree{\Upsilon}
\global\long\def\prob{\mathbb{P}}
\global\long\def\hm{\mathrm{Hm}}
\global\long\def\cross{\mathrm{Cross}}

\global\long\def\sf{\mathrm{SF}}
\global\long\def\wr{\varrho}

\global\long\def\Im{\operatorname{Im}}
\global\long\def\Re{\operatorname{Re}}

\global\long\def\ud{\mathrm{d}}
\global\long\def\pder#1{\frac{\partial}{\partial#1}}
\global\long\def\pdder#1{\frac{\partial^{2}}{\partial#1^{2}}}
\global\long\def\der#1{\frac{\ud}{\ud#1}}

\global\long\def\bZnn{\mathbb{Z}_{\geq 0}}

\global\long\def\Vfunc{\LG}
\global\long\def\gfunc{g^{(\rr)}}
\global\long\def\hfunc{h^{(\rr)}}

\global\long\def\SimplexInt{\rho}
\global\long\def\CubeInt{\widetilde{\rho}}

\global\long\def\ii{\mathfrak{i}}
\global\long\def\rr{\mathfrak{r}}
\global\long\def\chamber{\mathfrak{X}}
\global\long\def\Wchamber{\mathfrak{W}}

\global\long\def\SimplexIntKappa8{\SimplexInt}

\global\long\def\nested{\boldsymbol{\underline{\Cap}}}
\global\long\def\unnested{\boldsymbol{\underline{\cap\cap}}}
\global\long\def\unnested{\boldsymbol{\underline{\cap\cap}}}

\global\long\def\acycle{\vartheta}
\global\long\def\bcycle{\tilde{\acycle}}
\global\long\def\Gloop{\Theta}

\global\long\def\metric{\mathrm{dist}}

\global\long\def\adj#1{\mathrm{adj}(#1)}

\global\long\def\bs{\boldsymbol}

\global\long\def\edge#1#2{\langle #1,#2 \rangle}
\global\long\def\graph{G}

\newcommand{\conn}{\vartheta_{\scaleobj{0.7}{\mathrm{RCM}}}}
\newcommand{\hatconn}{\widehat{\vartheta}_{\mathrm{RCM}}}
\newcommand{\realpt}{\smash{\mathring{x}}}
\newcommand{\corrind}{\LC}
\newcommand{\bssymb}{\pi}
\newcommand{\PRCM}{\mu}
\newcommand{\coeff}{p}
\newcommand{\MainConst}{C}

\global\long\def\removeLink{/}
\maketitle

\begin{abstract}
The large-scale behavior of two-dimensional critical percolation is expected to be described by a \emph{conformal field theory} (CFT).
Moreover, this putative CFT is believed to be of the logarithmic type, exhibiting logarithmic corrections to the most commonly encountered behavior of CFT correlations.
While constructing a full-fledged percolation CFT is still an open problem, in this paper we prove various CFT features of the scaling limit of two-dimensional critical percolation.
In particular, we provide the first rigorous proof of the emergence of logarithmic singularities in the scaling limit of connection probabilities.
More precisely, we study several connectivity events, including \emph{arm-events} and the events that a vertex is \emph{pivotal} or belongs to the percolation \emph{backbone}, whose probabilities have conformally covariant scaling limits and can be interpreted as CFT correlation functions. For some of these probabilities, we prove asymptotic expansions that can be regarded as CFT \emph{operator product expansions} (OPEs).
Our analysis identifies various logarithmic singularities and explains the geometric mechanism that produces them.
In follow-up work, the results of this paper are used to define a percolation energy field and its \emph{logarithmic partner}.
\end{abstract}

\noindent\textbf{Keywords:} 
critical percolation, conformal invariance, connection probabilities, continuum scaling limit, conformal field theory, logarithmic correlations, operator product expansion\\ 

\noindent\textbf{MSC:} 60K35, 82B43, 82B27, 82B31, 60J67, 81T27

\tableofcontents

\section{Introduction and discussion of the main results}
\subsection{Percolation background and motivation}
Percolation was introduced by Broadbent and Hammersley to model the spread of a gas or a fluid through a porous medium \cite{broadbent_hammersley_1957}.
Because it is one of the simplest mathematical models of a continuous phase transition and due to the large number of applications, it has been extensively studied by both physicists and mathematicians (see, e.g.,~\cite{SABERI20151}). The two-dimensional version of the model is particularly well understood (see \cite{kesten1982percolation,stauffer1994introduction,grimmett1999percolation,br2006}), including at the critical density (the phase transition point), where its large-scale properties are believed to be described by a conformal field theory (see, e.g., \cite{francesco1997conformal,henkel1999conformal}).

A conformal field theory (CFT) is a special type of quantum field theory which is invariant under scale and more general conformal transformations.
One of the most studied examples is provided by the two-dimensional critical Ising model, defined in terms of \emph{spin} (random) variables located at the vertices of a regular lattice.
In the Ising model with no external magnetic field, when the temperature parameter approaches a critical value, the \emph{correlation length} of the model, defined as the rate of exponential decay in space of the covariance between two spin variables, diverges.
As a consequence, the large-scale statistical (thermal) fluctuations become scale invariant and the covariance has a power law behavior~\cite{McCoyWu}.
In such a situation, it was proposed by Polyakov and collaborators \cite{polyakov1970conformal,BelavinPolyakovZamolodchikovInfiniteConformalSymmetry2D,BelavinPolyakovZamolodchikovInfiniteConformalSymmetryCritical2D} that the large-scale fluctuations should be described by a CFT.
In order to obtain full scale and conformal invariance, one needs to take a \emph{continuum} (\emph{scaling}) \emph{limit}, in which the lattice spacing is sent to zero and the spin variables are replaced by a magnetization ``field,'' a generalized function which behaves homogeneously under scale and more general conformal transformations~\cite{CamiaGarbanNewman}.
In the continuum limit, the spin $n$-point functions, i.e., the expectations of products of $n$ spin variables, converge to the $n$-point correlation functions of the magnetization field.
Like the magnetization field itself, its correlation functions also transform homogeneously under scale and conformal transformations~\cite{ChelkakHonglerIzyurovConformalInvarianceCorrelationIsing}. 
This property, named \emph{conformal covariance}, is characteristic of so-called \emph{(conformal) primary fields} (or \emph{operators}), which are the building blocks of any CFT.
The Ising magnetization field and its $n$-point functions provide a prototypical example of a primary field and its correlations.

Unlike the Ising model, Bernoulli percolation is a purely geometric model in which independent, identically distributed, binary (say, black/white, open/closed, $+$/$-$) random variables are placed at the vertices or edges of a regular lattice.
If interpreted as the components of an Ising-type lattice field, the percolation field is trivial in the sense that all its $n$-point functions
are products of expectations of single variables, due to the independence of the variables.
Instead, these variables are used to define \emph{clusters}, i.e., maximal connected (according to the lattice adjacency notion) subsets of vertices with the same label (or vertices joined by edges with the same label), which are the main objects of interest.
The connectivity properties of clusters are encoded in the $n$-point \emph{connection probabilities}, the probabilities that $n \geq 2$ vertices belong to the same cluster.
The percolation phase transition corresponds to the emergence of an infinite cluster, signaled by the divergence of the mean size of the cluster of the origin, as the density of vertices or edges with a given label is increased to a critical value.

In the 1990s, Michael Aizenman conjectured that, at the critical density, connection probabilities should have a conformally covariant scaling limit and therefore behave like the correlation functions of a CFT (see~\cite{LanglandsPouliotSaintAubinConfInv2dPerco,aizenman1998continuum,Aizenman1998}).
The conjecture implies that one can try to apply the CFT formalism to critical percolation using connection probabilities instead of correlation functions.

Following the groundbreaking introduction of the Schramm-Loewner evolution (SLE) \cite{SchrammFirstSLE}, percolation was one of the earliest models for which conformal invariance was established \cite{Smirnov:Critical_percolation_in_the_plane}, followed shortly by a proof of convergence of percolation interfaces to SLE curves \cite{CamiaNewmanPercolationFull,CamiaNewmanPercolation}.
Nevertheless, no progress was made until recently in the direction of proving Aizenman's conjecture on connection probabilities and rigorously establishing a percolation CFT (see~\cite{SchrammSmirnovPercolation} for a discussion).

In~\cite{Cam23}, one of us proved Aizenman's conjecture for critical (Bernoulli) site percolation on the triangular lattice and constructed a percolation lattice field whose $n$-point functions are non-trivial combinations of connection probabilities and consequently have a conformally covariant scaling limit.

In this paper, building on the results and ideas of~\cite{Cam23}, we move one step forward and start to explore the CFT structure of critical percolation.
In particular, we identify new percolation events whose probabilities have conformally covariant scaling limits and can therefore be interpreted as $n$-point functions of primary fields,
providing further evidence for the assumption that the large-scale properties of percolation can be described using the CFT formalism.
Among the probabilities we study, some can be interpreted as correlation functions involving the two most fundamental percolation fields, the so-called \emph{density field} and \emph{energy field}, which are, in some sense, the analogs of the Ising magnetization and energy fields.

One of the main results of this paper is the identification of logarithmic singularities in critical percolation, including in the four-point function of the density field, which provides the first rigorous confirmation of similar predictions made in the physics literature.
This is particularly interesting because the field of logarithmic CFTs is significantly less developed than that of ordinary CFTs, despite the fact that logarithmic CFTs have attracted considerable attention in recent years due to their role in the study of important physical models and phenomena such as the Wess-Zumino-Witten (WZW) model, the quantum Hall effect, disordered critical systems, self-avoiding polymers, and the Fortuin-Kasteleyn (FK) model (see~\cite{creutzig2013logarithmic} for a review).

We show the presence of logarithmic singularities by studying the asymptotic behavior of certain four-point functions as two of the four points approach each other.
This analysis is of independent interest for at least two reasons.
First of all, it elucidates the ``physical'' mechanism that leads to the appearance of logarithmic singularities, in terms of lattice quantities.
To the best of our knowledge, this mechanism had not been previously explained, even in the physics literature. 
Secondly, it provides a connection with fundamental CFT concepts such as those of \emph{operator product expansion} (\emph{OPE}) and \emph{fusion rule}, and a way to rigorously understand them at the lattice level.
We briefly discuss these important concepts in the next section before presenting the main results of the paper.

Building on the results and techniques of this paper, the logarithmic CFT structure of the scaling limit of two-dimensional critical percolation is further explored in~\cite{PercolationEnergy}, where we identify a percolation ``energy'' field and its ``logarithmic partner,'' related to the four-arm event. In that paper, we rigorously show that the two- and three-point functions of the percolation ``energy'' field and its ``partner'' possess the structure predicted by Gurarie~\cite{Gur93} for pairs of logarithmic partner fields in a logarithmic CFT (see also~\cite{creutzig2013logarithmic}).

\subsection{CFT background and terminology} \label{sec::CFT_background}

In this section, we provide a brief and informal introduction to conformal field theory (CFT).
Our goal here is not to give a complete overview of CFT, but simply to provide some background and, more importantly, to introduce some of the ideas and terminology that appear in the physics literature, in order to be able to compare our (mathematically rigorous) results, presented in the next section, with that literature.
Excellent reviews of two-dimensional CFT include~\cite{Cardy1990,Ginsp1990,francesco1997conformal,henkel1999conformal}.

{
A natural starting point to discuss CFT is the theory of continuous phase transitions or critical phenomena.
The study of critical phenomena leads naturally to scale-invariant theories, which do not possess a characteristic length and therefore look the same at all distances.
This is a consequence of the emergence of fluctuations occurring over many distance scales as the critical point is approached, resulting in the divergence of the characteristic length.
Onsager's exact solution of the two-dimensional Ising model~\cite{OnsagerCrystalStatisticsI} provides an explicit, well-known and very influential example of scale invariance at the critical (phase transition) point.

The divergence of the characteristic length at the critical point is a general phenomenon and it is now understood that all critical points are described by scale-invariant theories.
This is consistent with Wilson’s renormalization group (RG) theory of phase transitions, introduced by Wilson and Kogut (1974) and Wilson (1983), according to which continuous phase transitions are described by the fixed points of RG flows.
Since (real-space) RG flows correspond to the rescaling of space, fixed points are automatically scale-invariant.

A theory is scale-invariant if it is unchanged by the uniform rescaling (dilation) of all coordinates.
Scale transformations are a special case of conformal transformations, which locally look like a rotation and a dilation.
Hence, conformal invariance can be considered a natural extension of scale invariance.
It is indeed natural to ask if there is a difference between theories that are invariant under uniform dilations and theories that are invariant under non-uniform dilations (i.e., conformal transformations), where the scale factor is allowed to vary with the position.

It is now understood, as first conjectured by Polyakov \cite{polyakov1970conformal}, that scale invariance generically implies conformal invariance (although there are exceptions, see~\cite{CardyRiva2005339}), so that physically relevant scale-invariant theories are typically conformally invariant.
For the critical Ising model, conformal invariance was proved by Smirnov \cite{SmirnovHolomorphicFermion}.

So far, our discussion has been very general and somewhat vague, since we haven't specified what types of theories we are referring to.
Two obvious questions concern the objects that populate these theories and how to express the conformal invariance in terms of those objects.

The main conceptual framework that has emerged in physics to study the large-scale behavior of critical systems, where scale invariance emerges, is that of quantum field theory (QFT), where the objects of interest are ``fields,'' which encode the spatial dependence of the physical quantities of interest (e.g., the magnetization in the ferromagnet-paramagnet phase transition modelled by the Ising model).
This is a very powerful framework, developed throughout much of the 20th century, with numerous important applications ranging from particle physics to condensed matter physics.
The term QFT is also used to indicate individual theories, and some of the most fundamental and most successful physical theories are QFTs (e.g., the Standard Model of Particle Physics).

A conformal field theory (CFT) is a conformally invariant QFT.
As with QFT, the term CFT is also used collectively, in this case, to denote the study of all quantum field theories that are invariant under conformal transformations.

In addition to their role in the theory of critical phenomena, CFTs are also extremely important for the study of more general QFTs that appear in high-energy (particle) physics and quantum condensed matter physics.
The study and classification of CFTs is a major goal of modern theoretical physics.}
In particular, the two-dimensional version of the theory has seen a rapid development following the groundbreaking work of Belavin, Polyakov and Zamolodchikov \cite{BelavinPolyakovZamolodchikovInfiniteConformalSymmetry2D,BelavinPolyakovZamolodchikovInfiniteConformalSymmetryCritical2D}.

The approach of Belavin, Polyakov and Zamolodchikov allows to study the large-scale behavior of a critical system using the constraints of conformal symmetry alone, with no need to consider the microscopic details of the system. The crucial idea is the so-called \emph{conformal bootstrap}, first described by Ferrara, Grillo, and Gatto~\cite{ferrara1973conformal} and Polyakov~\cite{polyakov1974nonhamiltonian}, which combines conformal invariance with the operator product expansion (OPE), another powerful concept which can be traced back to Wilson~\cite{WilsonNonlagrangian}
and Kadanoff~\cite{KadanoffOperatorAlgebra}. 
Following these developments, CFT has become an indispensable tool in the theory of two-dimensional critical phenomena.
More recently, many interesting results were obtained for conformal field theories in higher dimensions, including a precise determination of the critical exponents of the critical 3D Ising model~\cite{kos2016precision}.

{In the conformal bootstrap approach, a CFT is defined as a set of functions satisfying certain axiomatic properties.
These functions are interpreted as the $n$-point correlations, $\langle \phi_1(x_1) \ldots \phi_n(x_n) \rangle$, of some physical fields, $\phi_1, \ldots, \phi_n$, probed at locations $x_1, \ldots, x_n$, respectively.
The fields in a correlation function can be different or repeated and we assume that their position does not matter (i.e., the functions are invariant under permutations of the fields).
Mathematically, the fields can be thought of as indices that identify the functions that define the CFT.

The collection of all fields is called the \emph{field content} or \emph{operator content} of the theory.
The latter term comes from the fact that, in quantum field theory, classical fields are treated as canonical coordinates and elevated to the role of operators acting on a Hilbert space of possible physical states (the same way as, in quantum mechanics, canonical coordinates are turned into operators in the canonical quantization of classical mechanics).}

{The two main properties that $n$-point correlation functions are assumed to satisfy concern the way they transform when the $n$ points are mapped to other points by a conformal map and the existence of a short-distance expansion as two of the $n$ points are brought close to each other.
(Other simple properties are assumed to hold, but they are either relatively trivial or somewhat technical, and correspond to statements such as the existence of a ``unit field,'' invariance under permutations, the existence of a discrete set of scaling dimensions $\alpha_j$ associated to the fields $\phi_j$---see, e.g.~\cite{rychkov20203d}).
}

{Firstly, correlation functions are assumed to be \emph{conformally covariant} in the sense that, if $\varphi$ is a conformal transformation, then
    \begin{align} \label{eq:conf-cov}
    \langle \phi_1(\varphi(x_1)) \ldots \phi_n(\varphi(x_n)) \rangle = \Big(\prod_{i=1}^{n} \vert \varphi'(x_i) \vert^{-\alpha_i}\Big) \langle \phi_1(x_1) \ldots \phi_n(x_n) \rangle,
    \end{align}
    where the $\alpha_i$'s are called \emph{scaling dimensions}.
The same scaling dimension $\alpha_i$ has to appear in all $n$-point functions involving the field $\phi_i$, so a unique scaling dimension is associated to each field and therefore $\alpha_i$ is called the scaling dimension of the field $\phi_i$.
(However, different fields can have the same scaling dimension.)
It is furthermore assumed that, for all $\phi_i$,
\begin{align*}
    \langle\phi_i(x)\rangle=0
\end{align*}
and
\begin{align} \label{eq:2-pt-fnct}
    \langle\phi_i(x_1)\phi_i(x_2)\rangle = \vert x_1-x_2 \vert^{-2\alpha_i},
\end{align}
where it is easy to check that the behavior of the two-point function above is, up to a multiplicative constant, the only one compatible with \eqref{eq:conf-cov}.
Moreover, it is assumed that distinct fields are ``orthogonal'' in the sense that, if $\phi_i\neq\phi_j$,
    \begin{align} \label{eq:orthogonal}
    \langle\phi_i(x_1)\phi_j(x_2) \rangle=0.
    \end{align}

It is not difficult to check that conformal covariance, Eq.~\eqref{eq:conf-cov}, determines not only the functional form of two-point functions, but also that of three-point functions, namely,
    \begin{align} \label{eq:3-pt-fnct}
    \langle \phi_1(x_1)\phi_2(x_2)\phi_3(x_3) \rangle = C_{\phi_1\phi_2\phi_3} \vert x_1-x_2 \vert^{-(\alpha_1+\alpha_2-\alpha_3)} \vert x_1-x_3 \vert^{-(\alpha_1+\alpha_3-\alpha_2)} \vert x_2-x_3 \vert^{-(\alpha_2+\alpha_3-\alpha_1)},
    \end{align}
where $C_{\phi_1\phi_2\phi_3}$ is called a \emph{structure constant} of the theory.
Four-point functions are the simplest correlations not fully determined, up to a multiplicative constant, by conformal covariance and are therefore of particular interest.

The second fundamental property of correlation functions is the \emph{operator product expansion} (\emph{OPE}), whose origin can be traced back to the work of Wilson and Zimmermann~\cite{WilsonNonlagrangian,WilsonZimmermannOPE} 
on QFT and of Kadanoff in statistical mechanics~\cite{KadanoffOperatorAlgebra}.
For a CFT four-point function, the OPE reads
\begin{align} \label{eq:OPE}
    \langle\phi_1(x_1)\phi_2(x_2)\phi_3(x_3)\phi_4(x_4)\rangle \sim \sum_{k \geq 0} \frac{C_{\phi_1\phi_2}^{\phi_k}}{\vert x_1-x_2 \vert^{\alpha_1+\alpha_2-\alpha_k}} F_{\phi_k}^{\phi_3\phi_4}(x,x_3,x_4) \; \text{ as } x_1,x_2 \to x,
\end{align}
where we assume that the theory contains a discrete set of fields $\{\phi_k\}_{k \geq 0}$ and the sum runs over all possible fields of the theory, including the ``unit field'' $\phi_0=\mathbb{1}$ with $\alpha_0=0$, which is identically equal to one and satisfies $\langle \mathbb{1}(x) \rangle = 1$ and $\langle \mathbb{1}(x_0) \phi_1(x_1) \ldots \phi_n(x_n) \rangle = \langle \phi_1(x_1) \ldots \phi_n(x_n) \rangle$.
However, the coefficients $C_{\phi_1\phi_2}^{\phi_k}$ can be zero, so the right-hand side of the OPE doesn't necessarily contain all fields.
Which fields are present is encoded in the so-called \emph{fusion rules}.

In the simplest case, $F_{\phi_k}^{\phi_3\phi_4}$ takes the form of a three-point function,
    \begin{align} \label{eq:OPE-term}
    F_{\phi_k}^{\phi_3\phi_4}(x,x_3,x_4) = C_{\phi_k\phi_3\phi_4} \vert x-x_3 \vert^{-(\alpha_k+\alpha_3-\alpha_4)} \vert x-x_4 \vert^{-(\alpha_k+\alpha_4-\alpha_3)} \vert x_3-x_4 \vert^{-(\alpha_3+\alpha_4-\alpha_k)},
    \end{align}
so that, formally, one can write
\begin{align} \label{eq:formal-OPE}
    \phi_1(x_1) \phi_2(x_2) \sim \sum_{k \geq 0} \frac{C_{\phi_1\phi_2}^{\phi_k}}{\vert x_1-x_2 \vert^{\alpha_1+\alpha_2-\alpha_k}} \, \phi_k(x) \; \text{ as } x_1,x_2 \to x.
\end{align}
In CFT language, Eq.~\eqref{eq:formal-OPE} can be interpreted as saying that the ``fusion'' of fields $\phi_1$ and $\phi_2$ produces the fields $\phi_k$ in the right-hand side of the equation.

Applying this form of the OPE to the three-point function $\langle\phi_1(x_1)\phi_2(x_2)\phi_3(x_3)\rangle$ and using \eqref{eq:orthogonal} and \eqref{eq:2-pt-fnct} shows that $C_{\phi_1\phi_2}^{\phi_3}=C_{\phi_1\phi_2\phi_3}$, so the coefficients of the OPE are the structure constants of Eq.~\eqref{eq:3-pt-fnct}.
Moreover, using the OPE \eqref{eq:formal-OPE} one can reduce any $n$-point function for $n \geq 4$ to an $(n-1)$-point function.
It is therefore possible, at least in principle, to ``solve'' a CFT by calculating the structure constants and applying the OPE repeatedly, provided that one knows the field content of the theory.
Conformal invariance is very useful in carrying out this plan, since it poses significant constraints on the correlation functions and consequently on the OPE coefficients.
These observations are at the heart of the conformal bootstrap approach to CFT.

What we just described is the simplest case, but in general the function $F_{\phi_k}^{\phi_3,\phi_4}$ in Eq.~\eqref{eq:OPE} can contain additional terms, corresponding to derivatives of three-point functions, as well as logarithmic terms.
The latter case is particularly interesting and relevant for percolation, as we will show (see also~\cite{CamiaFeng2024logarithmic}, where some of the results of this paper are announced and their consequences for the percolation CFT are discussed).
A CFT that contains logarithmic terms in its four-point functions and OPEs is called \emph{logarithmic} (see, e.g.,~\cite{creutzig2013logarithmic}).
Logarithmic CFTs emerge naturally in several contexts, have many applications and have attracted considerable attention since the work of Gurarie~\cite{Gur93}.
However, despite significant recent progress, the field of logarithmic CFTs is still considerably less developed than that of ordinary CFTs.

From an algebraic perspective, it is well understood that logarithmic CFTs are linked to non-diagonal representations of the Virasoro algebra and can be analyzed by studying the indecomposable modules of the Virasoro algebra. In recent years, this approach, combined with numerical techniques and conformal bootstrap methods, has led to tremendous progress, both on the side of a general theory of logarithmic CFTs and for specific models~\cite{JS19,PRS19,HJS20,NR21,NRJ23}.
However, these methods are not rigorous and, moreover, rely on the implicit assumption that an appropriate scaling limit exists and admits a
CFT description, as well as on additional assumptions on the field content of the putative CFT. In practice, when explicit expressions for correlation functions are found, they are typically obtained by solving differential equations derived from the Ward identities (see, e.g.,~\cite{francesco1997conformal}).
This approach does not explain the physical mechanism leading to the appearance of logarithms in terms of the lattice variables of the original model. Our analysis elucidates this physical mechanism for percolation, while providing a connection with the fundamental CFT concepts described in this section.

\subsection{Definitions and main results}
\label{subsec::def_result}
Let $\mathcal{T}$ denote the triangular lattice and let  $\mathcal{H}$ denote the hexagonal lattice dual to $\mathcal{T}$. Then each vertex of $\mathcal{T}$ corresponds to a face (that is, a hexagon) of $\mathcal{H}$ in a natural way and we often identify them. Assume that $\mathcal{T}$ and $\mathcal{H}$ are embedded in $\mathbb{C}$ in such a way that one of the vertices of $\mathcal{T}$ coincides with the origin of $\mathbb{C}$ (see Figure~\ref{fig::embedding}).
We consider critical site percolation on the scaled triangular lattice $a\mathcal{T}$, where each vertex of $a\mathcal{T}$ is assigned a black or white label independently, with equal probability. We denote this measure by $\mathbb{P}^a$. For a subgraph $G$ of $\mathcal{T}$, we define its (outer) boundary $\partial G$ as
\begin{equation*} 
\partial G:=\{v\in \mathcal{T}\setminus G: \exists u\in G \text{ such that }u\sim v \},
\end{equation*}
where $u\sim v$ denotes that $u$ and $v$ are adjacent in $G$.
We call a sequence of vertices $(v_1,\ldots,v_{n+1})$ a black (respectively, white) \emph{path} if $v_1,\ldots,v_{n+1}$ are all black (resp., white) vertices and $v_j \sim v_{j+1}$ for $j=1,\ldots,n$.
If $v_{n+1}=v_1$, then $(v_1,\ldots,v_{n+1})$ is called a black (resp., white) \emph{circuit}.

\begin{figure}
	\includegraphics[width= 0.3\textwidth]{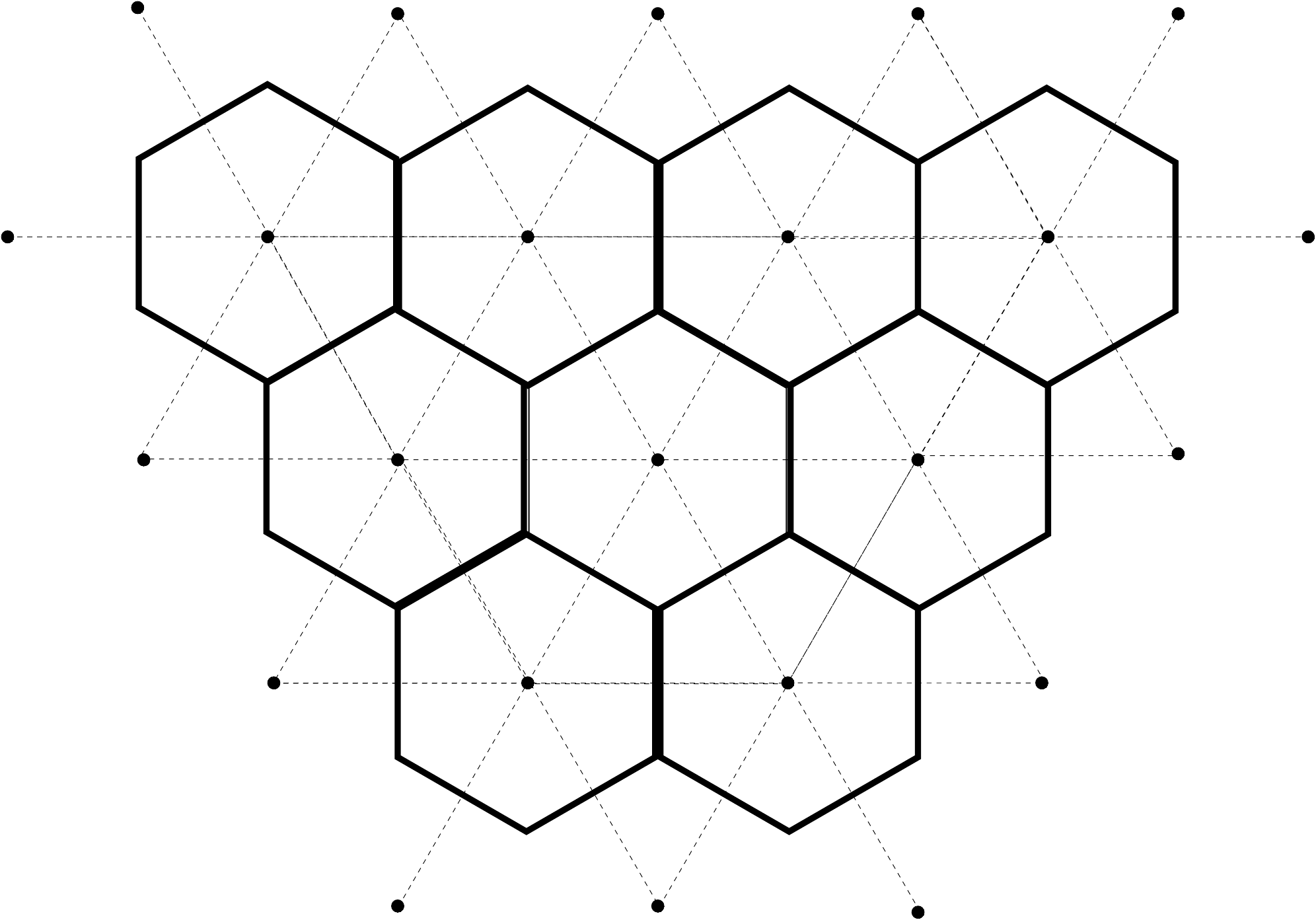}
	\caption{Embedding of the triangular and hexagonal lattices in $\mathbb{C}$.}
	\label{fig::embedding}
\end{figure}

Let $\{\mathcal{C}_j^a\}_j$ denote the collection of black clusters {(maximal connected components of the graph consisting of black vertices, with the adjacency relation $\sim$)} on $a\mathcal{T}$ and assign to each cluster $\mathcal{C}_j^a$ a random spin $\sigma_j=\pm 1$, where $\{\sigma_j\}_j$ is a collection of symmetric, $(\pm1)$-valued, i.i.d. random variables. Then for each $x^a\in a\mathcal{T}$, we let 
\begin{equation} \label{def:lattice-field}
	S_{x^a}=\begin{cases}
		\sigma_j,& \text{ if } x^a\in \mathcal{C}_j^a,\\
		0, & \text{ if }x^a \text{ is white}. 
	\end{cases}
\end{equation} 
We denote by $\langle \cdot \rangle^a$ the expectation with respect to the distribution of $\{S_{x^a}\}_{x^a\in a\mathcal{T}}$. 

The interest and relevance of the lattice field \eqref{def:lattice-field} stems from the fact that it provides an explicit version of the percolation \emph{density} or \emph{spin field} discussed in the physics literature.
As shown in \cite{Cam23}, its correlation functions, when appropriately rescaled, have a conformally covariant scaling limit as does the field itself.
This determines the two-point function up to a multiplicative constant, while all odd correlation functions are zero by symmetry.
The four-point function, discussed below, is not determined by conformal covariance and reveals the connection with the physics literature (see, e.g.,~\cite{GV18,HGJS20}).

Let $x_1,x_2,x_3,x_4\in \mathbb{C}$ be four distinct points on the complex plane and suppose that $x_1^a,x_2^a,x_3^a,x_4^a\in a\mathcal{T}$ satisfy $\lim_{a\to 0}x_j^a=x_j$ for $1\leq j\leq 4$. As explained in~\cite[Eq.~(1.9)]{Cam23}, we have 
\begin{align*}
	\langle S_{x_1^a}\ldots S_{x_4^a}\rangle^a=&\mathbb{P}^a\left[x_1^a\xlongleftrightarrow{B}x_2^a\xlongleftrightarrow{B}x_3^a\xlongleftrightarrow{B}x_4^a\right]+\mathbb{P}^a\left[x_1^a\xlongleftrightarrow{B}x_2^a\centernot{\xlongleftrightarrow{B}}x_3^a\xlongleftrightarrow{B}x_4^a\right]\\
	&+\mathbb{P}^a\left[x_1^a\xlongleftrightarrow{B}x_3^a\centernot{\xlongleftrightarrow{B}}x_2^a\xlongleftrightarrow{B}x_4^a\right]+\mathbb{P}^a\left[x_1^a\xlongleftrightarrow{B}x_4^a\centernot{\xlongleftrightarrow{B}}x_2^a\xlongleftrightarrow{B}x_3^a\right],
\end{align*}
{where $\{x^a_i \xlongleftrightarrow{B} x^a_j \}$ denotes the event that $x^a_i$ and $x^a_j$ belong to the same black cluster and $\{x^a_i \centernot{\xlongleftrightarrow{B}} x^a_j \}$ denotes its complement.}

For $z\in \mathbb{C}$ and $\epsilon>0$, we define $B_\epsilon(z)=\{w\in \mathbb{C}: |z-w|<\epsilon\}$ and write
\begin{equation*}
	\pi_a:=\mathbb{P}^a\left[0\xlongleftrightarrow{B}\partial B_{1}(0)\right],
\end{equation*}
where $\{0\xlongleftrightarrow{B}\partial B_{1}(0)\}$ denotes a \emph{one-arm event}, i.e., the event that there exists a black path connecting $0$ to $\partial B_{1}(0)$.

It is shown in the proof of~\cite[Theorem~1.5]{Cam23} that the following limits exist (their sum gives the function $\mathrm{C}_{D,n}$ in Eq.~(1.10) of~\cite[Theorem~1.5]{Cam23} for $n=4$):
\begin{align*}
	P(x_1\xlongleftrightarrow{B}x_2\xlongleftrightarrow{B}x_2\xlongleftrightarrow{B}x_4):=& \, \lim_{a\to 0} \pi_a^{-4}\times \mathbb{P}^a\left[x_1^a\xlongleftrightarrow{B}x_2^a\xlongleftrightarrow{B}x_3^a\xlongleftrightarrow{B}x_4^a\right],\\
	P(x_1\xlongleftrightarrow{B}x_2\centernot{\xlongleftrightarrow{B}}x_3\xlongleftrightarrow{B} x_4):=& \, \lim_{a\to 0} \pi_a^{-4}\times\mathbb{P}^a\left[x_1^a\xlongleftrightarrow{B}x_2^a\centernot{\xlongleftrightarrow{B}}x_3^a\xlongleftrightarrow{B}x_4^a\right],\\
	P(x_1\xlongleftrightarrow{B}x_3\centernot{\xlongleftrightarrow{B}}x_2\xlongleftrightarrow{B} x_4):=& \, \lim_{a\to 0} \pi_a^{-4}\times\mathbb{P}^a\left[x_1^a\xlongleftrightarrow{B}x_3^a\centernot{\xlongleftrightarrow{B}}x_2^a\xlongleftrightarrow{B}x_4^a\right],\\
	P(x_1\xlongleftrightarrow{B}x_4\centernot{\xlongleftrightarrow{B}}x_2\xlongleftrightarrow{B} x_3):=& \, \lim_{a\to 0} \pi_a^{-4}\times\mathbb{P}^a\left[x_1^a\xlongleftrightarrow{B}x_4^a\centernot{\xlongleftrightarrow{B}}x_2^a\xlongleftrightarrow{B}x_3^a\right].
\end{align*}
{Moreover, the limits are covariant under M\"obius transformations, in the sense of Eq.~\eqref{eq:conf-cov}.
The same covariance property, with exponents $\alpha_1=\ldots=\alpha_4=5/48$, is satisfied by the function}
\begin{align*}
	& C(x_1,x_2,x_3,x_4) := \lim_{a\to 0}\pi_a^{-4}\times \langle S_{x_1^a}\cdots S_{x_4^a}\rangle^a \\
	& \qquad = P(x_1\xlongleftrightarrow{B}x_2\xlongleftrightarrow{B}x_2\xlongleftrightarrow{B}x_4)+P(x_1\xlongleftrightarrow{B}x_2\centernot{\xlongleftrightarrow{B}}x_3\xlongleftrightarrow{B} x_4) \\
	& \qquad \quad + P(x_1\xlongleftrightarrow{B}x_3\centernot{\xlongleftrightarrow{B}}x_2\xlongleftrightarrow{B} x_4) + P(x_1\xlongleftrightarrow{B}x_4\centernot{\xlongleftrightarrow{B}}x_2\xlongleftrightarrow{B} x_3).
\end{align*}

As mentioned earlier, conformal covariance is not sufficient to determine the functional form of a function of four variables.
A possible asymptotic expansion of $C(x_1,x_2,x_3,x_4)$, when two of the four points are close to each other, is suggested in~\cite[Section~2.3]{Cam23}.
However, the heuristic analysis in~\cite[Section~2.3]{Cam23} fails to identify the presence of a logarithmic term, which is the most interesting feature of the expansion in Theorem~\ref{thm::expansion} below.

\begin{theorem} \label{thm::expansion}
	There are two universal constants $C_1,C_2\in(0,\infty)$ such that the following holds for the function $C(x_1,x_2,x_3,x_4)$ defined above:
	\begin{equation} \label{eqn::C_expansion}
		C(x_1,x_2,x_3,x_4)\sim C_1|x_2-x_1|^{-\frac{5}{24}}\left(|x_3-x_4|^{-\frac{5}{24}}+ C_2 |x_2-x_1|^{\frac{5}{4}} F(x,x_3,x_4) \left| \log\left|x_2-x_1\right|\right|\right),
	\end{equation}
as $x_1,x_2\to x\in \mathbb{C}\setminus \{x_3,x_4\}$, namely,
\begin{equation*}
	\lim_{x_1,x_2\to x} \frac{C(x_1,x_2,x_3,x_4)-C_1|x_3-x_4|^{-\frac{5}{24}}\times |x_2-x_1|^{-\frac{5}{24}}}{|x_2-x_1|^{\frac{25}{24}}\times \left| \log\left|x_2-x_1\right|\right|}= C_1C_2F(x,x_3,x_4),
\end{equation*}
where 
\begin{equation*} 
	F(x,x_3,x_4)=|x-x_3|^{-\frac{5}{4}}|x-x_4|^{-\frac{5}{4}}|x_3-x_4|^{\frac{25}{24}}. 
\end{equation*}
\end{theorem}
\begin{remark}
We emphasize that the independence of labels at different vertices in percolation is crucial to get the logarithmic term in Theorem~\ref{thm::expansion}. More precisely, independence is needed to get~\eqref{eqn::independence_log} in the proof of Lemma~\ref{lem::expansion_aux1}. 
\end{remark}

\begin{remark} \label{rem::2nfunction}
		Let $n\geq 1$ be an integer and let $x_1,\ldots,x_{2n}\in \mathbb{C}$ be $2n$ distinct points. Suppose that $x_j^a\in a\mathcal{T}$ satisfy $\lim_{a\to 0} x_j^a=x_j$ for $1\leq j\leq 2n$. Then according to~\cite[Theorem~1.5]{Cam23}, the following limit also exists:
		\begin{equation*}
			C(x_1,x_2,\ldots,x_{2n}):=\lim_{a\to 0}\pi_a^{-2n}\times \langle S_{x_1^a}\cdots S_{x_{2n}^a}\rangle^a.
		\end{equation*}
		One can still use the strategy in the proof of Theorem~\ref{thm::expansion} to study the asymptotic expansion of the $2n$-point function $C(x_1,\ldots,x_{2n})$ when two of $\{x_1,x_2,\ldots,x_{2n}\}$ are close to each other. 
\end{remark}

The presence of a logarithm in \eqref{eqn::C_expansion} is particularly interesting because it provides a rigorous example of a logarithmic divergence in a percolation four-point function.
Although the large-scale properties of two-dimensional critical percolation are expected to be described by a logarithmic CFT  (see, e.g.,~\cite{mathieu2007percolation,VJS12}),
finding explicit examples of logarithmic singularities has been challenging and until now there was no rigorous proof of such a singularity in a correlation function.

Equation~\eqref{eqn::C_expansion} can be interpreted as an operator product expansion (OPE), as discussed in Section~\ref{sec::CFT_background}.
Writing \eqref{eqn::C_expansion} as
\begin{equation*}
C(x_1,x_2,x_3,x_4)\sim \frac{C_1}{|x_2-x_1|^{\frac{5}{24}}} |x_3-x_4|^{-\frac{5}{24}} + \frac{C_1C_2}{|x_2-x_1|^{\frac{5}{48}+\frac{5}{48}-\frac{5}{4}}} F(x,x_3,x_4)\left| \log\left|x_2-x_1\right|\right| + \ldots,
\end{equation*}
we can analyze its terms by comparing this expression with \eqref{eq:OPE} and learn something about the fusion rules of the putative percolation CFT, which determine what fields are present in the right-hand side of the OPE.
The first term in the right-hand side involves the spin (density) field \eqref{def:lattice-field}, with scaling dimension $5/48$, and the ``unit field,'' with scaling dimension $0$.
The second term involves the spin (density) field and a new ``field'' with scaling dimension $5/4$.
The discussion in~\cite[Section~2.3]{Cam23} and the proof of Theorem~\ref{thm::expansion} show that this term is related to the \emph{four-arm event}, i.e., the event that there are four paths with alternating labels (black/white) crossing an annulus.
In the physics literature, this event is related to the so-called \emph{energy field}.
The ellipsis represents possible contributions from other ``fields,'' but these are higher order terms, which vanish faster than $\vert x_1-x_2 \vert^{25/24}\vert\log\vert x_1-x_2 \vert\vert$ as $x_1,x_2 \to x$.
Hence, we can interpret the OPE above as saying that the ``fusion'' of two spin (density) fields produces the ``unit field'' and the ``four-arm/energy field,'' plus possibly other fields whose contributions to the four-point function \eqref{eqn::C_expansion} are of higher order in $\vert x_1-x_2 \vert$.
The reader is referred to~\cite{CamiaFeng2024logarithmic} for a more detailed discussion of the CFT implications of~\eqref{eqn::C_expansion}.

\begin{remark}
  The interpretation above is corroborated by the observation that the function $F(x,x_3,x_4)$ in Theorem~\ref{thm::expansion} can be obtained as the scaling limit of the three-point function $\langle \mathcal{E}_{x^a}S_{x_3^a}S_{x_4^a}\rangle^a$ involving the \emph{discrete energy} $\mathcal{E}_{x^a}:=S_{x^a-a}S_{x^a+a}-\langle S_{x^a-a}S_{x^a+a}\rangle^a$.
  More precisely, using the techniques outlined in Section~\ref{subsec::explain_proof} and applied in Section~\ref{sec::3}, one can show that
		\begin{align} \label{eqn::<ESS>}
			\lim_{a\to 0} a^{-\frac{5}{4}}\vert \log a\vert^{-1}\pi_a^{-2}\times \langle \mathcal{E}_{x^a}S_{x_3^a}S_{x_4^a}\rangle^a= C F(x,x_3,x_4), 
		\end{align}
		for some universal constant $C\in (0,\infty)$.
    The logarithmic scaling in \eqref{eqn::<ESS>} suggests that the correct normalization for the energy field $\mathcal{E}$ is $a^{-\frac{5}{4}}\vert \log a\vert^{-1}$. The presence of the logarithm in the normalization has interesting ``physical'' consequences, as discussed in~\cite{CamiaFeng2024logarithmic}. The behavior of the energy field is further investigated in the follow-up work~\cite{PercolationEnergy}.
\end{remark}

Our next result shows that the probabilities of certain events involving several four-arm events at different locations, when appropriately rescaled, have a conformally covariant scaling limit and behave like CFT correlation functions.
This supports the claim that, in the CFT description of critical percolation, one should associate a conformal field to the occurrence of a four-arm event.

\begin{figure} 
	\begin{subfigure}[b]{0.48\textwidth}
		\begin{center}
			\includegraphics[width=0.85\textwidth]{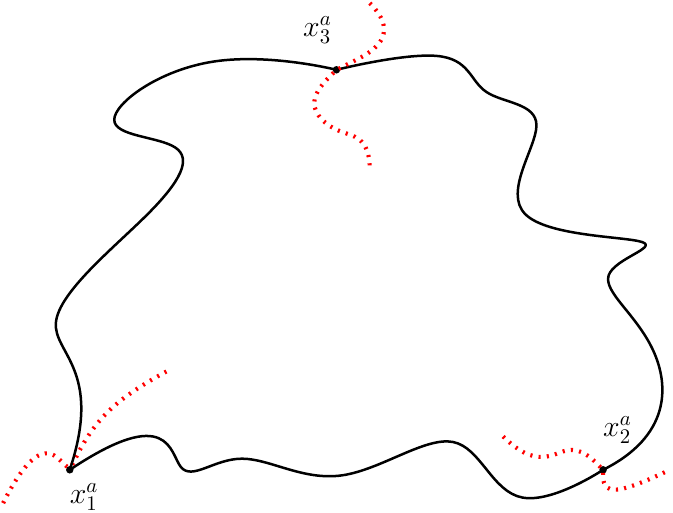}
		\end{center}
		\caption{The event $\mathcal{R}^a_3(x_1^a,x_2^a,x_3^a)$;}
	\end{subfigure}
	\begin{subfigure}[b]{0.48\textwidth}
		\begin{center}
			\includegraphics[width=0.85\textwidth]{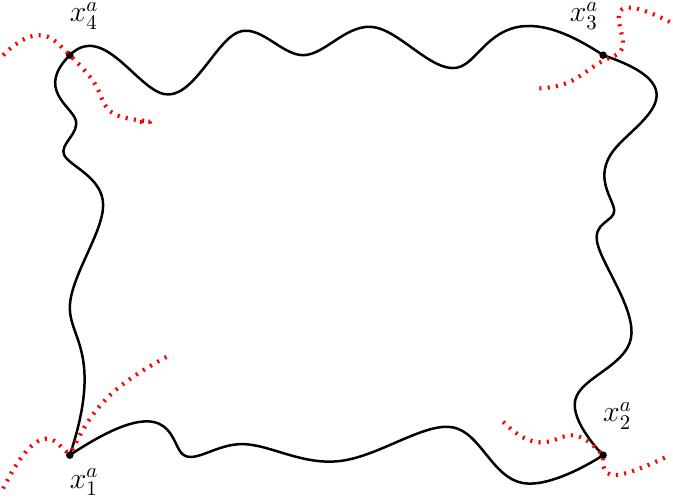}
		\end{center}
		\caption{The event $\mathcal{R}^a_4(x_1^a,x_2^a,x_3^a,x_4^a)$. }
	\end{subfigure}
	\caption{The events $\mathcal{R}_i^a(x_1^a,\ldots,x_i^a)$ for $i\in \{3,4\}$. The black, solid lines represent black paths and the red, dotted lines represent white paths. The $i$ black paths together with vertices $x_1^a,\ldots,x_i^a$ form naturally a lattice circuit. Note that the $i$ white paths that are inside (resp., outside) of this circuit belong to the same white cluster.}
	\label{fig::Rfigures}
	\end{figure}

Let $x_1^a,x_2^a,x_3^a,x_4^a\in a\mathcal{T}$ be four vertices satisfying $|x_j^a-x_k^a|\geq 10a$ for $1\leq j<k\leq 4$. For $i\in \{3,4\}$, we define $\mathcal{R}_i^a(x_1^a,\ldots,x_i^a)$ to be the event that there are $i$ black paths belonging to $i$ different black clusters connecting $x_j^a$ to $x_{j+1}^a$ for $1\leq j\leq i$, respectively, where we use the convention that $x_{i+1}^a=x_1^a$. See Figure~\ref{fig::Rfigures} for an illustration of the events $\mathcal{R}_i^a(x_1^a,\ldots,x_i^a)$ for $i\in \{3,4\}$.
Moreover, we write
\begin{equation*}
	\orho_a=\mathbb{P}^a\left[0\xlongleftrightarrow{BWBW}\partial B_1(0)\right],
\end{equation*}
where $\{0\xlongleftrightarrow{BWBW}\partial B_1(0)\}$ denotes the event that there are four paths with alternating {labels (black/white)} connecting $0$ to $\partial B_1(0)$.

\begin{theorem} \label{thm::4arm_expansion}
Let $x_1,\ldots,x_4\in \mathbb{C}$ be four distinct points.  Suppose that $x_1^a,\ldots,x_4^a\in a\mathcal{T}$ satisfy $x_j^a\to x_j$ as $a\to 0$ for $1\leq j\leq 4$. Let $i\in \{3,4\}$. Then the following statements hold:
\begin{enumerate}[label=\textnormal{(\arabic*)}, ref=\arabic*]
	\item \label{item::existence_limit} 
	The limit
	\begin{equation} \label{eqn::def_R}
		R_i(x_1,\ldots,x_i):=\lim_{a\to 0} \overline{\rho}_a^{-i} \times \mathbb{P}^a\left[\mathcal{R}_i^a\left(x_1^a,\ldots,x_i^a\right)\right]
	\end{equation}
	exists and belongs to $(0,\infty)$. 
\item \label{item::conformal_cova} 
Let $\varphi$ be any non-constant M\"obius transformation such that $\varphi(x_j)\neq \infty$ for $1\leq j\leq 4$. Then we have 
\begin{equation*}
	R_i(\varphi(x_1),\ldots,\varphi(x_i))=R_i(x_1,\ldots,x_i) \times \prod_{j=1}^i |\varphi'(x_j)|^{-\frac{5}{4}}.
\end{equation*} 
As a consequence, there exists a universal constant $C_3\in (0,\infty)$ such that 
\begin{equation*}
R_3(x_1,x_2,x_3)=C_3|x_1-x_2|^{-\frac{5}{4}}|x_1-x_3|^{-\frac{5}{4}}|x_2-x_3|^{-\frac{5}{4}}.
\end{equation*} 
\item  \label{item::expansion}
Let $x\in \mathbb{C}\setminus \{x_1,x_2\}$. Then there exists a universal constant $C_4\in (0,\infty)$ such that 
\begin{equation} \label{eqn::expansion_R}
	\lim_{x_3,x_4 \to x} \frac{R_4(x_1,x_2,x_3,x_4)}{|x_3-x_4|^{-\frac{5}{4}}}=C_4R_3(x_1,x_2,x). 
\end{equation}
\end{enumerate}
\end{theorem}

\begin{remark}
As in Remark~\ref{rem::2nfunction}, one can define $R_{i}(x_1,\ldots,x_i)$ for $i>4$ and study the asymptotic expansion of $R_i$ when two of $\{x_1,\ldots,x_i\}$ are close to each other, using the strategy in the proof of Theorem~\ref{thm::4arm_expansion}, Item~\ref{item::expansion}.
\end{remark}

Theorem \ref{thm::4arm_expansion} shows that the function $R_i$, which describes the correlation between $i$ four-arm events, is conformally covariant.
Moreover, Eq.~\eqref{eqn::expansion_R} can be interpreted as an OPE.
Writing \eqref{eqn::expansion_R} as
\begin{equation*}
    R_4(x_1,x_2,x_3,x_4) = C_3 \, C_4 \, |x_3-x_4|^{\frac{5}{4}} |x_1-x_2|^{-\frac{5}{4}}|x_1-x|^{-\frac{5}{4}}|x_2-x|^{-\frac{5}{4}} + o(|x_3-x_4|^{\frac{5}{4}}) \; \text{ as } x_3,x_4 \to x
\end{equation*}
and comparing this expression with \eqref{eq:OPE}, we see that, at the lowest order in $\vert x_3-x_4 \vert$, the ``fusion'' of two ``four-arm fields'' produces a ``four-arm field.''

To state our next result, we need to define the concept of \emph{pivotal vertex}.
Given an event $\mathcal{A}$ and a percolation configuration in which all vertices except $z^a$ have been assigned a label, we say that $z^a$ is \emph{pivotal for} $\mathcal{A}$ if $\mathcal{A}$ occurs when $z^a$ is black and does not occur when $z^a$ is white.
The event that a vertex $z^a$ is pivotal for $\mathcal{A}$ is the set of all percolation configurations such that $z^a$ is pivotal for $\mathcal{A}$.

We write
\begin{equation*}
	\orho_a=\mathbb{P}^a\left[0\xlongleftrightarrow{BWBW}\partial B_1(0)\right],
\end{equation*}
where $\{0\xlongleftrightarrow{BWBW}\partial B_1(0)\}$ denotes the event that there are four paths with alternating labels (black and white) connecting $0$ to $\partial B_1(0)$.

\begin{theorem} \label{thm::pivotal}
	Suppose that $0<\theta^1<\theta^2<\theta^3<\theta^4<\pi$ and write $x_j=\exp(2\ii\theta^j)$ for $1\leq j\leq 4$ and $(x_j x_{j+1}):=\{\exp(2\ii\theta): \theta^j<\theta<\theta^{j+1}\}$, with the convention that $\theta^5=\theta^1$. Let $\Omega^a=B_1(0)\cap a\mathcal{T}$ and assume that $x_1^a,x_2^a,x_3^a,x_4^a\in \partial \Omega^a$ satisfy $x_j^a\to x_j$ as $a\to 0$ for $1\leq j\leq 4$. Then there exists a universal constant $C_5\in (0,\infty)$ such that
	\begin{equation*}
\lim_{a\to 0} \orho_a^{-1}\times \mathbb{P}^a\left[ 0\text{ is pivotal for }\{(x_1^a x_2^a)\xlongleftrightarrow[\Omega^a]{B} (x_3^ax_4^a)\} \right]=C_5\prod_{1\leq j<k\leq 4}	|\sin(\theta^k-\theta^j)|^{\frac{1}{3}},
	\end{equation*}
where $\{(x_1^a x_2^a)\xlongleftrightarrow[\Omega^a]{B} (x_3^ax_4^a)\}$ denotes the event that there exists a black path contained in $\Omega^a$ and connecting $(x_1^ax_2^a)$ to $(x_3^ax_4^a)$. 
\end{theorem}

\begin{remark}
    Conformal covariance implies that Theorem~\ref{thm::pivotal} can be easily extended to more general domains and sets of points.
\end{remark}

{We now turn to percolation on the upper half-plane $\mathbb{H}:=\{z \in \mathbb{C}:\Im z>0\}$, which can be realized by deterministically declaring white all vertices contained in the lower half-plane $\mathbb{C} \setminus \overline{\mathbb{H}}$.
Our first result in this context is an analog of Theorem~\ref{thm::expansion} and concerns the lattice field \eqref{def:lattice-field} defined on the closed upper half-plane $\overline{\mathbb{H}}$, that is, with $S_{x^a}=0$ for all $x^a \in \mathbb{C} \setminus \overline{\mathbb{H}}$.
With this definition, if $x_1^a<x_2^a<x_3^a<x_4^a \in a\mathcal{T}\cap\mathbb{R}$, we have 
\begin{align*}
	& \langle S_{x_1^a}\ldots S_{x_4^a}\rangle^a_{\mathbb{H}}=\mathbb{P}^a\left[x_1^a \xlongleftrightarrow[\mathbb{H}]{B} x_2^a \xlongleftrightarrow[\mathbb{H}]{B} x_3^a \xlongleftrightarrow[\mathbb{H}]{B} x_4^a\right]+\mathbb{P}^a\left[x_1^a \xlongleftrightarrow[\mathbb{H}]{B} x_2^a \centernot{\xlongleftrightarrow[\mathbb{H}]{B}} x_3^a \xlongleftrightarrow[\mathbb{H}]{B} x_4^a\right] \\
	& \qquad \qquad \qquad \quad + \mathbb{P}^a\left[x_1^a \xlongleftrightarrow[\mathbb{H}]{B} x_4^a \centernot{\xlongleftrightarrow[\mathbb{H}]{B}} x_2^a \xlongleftrightarrow[\mathbb{H}]{B} x_3^a\right],
\end{align*}
where the subscript $\mathbb{H}$ in $\langle\cdot\rangle^a_{\mathbb{H}}$ indicates that $S_{x^a}=0$ for all $x^a \in \mathbb{C} \setminus \overline{\mathbb{H}}$, $\{x^a_i \xlongleftrightarrow[\mathbb{H}]{B} x^a_j \}$ denotes the event that $x^a_i$ and $x^a_j$ are connected by a black path contained in the upper half-plane and $\{x^a_i \centernot{\xlongleftrightarrow[\mathbb{H}]{B}} x^a_j \}$ denotes its complement.
Now let
\begin{equation} \label{def::opi}
	\opi_a:=\mathbb{P}^a\left[0\xlongleftrightarrow[\mathbb{H}]{B}\partial B_{1}(0)\right],
\end{equation}
where $\{0\xlongleftrightarrow{B}\partial B_{1}(0)\}$ denotes a \emph{boundary one-arm event}, i.e., the event that there exists a black path in $\mathbb{H}$ connecting $0$ to $\partial B_{1}(0)$.

\begin{theorem} \label{thm::expansion-C_H}
	Consider $x_1<x_2<x_3<x_4\in \mathbb{R}$ and assume that $x_1^a<x_2^a<x_3^a<x_4^a \in a\mathcal{T}\cap\mathbb{R}$ satisfy $x_j^a\to x_j$ as $a\to 0$ for $1\leq j\leq 4$. Then
	\begin{equation*}
		C_{\mathbb{H}}(x_1,x_2,x_3,x_4) := \lim_{a\to 0}\opi_a^{-4}\times \langle S_{x_1^a}\cdots S_{x_4^a}\rangle^a_{\mathbb{H}}
	\end{equation*}
exists and belongs to $(0,\infty)$.
Moreover, for any non-constant M\"obius transformation $\varphi: \overline{\mathbb{H}}\to \overline{\mathbb{H}}$ with $\varphi(x_1),\varphi(x_2), \varphi(x_3), \varphi(x_4) \neq \infty$, we have 
\begin{equation} \label{eqn::conf-cov-C_H}
	C_{\mathbb{H}}(\varphi(x_1),\varphi(x_2),\varphi(x_3),\varphi(x_4)) = C_{\mathbb{H}}(x_1,x_2,x_3,x_4)\times \prod_{j=1}^4 |\varphi'(x_j)|^{-\frac{1}{3}}. 
\end{equation}
	Furthermore, there are two universal constants $C^{\mathbb{H}}_1,C^{\mathbb{H}}_2\in(0,\infty)$ such that
	\begin{equation} \label{eqn::C_H-expansion}
		C_{\mathbb{H}}(x_1,x_2,x_3,x_4)\sim C^{\mathbb{H}}_1(x_2-x_1)^{-\frac{2}{3}}\left((x_4-x_3)^{-\frac{2}{3}}+ C^{\mathbb{H}}_2 (x_2-x_1)^{2} F_{\mathbb{H}}(x,x_3,x_4) \left|\log(x_2-x_1)\right|\right),
	\end{equation}
as $x_1,x_2\to x<x_3$, namely,
\begin{equation*}
	\lim_{x_1,x_2\to x} \frac{C_{\mathbb{H}}(x_1,x_2,x_3,x_4)-C^{\mathbb{H}}_1(x_4-x_3)^{-\frac{2}{3}}\times (x_2-x_1)^{-\frac{2}{3}}}{(x_2-x_1)^{\frac{4}{3}}\times \left| \log(x_2-x_1)\right|}= C^{\mathbb{H}}_1C^{\mathbb{H}}_2F_{\mathbb{H}}(x,x_3,x_4),
\end{equation*}
where 
\begin{equation*} \label{def:F}
	F_{\mathbb{H}}(x,x_3,x_4)=(x_3-x)^{-2}(x_4-x)^{-2}(x_4-x_3)^{\frac{4}{3}}.
\end{equation*}
\end{theorem}

The proof of Theorem~\ref{thm::expansion-C_H} is essentially the same as that of Theorem~\ref{thm::expansion}, so we omit it in the present article. Eq.~\eqref{eqn::conf-cov-C_H} shows that $C_{\mathbb{H}}$ transforms covariantly under conformal maps, as expected of a CFT correlation function.
When seen as a correlation function, \eqref{eqn::conf-cov-C_H} shows that the scaling dimension of the boundary spin (density) field is $1/3$.

As in the case of \eqref{eqn::C_expansion}, Eq.~\eqref{eqn::C_H-expansion} can be interpreted as an OPE.
As such, it reveals that the fusion of two boundary spin fields produces the ``unit field'' and a new boundary field of scaling dimension $2$.
The latter is related to the \emph{polychromatic boundary three-arm event} $\{0\xlongleftrightarrow[\mathbb{H}]{BWB} \partial B_1(0)\}$, corresponding to the presence of two black paths and one white path in $\mathbb{H}$ connecting $0$ and its neighbors to $\partial B_1(0)$, which is relevant for our next theorem.
We remark that the polychromatic boundary three-arm event is related to the so-called \emph{boundary stress-energy tensor}, as explained in Sections~2.2 and~3.3 of~\cite{CamiaFeng2024logarithmic}.

Now let
	\begin{equation*}
		\iota_a:=\mathbb{P}^a\left[0\xlongleftrightarrow[\mathbb{H}]{BW} \partial B_1(0)\right],
	\end{equation*}
where $\{0\xlongleftrightarrow[\mathbb{H}]{BW} \partial B_1(0)\}$ denotes the event that there are a black path and a white path in $\mathbb{H}$ connecting $0$ to $\partial B_1(0)$.

Given four vertices $x_1^a<x_2^a<x_3^a<x_4^a$ of $a\mathcal{T}\cap \mathbb{R}$, we let $\mathcal{K}^a(x_1^a,x_2^a,x_3^a,x_4^a)$ denote the event that there are a black path and a white path in the upper half-plane connecting $x_1^a$ to $x_2^a$, with the black path ``below'' the white one, and the same for $x_3^a$ and $x_4^a$ (see Figure~\ref{fig::K4} for an illustration of the event).
		\begin{figure} 
			\begin{center}
				\includegraphics[width=0.7\textwidth]{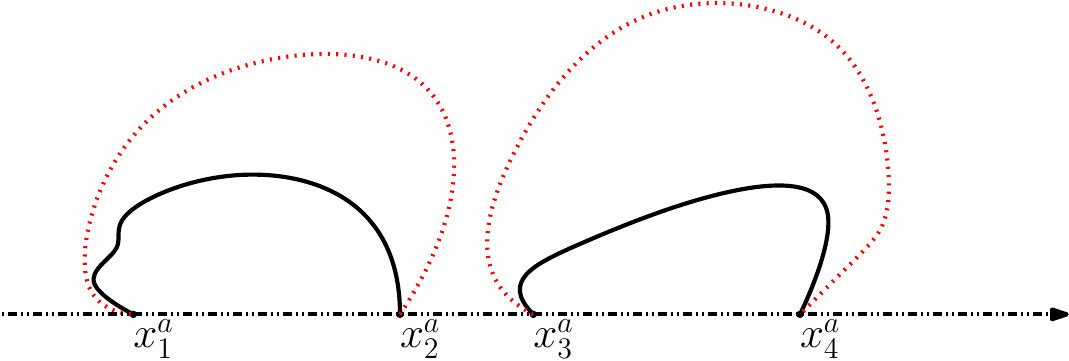}
			\end{center}
			\caption{The event ${\mathcal{K}}^a(x_1^a,x_2^a,x_3^a,x_4^a)$. The black, solid lines represent black paths and the red, dotted lines represent white paths.}
			\label{fig::K4}
		\end{figure}
	
Using the same strategy as in the proof of Theorem~\ref{thm::4arm_expansion}, one can show:
	\begin{theorem} \label{thm::three_arm}
		Let $-\infty<x_1<x_2<x_3<x_4<\infty$ be four distinct points. Suppose that $x_1^a<x_2^a<x_3^a<x_4^a\in a\mathcal{T}\cap \mathbb{R}$ satisfy $x_j^a\to x_j$ as $a\to 0$ for $1\leq j\leq 4$. Then the limit
		\begin{equation*}
			K(x_1,x_2,x_3,x_4):=\lim_{a\to 0} \iota_a^{-4}\times \mathbb{P}^a\left[\mathcal{K}^a(x_1^a,x_2^a,x_3^a,x_4^a)\right]
		\end{equation*}
		exists and belongs to $(0,\infty)$ and, for any non-constant M\"obius transformation $\varphi:\overline{\mathbb{H}}\to \overline{\mathbb{H}}$, we have
		\begin{equation*}
			K(\varphi(x_1),\varphi(x_2),\varphi(x_3),\varphi(x_4))=K(x_1,x_2,x_3,x_4)\times \prod_{j=1}^4 |\varphi'(x_j)|^{-1}. 
		\end{equation*}
		Moreover, there exists a universal constant $C_K\in (0,\infty)$ such that, for any $x\in (x_1,x_4)$, we have
		\begin{equation} \label{eqn::lim-K}
			\lim_{x_{2},x_3\to x} K(x_1,x_2,x_3,x_4)=C_K (x-x_1)^{-2}(x_4-x)^{-2}. 
		\end{equation}
	\end{theorem}

Eq.~\eqref{eqn::lim-K} can be seen as an OPE and the CFT interpretation is that two boundary two-arm fields/events, with scaling dimension $1$, fuse into a boundary three-arm field/event, with scaling dimension $2$, turning the four-point function $K(x_1,x_2,x_3,x_4)$ into a three-point function between two boundary two-arm fields/events, at $x_1$ and $x_4$, and a boundary three-arm field/event, at $x$.}

After discussing the boundary two-arm event, we consider the \emph{interior monochromatic two-arm event} $\{0\xlongleftrightarrow{BB}\partial B_1(0)\}$, i.e., the event that there are two disjoint black paths connecting $0$ to $\partial B_1(0)$, and the related concept of \emph{percolation backbone}.

Let $x_1^a, x_2^a\in a\mathcal{T}\cap \mathbb{R} $ be two vertices on the real line and let $z^a\in a\mathcal{T}\cap \mathbb{H}$ be a vertex in $\mathbb{H}$.
We denote by $\{x_1^a\xlongleftrightarrow[\mathbb{H}]{B}z^a\}\circ \{x_2^a\xlongleftrightarrow[\mathbb{H}]{B}z^a\}$ the event that there are two disjoint black paths in the upper half-plane connecting $z^a$ to $x_1^a$ and $z^a$ to $x_2^a$, respectively.
If the event happens, we say that $z^a$ belongs to the \emph{backbone} connecting $x_1^a$ and $x_2^a$.
In other words, the backbone connecting $x_1^a$ and $x_2^a$ is the set of black vertices in the upper half-plane connected to $x_1^a$ and $x_2^a$ by two black paths that have no vertex in common.
In an electrical circuit in which current can flow only through black vertices, the backbone connecting $x_1^a$ and $x_2^a$ is the set of vertices through which the current would flow if we applied a potential difference between $x_1^a$ and $x_2^a$.
The backbone is relevant to transport properties and has been extensively studied (see, e.g.,~\cite{bunde2012fractals} 
and references therein, as well as~\cite{stauffer1994introduction,sahimi1994applications,GRASSBERGER1999251}).

Let 
\begin{equation*}
	\rho_a:=\mathbb{P}^a\left[0\xlongleftrightarrow{BB}\partial B_1(0)\right]
\end{equation*}
and note that the event $\{x_1^a\xlongleftrightarrow[\mathbb{H}]{B}z^a\}\circ \{x_2^a\xlongleftrightarrow[\mathbb{H}]{B}z^a\}$ forces two boundary one-arm events near $x_1^a$ and $x_2^a$, $\{x^a_i\xlongleftrightarrow[\mathbb{H}]{B} \partial B_{\epsilon}(x^a_i)\}$ with $i=1,2$, as well as an interior monochromatic two-arm event near $z^a$, $\{z^a\xlongleftrightarrow{BB} \partial B_{\epsilon}(z^a)\}$, for any $\epsilon<\min(\vert z^a-x^a_1 \vert, \vert z^a-x^a_2 \vert)$.
The probabilities of both arm events decay like a power of $a$ as $a \to 0$, but while the exponent for the boundary one-arm event was conjectured decades ago~\cite{den1979kadanoff,cardy1983transitions}
and was proved to be $1/3$ more than twenty years ago~\cite{SmirnovWernerCriticalExponents}, the exponent governing the decay of $\rho_a$ remained unknown until very recently and was computed for the first time in~\cite{nolin2023backbone}, using Liouville quantum gravity techniques. The latter exponent is shown to be transcendental in~\cite[Theorem~1.2]{nolin2023backbone}, and we denote it by $\xi$.

\begin{theorem} \label{thm::backbone}
	Let $x_1,x_2\in \mathbb{R}$ and $z\in \mathbb{H}$ be three distinct points. Suppose that $x_1^a,x_2^a \in a\mathcal{T}\cap\mathbb{R}, z^a\in a\mathcal{T}\cap \mathbb{H}$ are vertices satisfying $x_1^a\to x_1$, $x_2^a\to x_2$ and $z^a\to z$ as $a\to 0$. Then 
	\begin{equation*}
		P(x_1,x_2,z):=\lim_{a\to 0} \opi_a^{-2}\rho_a^{-1}\times \mathbb{P}^a \left[\{x_1^a\xlongleftrightarrow[\mathbb{H}]{B}z^a\}\circ \{x_2^a\xlongleftrightarrow[\mathbb{H}]{B}z^a\}\right]
	\end{equation*}
exists and belongs to $(0,\infty)$.
Moreover, for any non-constant M\"obius transformation $\varphi: \overline{\mathbb{H}}\to \overline{\mathbb{H}}$ with $\varphi(x_1),\varphi(x_2)\neq \infty$, we have 
\begin{equation*}
	P(\varphi(x_1),\varphi(x_2),\varphi(z))=P(x_1,x_2,z)\times |\varphi'(z)|^{-\xi}\times\prod_{j=1}^2 |\varphi'(x_j)|^{-\frac{1}{3}}. 
\end{equation*}
\end{theorem}

Our last theorem concerns an event on the upper half-plane whose probability, when appropriately rescaled, has a scaling limit involving a logarithm.
The limit can be interpreted as related to a CFT correlation function between four boundary fields, so this result provides another example of a logarithmic correlation function.
For a more detailed discussion on the CFT interpretation of the result, the reader is referred to~\cite{CamiaFeng2024logarithmic}.

Let $x_1^a<x_2^a<x_3^a<x_4^a$ be four vertices on $a\mathcal{T}\cap \mathbb{R}$. We let $\mathcal{L}^a(x_1^a,x_2^a,x_3^a,x_4^a)$ denote the following event: (1) {there are two disjoint white paths in the upper half-plane connecting a neighbor of $x_1^a$ and $x_2^a+a$ to the segment $[x_3^a,x_4^a]$, respectively}; (2) there is a black path in the upper half-plane connecting $x_1^a$ to $x_2^a$; (3) there is no black path in the upper half-plane connecting $x_1^a$ to the segment $[x_3^a,x_4^a]$. See Figure~\ref{fig::L4} for an illustration of the event $\mathcal{L}^a(x_1^a,x_2^a,x_3^a,x_4^a)$.

\begin{figure} 
    \begin{center}
		\includegraphics[width=0.7\textwidth]{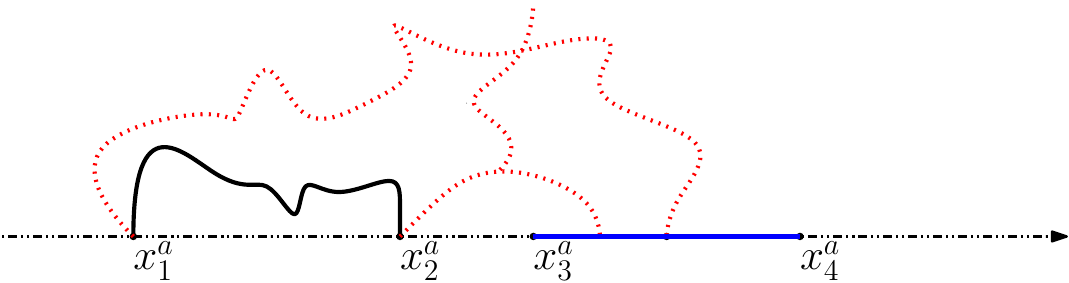}
	\end{center}
	\caption{The event $\mathcal{L}^a(x_1^a,x_2^a,x_3^a,x_4^a)$. The black, solid line represents a black path, and the red, dotted lines represent white paths.}
	\label{fig::L4}
\end{figure}

\begin{theorem} \label{thm::log2}
 Let $-\infty<x_1<x_2<x_3<x_4<\infty$ be four distinct points. Suppose that $x_1^a<x_2^a<x_3^a<x_4^a\in a\mathcal{T}\cap \mathbb{R}$ satisfy $x_j^a\to x_j$ as $a\to 0$ for $1\leq j\leq 4$. Then there exists a universal constant $C_L\in (0,\infty)$ such that 
\begin{equation}\label{eqn::log_formula2}
	L(x_1,x_2,x_3,x_4)=\lim_{a\to 0} \iota_a^{-2}\times \mathbb{P}^a\left[\mathcal{L}^a(x_1^a,x_2^a,x_3^a,x_4^a)\right]=\frac{C_L}{(x_2-x_1)^2}\times \log \frac{(x_4-x_2)(x_3-x_1)}{(x_3-x_2)(x_4-x_1)}. 
\end{equation}
	In particular, we have that
\begin{equation} \label{eqn::log_asy2}
	\lim_{x_3\to x_2}  \frac{L(x_1,x_2,x_3,x_4)}{\left| \log(x_3-x_2) \right|}= \frac{C_L}{(x_2-x_1)^2}
\end{equation}
and, for any non-constant M\"obius transformation $\varphi:\mathbb{H}\to \mathbb{H}$ satisfying $-\infty<\varphi(x_1)<\varphi(x_2)<\varphi(x_3)<\varphi(x_4)<\infty$,
\begin{equation} \label{eqn::log_cov2}
	L(\varphi(x_1),\varphi(x_2),\varphi(x_3),\varphi(x_4))=  L(x_1,x_2,x_3,x_4) \times \prod_{j=1}^2 |\varphi'(x_j)|^{-1}.
	\end{equation}
\end{theorem}

	\begin{remark} \label{rem::sharp_arm}
	It was shown in~\cite{SmirnovWernerCriticalExponents,nolin2023backbone} that 
	\begin{equation*}
		\opi_a=a^{\frac{1}{3}+o(1)},\quad \pi_a=a^{\frac{5}{48}+o(1)},\quad \iota_a=a^{1+o(1)},\quad \rho_a=a^{\xi+o(1)},\quad \orho_a=a^{\frac{5}{4}+o(1)}.
	\end{equation*}
	Recently, the estimates for $\opi_a,\iota_a$ and $\orho_a$ were improved in~\cite{SharpArmExponent} as follows: there exist universal constants $\overline{c}_{\rho},\overline{c}_{\pi}, c_{\iota},\overline{C}_{\pi}, C_{\iota}\in (0,\infty)$ such that 
	\begin{equation*}
		\opi_a=\overline{c}_{\pi} a^{\frac{1}{3}}\left(1+O\big(a^{\overline{C}_{\pi}}\big)\right),\quad \iota_a=c_{\iota}a\left(1+O\big(a^{C_{\iota}}\big)\right),\quad  \orho_a=\overline{c}_{\rho}a^{\frac{5}{4}}\left(1+o(1)\right).
	\end{equation*}
\end{remark}

\subsection{Organization of the rest of the paper and discussion of a logarithmic singularity}\label{subsec::explain_proof}
In Section~\ref{sec::2}, we consider the density of pivotal points and the probability of backbone events and prove Theorems~\ref{thm::pivotal} and~\ref{thm::backbone}, using results and ideas from~\cite{Cam23}. In Section~\ref{sec::3}, we first study the four-point function of the density (spin) field \eqref{def:lattice-field} and prove Theorem~\ref{thm::expansion}. We then consider correlations and the fusion of four-arm events and prove Theorem~\ref{thm::4arm_expansion}. Section~\ref{sec::3} ends with the proofs of Theorems~\ref{thm::three_arm} and~\ref{thm::log2} concerning, respectively, the four-point function of boundary two-arm events and a boundary connection probability with a logarithmic scaling limit.

To conclude this section, we sketch the main arguments of the proof of the logarithmic correction in the four-point function of the density field (Theorem~\ref{thm::expansion}), which represents the core of this article.
We do this for the reader's convenience and because the arguments themselves are of independent interest since they explain the physical mechanism that leads to the logarithmic singularity. As explained in Section~\ref{subsec::def_result}, the four-point function $C(x_1,x_2,x_3,x_4)$ is a linear combination of several connection probabilities. The key is to show that 
	\begin{align}
		\begin{split} \label{eqn::explain_log_aux1}
			&P(x_1\xlongleftrightarrow{B}x_2, x_3\xlongleftrightarrow{B}x_4)-P(x_1\xlongleftrightarrow{B}x_2)P(x_3\xlongleftrightarrow{B}x_4)\\
			&	\qquad=C_1C_2|x_2-x_1|^{\frac{5}{4}-\frac{5}{24}}F(x,x_3,x_4)\vert\log\vert x_2-x_1\vert\vert+o\big(|x_2-x_1|^{\frac{5}{4}-\frac{5}{24}}\vert\log\vert x_2-x_1\vert\vert\big),\quad \text{as }x_1,x_2\to x,
		\end{split}
	\end{align}
	where
	\begin{align*}
		P(x_1\xlongleftrightarrow{B}x_2, x_3\xlongleftrightarrow{B}x_4):=& P(x_1\xlongleftrightarrow{B}x_2\xlongleftrightarrow{B}x_3\xlongleftrightarrow{B}x_4)+P(x_1\xlongleftrightarrow{B}x_2\centernot{\xlongleftrightarrow{B}}x_3\xlongleftrightarrow{B} x_4).
	\end{align*}
	It is not hard to show that, as $x_1,x_2\to x$, the difference in the first line of~\eqref{eqn::explain_log_aux1} decays to $0$ at most polynomially in $|x_2-x_1|$. The more challenging part is to figure out the exact speed of this decay and identify the logarithmic correction, as we briefly explain below (see Figure~\ref{Figure}).
	
	Given two subsets of the plane, $C$ and $D$,  we consider the following events:
	\begin{itemize}
		\item $\{x_1\xlongleftrightarrow{B;C}x_2\}$: there is a black path connecting $x_1$ to $x_2$ contained in $C$;
		\item $\{x_1\xlongleftrightarrow[D]{B}x_2\}$: $x_1$ and $x_2$ belong to the same black cluster but there is no black path fully contained in $D$;
		\item $\{x_1\xlongleftrightarrow[D]{B;C}x_2\}$: there is a black path connecting $x_1$ to $x_2$ contained in $C$ but no black path fully contained in $D$. 
	\end{itemize}
	Now consider disks $B_m=\{z: |z-\frac{x_1+x_2}{2}|\leq 2^m |x_2-x_1|\}$ for $m=1,\ldots, M$, where $M$ is chosen so that $2^M\sim 1/|x_2-x_1|$,  that is, $M\sim -\log| x_2-x_1|$, and so that $x_3$ and $x_4$ are outside $B_M$. Then using the independence of labels at different vertices in percolation, one can show that
	\begin{align*}
		& P(x_1\xlongleftrightarrow{B}x_2, x_3\xlongleftrightarrow{B}x_4)-P(x_1\xlongleftrightarrow{B}x_2)P(x_3\xlongleftrightarrow{B}x_4)\\
		& \quad = \big[P(x_3\xlongleftrightarrow[B_1^c]{B}x_4 |x_1\xlongleftrightarrow{B;B_1}x_2)-P(x_3\xlongleftrightarrow[B_1^c]{B}x_4 )\big]P(x_1\xlongleftrightarrow{B;B_1}x_2)\\
		& \qquad + \sum_{m=2}^M \big[P(x_1\xlongleftrightarrow[B_{m-1}]{B}x_2,x_3\xlongleftrightarrow[B_m^c]{B} x_4| x_1\xlongleftrightarrow{B;B_m}x_2)-P(x_1\xlongleftrightarrow[B_{m-1}]{B}x_2| x_1\xlongleftrightarrow{B;B_m}x_2)P(x_3\xlongleftrightarrow[B_m^c]{B} x_4)\big]\\
		& \qquad\qquad\qquad\qquad\times P(x_1\xlongleftrightarrow{B;B_m}x_2)\\
		& \qquad + \big[P(x_1\xlongleftrightarrow[B_M]{B}x_2,x_3\xlongleftrightarrow{B}x_4| x_1\xlongleftrightarrow{B}x_2)-P(x_1\xlongleftrightarrow[B_M]{B}x_2| x_1\xlongleftrightarrow{B}x_2)P(x_3\xlongleftrightarrow{B}x_4)\big] \\
        & \qquad\qquad\qquad\qquad\times P(x_1\xlongleftrightarrow{B}x_2). 
	\end{align*}

	For $m=2,\ldots, M$, on the one hand, the event $\{x_1\xlongleftrightarrow[B_{m-1}]{B}x_2\}$ implies that the annulus $B_{m-1}\setminus B_1$ is crossed by two black paths and two white paths.  Since the four-arm exponent equals $5/4$~\cite{SmirnovWernerCriticalExponents}, we then conclude that $P(x_1\xlongleftrightarrow[B_{m-1}]{B}x_2|x_1\xlongleftrightarrow{B;B_m}x_2)\sim \big((1/2)^{m-2}\big)^{5/4}$. On the other hand, the event $\{x_3\xlongleftrightarrow[B_{m}^c]{B}x_4\}$ implies that the annulus $B_{M}\setminus B_m$ is crossed by two black paths and two white paths. We call this event $\mathcal{F}(x_3,x_4;B_{m}^c)$ and note, using again the four-arm exponent, that its probability is of order $(2^m|x_2-x_1|)^{5/4}$. Consequently, we can write 
	\begin{align*}
		&\big[P(x_1\xlongleftrightarrow[B_{m-1}]{B}x_2,x_3\xlongleftrightarrow[B_m^c]{B} x_4| x_1\xlongleftrightarrow{B;B_m}x_2)-P(x_1\xlongleftrightarrow[B_{m-1}]{B}x_2| x_1\xlongleftrightarrow{B;B_m}x_2)P(x_3\xlongleftrightarrow[B_m^c]{B} x_4)\big]\\
		&\qquad\qquad\qquad\qquad\times P(x_1\xlongleftrightarrow{B;B_m}x_2)\\
		&\qquad\sim g_{m}(x_1,x_2,x_3,x_4)|x_2-x_1|^{-\frac{5}{24}}|x_2-x_1|^{\frac{5}{4}},
	\end{align*}
	where 
	\begin{align*}
		g_m(x_1,x_2,x_3,x_4):=P\big(x_3\xlongleftrightarrow[B_m^c]{B}x_4| x_1\xlongleftrightarrow[B_{m-1}]{B;B_m}x_2,\mathcal{F}(x_3,x_4;B_m^c)\big)-P(x_3\xlongleftrightarrow[B_m^c]{B}x_4 |\mathcal{F}(x_3,x_4;B_m^{c})).
	\end{align*}
	Roughly speaking, thanks to the positive association of percolation (FKG inequality), the black path connecting $x_1$ to $x_2$ in the event $\{x_1\xlongleftrightarrow[B_{m-1}]{B;B_m}x_2\}$ ``helps" the connectivity event $\{x_3\xlongleftrightarrow[B_m^c]{B}x_4\}$ to occur, which implies that $g_m(x_1,x_2,x_3,x_4)\geq 0$. 
 In Section~\ref{sec::tech_log}, maybe the most intricate part of the proof of Theorem~\ref{thm::expansion}, we show that $g_m$ is bounded away from zero uniformly in $m$ and in $|x_2-x_1|$. Consequently, summing over $m$ from $2$ to $M$ (recall that $M\sim -\log |x_2-x_1|$) gives the logarithmic correction in~\eqref{eqn::explain_log_aux1}.

 \begin{figure}
	\includegraphics[width= 0.5\textwidth]{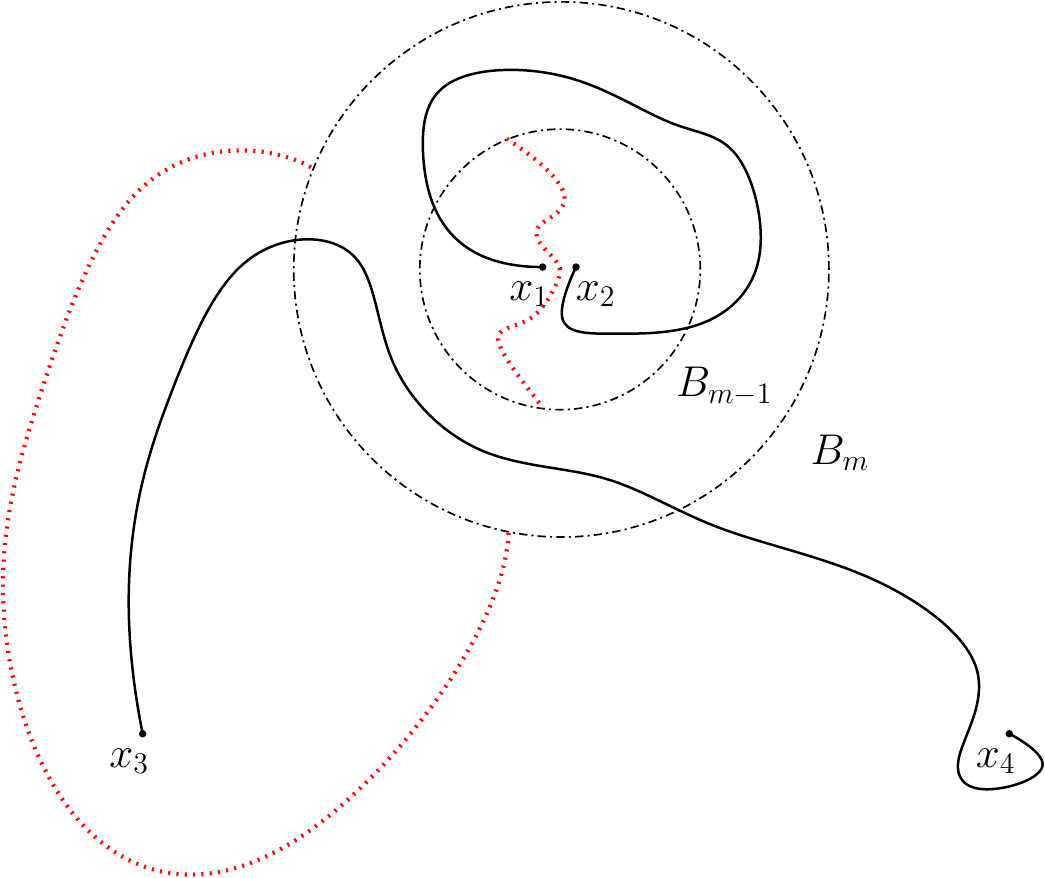}
	\caption{The event $\{ x_1 \xleftrightarrow[B_{m-1}]{B;B_m} x_2, x_3 \xleftrightarrow[B_m^c]{B} x_4 \}$. The black, solid lines denote black paths, while the red, dotted lines denote white paths. $x_1$ and $x_2$ are contained in $B_{m-1}$. They are not connected by a black path within the disk $B_{m-1}$, but are connected within the larger disk $B_{m}$, with radius twice that of $B_{m-1}$. $x_3$ and $x_4$ are connected by a black path, but not outside $B_m$.  The number, $M$, of disks one can insert between the two groups of points $\{x_1,x_2\}$ and $\{x_3,x_4\}$ is of order $-\log{\vert x_1-x_2 \vert}$.}
	\label{Figure}
\end{figure}
	
To summarize, as $x_1,x_2\to 0$, the two events $\{x_1\xlongleftrightarrow{B}x_2\}$ and $\{x_3\xlongleftrightarrow{B}x_4\}$ become asymptotically independent; however, the weak ``interaction" between these two events produces a logarithmic correction.
More precisely, the four-point function $C(x_1,x_2,x_3,x_4)$ contains terms that correspond to the probabilities of events of the following type (see Figure~\ref{Figure}):
given disks $B_m=\{z: |z-\frac{x_1+x_2}{2}|\leq 2^m |x_2-x_1|\}$,
\begin{itemize}
\item there is an open path between $x_1$ and $x_2$ contained inside $B_{m}$ but not inside $B_{m-1}$,
\item there is an open path between $x_3$ and $x_4$ contained inside $B^c_{m-1}$ but not inside $B^c_m$.
\end{itemize}
Due to scale invariance, the probabilities of events of this type are of the same order for different values of $m$, ranging from $1$ to $M\sim -\log| x_2-x_1|$.
The logarithm in~\eqref{eqn::explain_log_aux1} corresponds to the number of annuli $B_{m+1} \setminus B_m$ contained in the space between $x_1,x_2$ and $x_3,x_4$, which is of order $\log\frac{1}{\vert x_1-x_2 \vert}$.

It would be interesting to explore if a similar mechanism is at work in higher dimensions and in other models, and if our methods can be applied to higher dimensions and to models such as FK percolation and the $O(n)$ loop model.

\paragraph*{Acknowledgments.} Yu Feng thanks NYUAD for its hospitality during a visit in the fall of 2023, when this project was started.
The visit was partially supported by the Short-Term Visiting Fund for Doctoral Students of Tsinghua University.

\section{Conformal covariance of pivotal and backbone probabilities} \label{sec::2}

In this section, we prove Theorems~\ref{thm::pivotal} and Theorem~\ref{thm::backbone}.

\subsection{Percolation interfaces and their scaling limit}
We briefly recall the definition of percolation interfaces, which are curves separating black and white clusters. The scaling limit of the full collection of interfaces for critical site percolation on the triangular lattice is obtained in~\cite{CamiaNewmanPercolationFull}.

We let $a>0$ and consider critical Bernoulli site percolation on $a\mathcal{T}$. Given a percolation configuration,  the percolation interfaces between black and white clusters are polygonal circuits (with probability one) on the edges of the hexagonal lattice $a\mathcal{H}$ dual to the triangular lattice $a\mathcal{T}$. We give these circuits an orientation in such a way that they wind counterclockwise around black clusters and clockwise around white clusters (in other words, they are oriented in such a way that black hexagons are on the left and white hexagons are on the right). Note that the interfaces form a nested collection of loops with alternating orientations and a natural tree structure. 

In order to state the weak convergence of the collection of percolation interfaces, we need to specify a topology on the space of collections of loops. First, we introduce a distance function $\Delta$ on $\mathbb{C}\times \mathbb{C}$,
\begin{equation*}
	\Delta(u,v):=\inf_{f} \int_0^1 \frac{|f'(t)|}{1+|f(t)|^2}\ud t,
\end{equation*}
where the infimum is over all differentiable curves $f: [0,1]\to \mathbb{C}$ with $f(0)=0$ and $f(1)=v$. 
Second, for two planar oriented curves $\gamma_1,\gamma_2:[0,1]\to \mathbb{C}$, we define
\begin{equation}  \label{eqn::curve_metric}
	\dist\left(\gamma_1,\gamma_2\right):=\inf_{\psi,\tilde{\psi}}\sup_{t\in[0,1]} \Delta\left( \gamma_1(\psi(t)), \gamma_2(\tilde{\psi}(t))\right),
\end{equation}
where the infimum is taken over all increasing homeomorphisms $\psi,\tilde{\psi}:[0,1]\to [0,1]$. Note that planar oriented loops can be viewed as planar oriented curves. Third, we define a distance between two closed sets of loops, $\Gamma_1$ and $\Gamma_2$, as follows:
\begin{equation} \label{eqn::loop_metric}
	\mathrm{Dist}(\Gamma_1,\Gamma_2):=\inf\{\epsilon>0: \forall \gamma_1\in \Gamma_1 \enspace \exists \gamma_2\in \Gamma_2 \text{ s.t. }\dist(\gamma_1,\gamma_2)\leq \epsilon \text{ and vice versa}\}. 
\end{equation}
The space $X$ of collections of loops with this distance is a separable metric space.

It was shown in~\cite{CamiaNewmanPercolationFull} that, as $a\to 0$, the collection of percolation interfaces has a unique limit in distribution in the topology induced by~\eqref{eqn::loop_metric}. We call this limit the \textit{full scaling limit of percolation} and let $\mathbb{P}$ denote its distribution. We let $\Lambda$ denote a loop configuration distributed according to $\mathbb{P}$.  As explained in~\cite{SLECLE6}, $\Lambda$ is distributed like the full-plane, nested $\mathrm{CLE}_6$. It is invariant, in a distributional sense, under all  non-constant M\"obius transformations~\cite{CamiaNewmanPercolationFull,GwynneMillerQianCLE}.

\subsection{The density of pivotal points: Proof of Theorem~\ref{thm::pivotal}} \label{sec::pivotal}
From now on, we denote by $\{\delta_m\}_{m=1}^{\infty}$ a decreasing sequence with $\delta_m\in (0,1)$ and $\lim_{m\to \infty}\delta_m=0$. 
For $0<\eta<\epsilon$ and $z\in \mathbb{C}$, we denote by $A_{\eta,\epsilon}(z)$ the annulus $B_{\epsilon}(z)\setminus B_{\eta}(z)$ and denote by $\mathcal{F}^a_{\eta,\epsilon}(z)$ the event that there are four paths with alternating labels, black and white, crossing the annulus $A_{\eta,\epsilon}(z)$. Furthermore, we denote by $\mathcal{F}^a_{\eta}(x_1,x_2, x_3, x_4)$ the event that there are four paths with alternating labels (BWBW) starting from the arcs $(x_1 x_2), (x_2 x_3), (x_3 x_4), (x_4 x_1)$, respectively, and crossing the annulus $A_{\eta,1}(0)$. 
One can express these two events in terms of percolation interface loops. More precisely, $\mathcal{F}_{\eta,\epsilon}^a(z)$ means that there are four distinct segments of interface loops, with alternating orientations, crossing the annulus $A_{\eta,\epsilon}(z)$. The situation for $\mathcal{F}^a_{\eta}(x_1,x_2,x_3,x_4)$ is similar. The event $\mathcal{F}^a_{\eta}(x_1,x_2,x_3,x_4)$ means that there are four distinct segments of interface loops, $S_1,\ldots,S_4$, with alternating orientations, crossing the annulus $A_{\eta,1}(0)$ in such a way that the following holds: for $j$ odd (respectively, even), let $H_j$ denote the black (resp., white) cluster immediately to the left of $S_j$, then $H_j\cap (x_j x_{j+1})\neq \emptyset$, where we use the convention $x_5:=x_1$.

Using the loop definition of $\mathcal{F}_{\eta,\epsilon}^a(z)$ and the percolation full scaling limit $\Lambda$ in terms of interface loops given in~\cite{CamiaNewmanPercolationFull}, we can define the analog of $\mathcal{F}_{\eta,\epsilon}^a(z)$ in the continuum, which we denote by $\mathcal{F}_{\eta,\epsilon}(z)$.
Now note that Theorems~1 and~11 and Lemma~9 of~\cite{MR4044275} imply that the collection of critical percolation clusters has a well-defined scaling limit, which is measurable with respect to the collection of interface loops $\Lambda$.
This allows us to define a continuum analog of $\mathcal{F}_{\eta}^a(x_1,x_2,x_3,x_4)$, which we denote by $\mathcal{F}_{\eta}(x_1,x_2,x_3,x_4)$.
Moreover, since the polychromatic boundary $3$-arm exponent for critical site percolation is strictly larger than $1$~\cite{SmirnovWernerCriticalExponents}, the events $\mathcal{F}_{\eta,\epsilon}(z)$ and $\mathcal{F}_{\eta}(x_1,\ldots,x_4)$ are continuity events for $\mathbb{P}$.
This follows from the fact that the boundary of $\mathcal{F}_{\eta,\epsilon}(z)$ (resp., $\mathcal{F}_{\eta}(x_1,x_2,x_3,x_4)$), defined using \eqref{eqn::loop_metric}, implies a boundary $3$-arm event along $\partial B_{\epsilon}(z)\cup \partial B_{\eta}(z)$ (resp., $\partial B_{1}(0)\cup \partial B_{\eta}(0)$).
According to Lemma~6.1 of~\cite{CamiaNewmanPercolationFull}, the latter has probability zero.
Weak convergence then implies the claim.
Therefore, we have 
\begin{equation} \label{eqn::continuity_four_arm}
\lim_{a\to 0}\mathbb{P}^a\left[\mathcal{F}_{\eta,\epsilon}^a(z)\right]=\mathbb{P}\left[\mathcal{F}_{\eta,\epsilon}(z)\right],\quad \text{and}\quad \lim_{a\to 0}\mathbb{P}^a\left[\mathcal{F}_{\eta}^a(x_1,x_2,x_3,x_4)\right]=\mathbb{P}\left[\mathcal{F}_{\eta}(x_1,x_2,x_3,x_4)\right].
\end{equation}

We write 
\begin{equation*}
	\mathcal{E}^a(x_1,x_2,x_3,x_4):=\{0\text{ is pivotal for }\{(x_1 x_2)\xlongleftrightarrow[\Omega^a]{B} (x_3 x_4)\}\}.
\end{equation*}
Note that the event $\mathcal{E}^a(x_1,x_2,x_3,x_4)$ forces a four-arm event surrounding $0$.
In the proof of Theorem~\ref{thm::pivotal}, we will use the following coupling result concerning four-arm events.

\begin{lemma} \label{lem::four_arm_coupling_inner}
		Consider $\epsilon>\delta>a$. There exists a universal constant $c_1\in (0,\infty)$ such that the following holds. For any $\delta>\eta>a$, there exist an event $\mathcal{S}$ and a coupling, $\mathbb{P}_{\eta}^a$, between $\mathbb{P}^a\left[\, \cdot \, |0\xlongleftrightarrow{BWBW}\partial B_{\epsilon}(0)\right]$ and $\mathbb{P}^a\left[\, \cdot \, | \mathcal{F}_{\eta,\epsilon}^a(0)\right]$ (that is, a joint distribution on pairs $(\tilde{\Lambda}^a,\hat{\Lambda}^a)$ such that $\tilde{\Lambda}^a$ and $\hat{\Lambda}^a$ are distributed according to $\mathbb{P}^a\left[\, \cdot \, |0\xlongleftrightarrow{BWBW}\partial B_{\epsilon}(0)\right]$ and $\mathbb{P}^a\left[\, \cdot \, | \mathcal{F}_{\eta,\epsilon}^a(0)\right]$, respectively) with the following properties:
	\begin{equation*}
\mathbb{P}_{\eta}^a\left[\mathcal{S}\right]\geq 1-\left(\frac{\eta}{\delta}\right)^{c_1},
	\end{equation*}
	and, for any event $\mathcal{A}$ that depends only on the states of hexagons of a single percolation configuration outside $B_{\delta}(0)$,
	\begin{equation*}
		\mathbb{P}_{\eta}^a\left[\tilde{\Lambda}^a\in \mathcal{A}|\mathcal{S}\right]=\mathbb{P}_{\eta}^a\left[\hat{\Lambda}^a\in \mathcal{A}|\mathcal{S}\right].
	\end{equation*} 
\end{lemma}
\begin{proof}
Lemma~\ref{lem::four_arm_coupling_inner} can be proved using the strategy in the proof of Proposition~3.6 of~\cite{GarbanPeteSchrammPivotalClusterInterfacePercolation}, with an exploration process that starts at the origin and moves outwards (see also~\cite[Lemma~3.9]{FWYIsingRadial} for a general result concerning a coupling of measures conditioned on an alternating $2N$-arm event). 
\end{proof}
The last ingredient we need for the proof of Theorem~\ref{thm::pivotal} is 
\begin{equation} \label{eqn::scal_four}
	\lim_{a\to 0} \frac{\mathbb{P}^a\left[0\xlongleftrightarrow{BWBW} \partial B_{\epsilon}(0)\right]}{\orho_a}=\epsilon^{-\frac{5}{4}},
\end{equation}
which follows from~\cite[Proposition~4.9]{GarbanPeteSchrammPivotalClusterInterfacePercolation} (or just by mimicking the proof of~\eqref{eqn::backbone_scal} in Section~\ref{sec::thm::backbone} below, with Lemma~\ref{lem::backbone} replaced by Lemma~\ref{lem::four_arm_coupling_inner}).

\begin{proof}[Proof of Theorem~\ref{thm::pivotal}]
	 We let $a<\eta<\delta_m<\epsilon$ and denote by $\mathcal{E}^{(a,\eta,\delta_m,\epsilon)}(x_1,x_2,x_3,x_4)$ the following event: (1) there are four clusters $\mathcal{C}_j$, for $1\leq j\leq 4$, with alternating labels, black and white, in counterclockwise order, connecting $\partial B_{\eta}(0)$ to $\partial B_{\epsilon}(0)$; (2) the cluster $\mathcal{C}_j$ is connected to $(x_j x_{j+1})$ by a black (resp., white) path inside $\mathbb{C}\setminus B_{\delta_m}(0)$ when $j$ is odd (resp., even).
   As in the case of the event $\mathcal{F}_{\eta}^a(x_1,x_2,x_3,x_4)$ discussed above, we can use \cite{CamiaNewmanPercolationFull} and~\cite{MR4044275} to define a continuum analog of the event $\mathcal{E}^{(a,\eta,\delta_m,\epsilon)}(x_1,x_2,x_3,x_4)$, which we denote by $\mathcal{E}^{(\eta,\delta_m,\epsilon)}(x_1,x_2,x_3,x_4)$. Since the polychromatic boundary $3$-arm exponent for critical site percolation is strictly larger than $1$~\cite{SmirnovWernerCriticalExponents}, the event $\mathcal{E}^{(\eta,\delta_m,\epsilon)}(x_1,x_2,x_3,x_4)$ is a continuity event for $\mathbb{P}$, {by the same argument used earlier to prove the continuity of $\mathcal{F}_{\eta,\epsilon}(z)$ and $\mathcal{F}_{\eta}(x_1,x_2,x_3,x_4)$.}

Next, we prove that 
\begin{equation}\label{eqn::existence_pivotal_proba_limit}
	 	\lim_{a\to 0} \overline{\rho}_a^{-1}\times \mathbb{P}^a\left[\mathcal{E}^a\left(x_1,x_2,x_3,x_4\right)\right]= \epsilon^{-\frac{5}{4}}\lim_{m\to \infty}\lim_{\eta\to 0}\mathbb{P}\left[\mathcal{E}^{(\eta,\delta_m,\epsilon)}(x_1,x_2,x_3,x_4)|\mathcal{F}_{\eta,\epsilon}(0)\right]\in (0,\infty).
	 \end{equation}
A standard application of RSW estimates (see, e.g., the proofs of Lemmas~2.1 and~2.2 of~\cite{CamiaNewman2009ising}) implies that there exist constants $0<K_1<K_2<\infty$, independent of $a$, such that 
	\begin{equation} \label{eqn::RSW_0}
		\overline{\rho}_a^{-1}\times \mathbb{P}^a\left[\mathcal{E}^a\left(x_1,x_2,x_3,x_4\right)\right]\in [K_1,K_2],
	\end{equation} 
which shows that all subsequential limits of the left-hand side of~\eqref{eqn::RSW_0} belong to $(0,\infty)$.  
We let $\epsilon\in \big(0,\frac{1}{100}\big)$ and 
write 
\begin{align*}
	 \overline{\rho}_a^{-1}\times \mathbb{P}^a\left[\mathcal{E}^a\left(x_1,x_2,x_3,x_4\right)\right]= \underbrace{\mathbb{P}^a\left[\mathcal{E}^a(x_1,x_2,x_3,x_4)\vert {0}\xlongleftrightarrow{BWBW}\partial B_{\epsilon}{(0)}\right]}_{T_1}\times \underbrace{\frac{\mathbb{P}^a\left[0\xlongleftrightarrow{BWBW}\partial B_{\epsilon}(0)\right]}{\overline{\rho}_a}}_{T_2}.
\end{align*}
For $T_2$, according to~\eqref{eqn::scal_four}, we have 
\begin{equation*}
	\lim_{a\to 0}T_2= \epsilon^{-\frac{5}{4}}.
\end{equation*}
For $T_1$, with~\cite[Lemma~2.1]{Cam23} replaced by Lemma~\ref{lem::four_arm_coupling_inner}, we can proceed as in~\cite[Proof of Theorem~1.1]{Cam23} to show that 
\begin{align*}
	\lim_{a\to 0} T_1 =\mathbb{P}\left[\mathcal{E}(x_1,x_2,x_3,x_4)|0\xlongleftrightarrow{BWBW}\partial B_{\epsilon}(0)\right] := \lim_{m\to \infty}\lim_{\eta\to 0}\mathbb{P}\left[\mathcal{E}^{(\eta,\delta_m,\epsilon)}(x_1,x_2,x_3,x_4)|\mathcal{F}_{\eta,\epsilon}(0)\right].
\end{align*}
Combining all of these observations, we get~\eqref{eqn::existence_pivotal_proba_limit}.

Now let us derive the desired explicit formula. 
We write 
\begin{align*}
	&\orho_a^{-1}\times \mathbb{P}^a\left[\mathcal{E}^a\left(x_1,x_2,x_3,x_4\right)\right] = \underbrace{\frac{\mathbb{P}^a\left[0\xlongleftrightarrow{BWBW}\partial B_{\delta_m}(0)\right]}{\orho_a}}_{T_3}\times \underbrace{\mathbb{P}^a\left[\mathcal{F}^a_{\delta_m}(x_1,x_2,x_3,x_4)\right]}_{T_4}\\
	& \qquad \qquad \qquad \times \underbrace{\frac{\mathbb{P}^a\left[\mathcal{E}^a(x_1,x_2,x_3,x_4)| \mathcal{F}^a_{\delta_m}\left(x_1,x_2,x_3,x_4\right)\right]}{\mathbb{P}^a\left[0\xlongleftrightarrow{BWBW} \partial B_{\delta_m}(0)\right]}}_{T_5}.
\end{align*}
According to~\eqref{eqn::scal_four}, we have
\begin{equation} \label{eqn::four_arm_aux1} 
	\lim_{a\to 0} T_3=\delta_{m}^{-\frac{5}{4}}.
\end{equation}
For $T_4$, according to~\eqref{eqn::continuity_four_arm}, we have
\begin{equation}\label{eqn::four_arm_aux2} 
	\lim_{a\to 0} T_4=\mathbb{P}\left[\mathcal{F}_{\delta_m}(x_1,x_2,x_3,x_4)\right].
\end{equation}
Combining~\eqref{eqn::four_arm_aux1} and \eqref{eqn::four_arm_aux2} with the existence of $\lim_{a\to 0}\orho_a^{-1}\mathbb{P}^a\left[\mathcal{E}^a\left(x_1,x_2,x_3,x_4\right)\right]$, we can define 
\begin{equation*}
	f(x_1,x_2,x_3,x_4;\delta_m):=\lim_{a\to 0} T_5.
\end{equation*}
Combining all these observations together, we obtain
\begin{equation} \label{eqn::four_aux3}
\lim_{a\to 0} \orho_a^{-1}\times \mathbb{P}^a\left[\mathcal{E}^a\left(x_1,x_2,x_3,x_4\right)\right]=\underbrace{\delta_m^{-\frac{5}{4}}\times \mathbb{P}\left[\mathcal{F}_{\delta_m}(x_1,x_2,x_3,x_4)\right]}_{T_6}\times f(x_1,x_2,x_3,x_4;\delta_m).
\end{equation} 

The rest of the proof uses two lemmas which are stated below and proved in the next section.
Thanks to Lemma~\ref{lem::four_aux1} below, we have 
\begin{equation} \label{eq::four_aux4}
	\lim_{m\to \infty} T_6=C_6\prod_{1\leq j<k\leq 4}|\sin(\theta^k-\theta^j)|^{\frac{1}{3}}. 
\end{equation}
{Since the left-hand side of \eqref{eqn::four_aux3} does not contain $\delta_m$, \eqref{eq::four_aux4} implies that $\lim_{m \to \infty}f(x_1,x_2,x_3,x_3;\delta_m)$ exists.
In particular, we let $C_7:=\lim_{m \to \infty}f(1,\ii,-1,-\ii;\delta_m) \in (0,\infty)$.} Then, according to Lemma~\ref{lem::four_aux2} below, we have 
\begin{equation*}
	\lim_{m \to \infty} f(x_1,x_2,x_3,x_4;\delta_m)=C_7.
\end{equation*}
{Therefore, letting $m \to \infty$ in~\eqref{eqn::four_aux3} yields}
\begin{equation} \label{eqn::C7_independent}
\lim_{a\to 0}	\orho_a^{-1}\times \mathbb{P}^a\left[\mathcal{E}^a\left(x_1,x_2,x_3,x_4\right)\right]=C_6C_7\prod_{1\leq j<k\leq 4} |\sin(\theta^k-\theta^j)|^{\frac{1}{3}},
\end{equation}
which gives the desired result with $C_5=C_6C_7$.
\end{proof}

\begin{lemma} \label{lem::four_aux1}
	Assume the same setup as in Theorem~\ref{thm::pivotal}. There exists a universal constant $C_6\in (0,\infty)$ such that 
	\begin{equation*}
		\lim_{m\to \infty} \delta_m^{-\frac{5}{4}}\times \mathbb{P}\left[\mathcal{F}_{\delta_m}\left(x_1,x_2,x_3,x_4\right)\right]= C_6\prod_{1\leq j<k\leq 4}	|\sin(\theta^k-\theta^j)|^{\frac{1}{3}}.
	\end{equation*}
\end{lemma}

\begin{lemma} \label{lem::four_aux2}
Let $f(x_1,x_2,x_3,x_4;\delta_m)$ be defined as in the proof of Theorem~\ref{thm::pivotal}.
{Then, we have 
\begin{equation*}
	\lim_{m\to \infty} f(x_1,x_2,x_3,x_4;\delta_m)=\lim_{m\to \infty} f(1,\ii,-1,-\ii;\delta_m).
\end{equation*}}
\end{lemma}

We end this section with a corollary of the proof of Theorem~\ref{thm::pivotal}.

\begin{corollary} \label{coro::four_arm}
	There exists  a universal constant $C_8\in (0,\infty)$ such that 
	\begin{align} 
	\lim_{m\to \infty} \delta_{m}^{-\frac{5}{4}}\times \mathbb{P}\left[\mathcal{F}_{\delta_m,1}(0)\right]=C_8. \label{eqn::coro_aux2}
\end{align}
\end{corollary}
\begin{proof}[Sketch of the proof]
	 As in the proof of~\eqref{eqn::four_aux3}, we can show that
	\begin{equation*}
			1=\lim_{a\to 0} \orho_a^{-1}\times \mathbb{P}^a\left[0\xlongleftrightarrow{BWBW}\partial B_1(0)\right]=\delta_m^{-\frac{5}{4}}\times\mathbb{P}\left[\mathcal{F}_{\delta_m,1}(0)\right]\times f_{\delta_m},
	\end{equation*}
where 
\begin{equation*}
	f_{\delta_m}:=\lim_{a\to 0}\frac{\mathbb{P}^a\left[0\xlongleftrightarrow{BWBW}\partial B_1(0)|\mathcal{F}_{\delta_m,1}^a(0)\right]}{\mathbb{P}^a\left[0\xlongleftrightarrow{BWBW}\partial B_{\delta_m}(0)\right]}.
\end{equation*}
From the proof of Theorem~\ref{thm::pivotal}, we know that $\lim_{m\to\infty}f(x_1,x_2,x_3,x_4;\delta_m)$ exists; we call this limit $C_{7}$.
One then can proceed as in the proof of Lemma~\ref{lem::four_aux2}, given in the next section, to show that\footnote{Compared with the proof of Lemma~\ref{lem::four_aux2}, there is some additional work to do, due to the lack of initial faces. However, the strategy is still clear, see~\cite[Sketch of the proof of Proposition~3.6]{GarbanPeteSchrammPivotalClusterInterfacePercolation}.}
\begin{equation*}
	\lim_{m\to \infty}\frac{f_{\delta_m}}{f(x_1,x_2,x_3,x_4;\delta_m)} =1.
\end{equation*} 
Combining all of these observations together, we obtain
\begin{equation*}
	\lim_{m\to \infty}\delta_m^{-\frac{5}{4}}\times \mathbb{P}\left[\mathcal{F}_{\delta_m}(0)\right]=C_8:=1/ C_{7}.
\end{equation*}
\end{proof}

\subsection{Proof of technical lemmas}
The goal of this section is to prove Lemmas~\ref{lem::four_aux1} and \ref{lem::four_aux2}. 
We start with Lemma~\ref{lem::four_aux1}.
We assume the following boundary condition on $\partial B_1(0)$:
\begin{equation} \label{eqn::decla}
	\text{the hexagons intersecting $(x_1 x_2) \cup (x_3 x_4)$ are black and other hexagons intersecting $\partial B_1(0)$ are white.}
\end{equation}
Note that this has no influence on the events $\mathcal{F}^a_{\delta_m}(x_1,x_2,x_3,x_4)$ and $\mathcal{E}^a(x_1,x_2,x_3,x_4)$ or on their probabilities.
The choice of boundary conditions induces two interfaces inside $B_1(0)$.
We use these interfaces to define four paths, $\gamma_j^a, j=1,2,3,4$, corresponding to parametrizations of the interfaces starting at the points, $y^a_1,\ldots,y^a_4$, where they intersect $\partial B_{1}(0)$.
Let $y_j^a\in a\mathcal{H}$ be the vertex of $a\mathcal{H}$ where the interface $\gamma_j^a$ starts, for $j=1,2,3,4$. For $z\in \mathbb{C}$ and $A\subseteq \mathbb{C}$, define
	\begin{equation*}
		d(z,A):=\inf_{w\in A} |w-z|. 
\end{equation*}
Note that 
\begin{align} \label{eqn::proba_discre_to_conti}
	\mathcal{E}^a\left(x_1,x_2,x_3,x_4\right)=\{\gamma_j^a \text{ reaches }0,\enspace 1\leq j\leq 4\},\quad \mathcal{F}_{\delta_m}^a\left(x_1,x_2,x_3,x_4\right)=\{d(0,\gamma_j^a)\leq \delta_m, \enspace 1\leq j\leq 4\}.
\end{align}

Write
\begin{equation*}
	\alpha_1:=\{\{1,2\},\{3,4\}\},\quad \alpha_2:=\{\{1,4\},\{2,3\}\}.
\end{equation*}
Let $\varphi $ be a conformal map from $\mathbb{D}$ onto $\mathbb{H}$ with $-\infty<\mathring{x}_1<\cdots<\mathring{x}_4<\infty$, where $\mathring{x}_j=\varphi(x_j)$ for $1\leq j\leq 4$. 
Define
\begin{align}
	\mathcal{Z}_{\alpha_1}(\mathring{x}_1,\mathring{x}_2,\mathring{x}_3,\mathring{x}_4)=&\chi(\mathring{x}_1,\mathring{x}_2,\mathring{x}_3,\mathring{x}_4)^{\frac{1}{3}} \frac{H\left(\chi(\mathring{x}_1,\mathring{x}_2,\mathring{x}_3,\mathring{x}_4)\right)}{H(1)}, \label{eqn::alpha_1}\\
	\mathcal{Z}_{\alpha_2}(\mathring{x}_1,\mathring{x}_2,\mathring{x}_3,\mathring{x}_4)=&\left(1-\chi(\mathring{x}_1,\mathring{x}_2,\mathring{x}_3,\mathring{x}_4)\right)^{\frac{1}{3}} \frac{H\left(1-\chi(\mathring{x}_1,\mathring{x}_2,\mathring{x}_3,\mathring{x}_4)\right)}{H(1)},\label{eqn::alpha_2}
\end{align}
where $\chi$ is the cross-ratio and $H$ is a hypergeometric function:
\begin{equation*}
	\chi(\mathring{x}_1,\mathring{x}_2,\mathring{x}_3,\mathring{x}_4):=\frac{(\mathring{x}_2-\mathring{x}_1)(\mathring{x}_4-\mathring{x}_3)}{(\mathring{x}_4-\mathring{x}_2)(\mathring{x}_3-\mathring{x}_1)},\quad H(z)=\ _2F_1\Big(\frac{2}{3},\frac{1}{3},\frac{4}{3};z\Big).
\end{equation*}
 
For two collections of planar continuous oriented curves $\left(\gamma_j\right)_{1\leq j\leq 4}$ and $\left(\tilde{\gamma}_j\right)_{1\leq j\leq 4}$,
we define 
\begin{equation}\label{eqn::def_dist4}
	\dist \left(\left(\gamma_j\right)_{1\leq j\leq 4},\left(\tilde{\gamma}_j\right)_{1\leq j\leq 4}\right):=\sup_{1\leq j\leq 4} \dist\left(\gamma_j.\tilde{\gamma}_j\right),
\end{equation}
where $\dist(\gamma_j,\tilde{\gamma}_j)$ is defined in~\eqref{eqn::curve_metric}. 

\begin{lemma} \label{lem::four_aux11}
	Assume the same setup as in Theorem~\ref{thm::pivotal}. Then as $a\to 0$, the law of $\left(\gamma_j^a\right)_{1\leq j\leq 4}$ converges weakly in the topology induced by $\dist$ in~\eqref{eqn::def_dist4} to 
	\begin{equation*}
		\mathbb{P}_4:=\mathcal{Z}_{\alpha_1}(\mathring{x}_1,\mathring{x}_2,\mathring{x}_3,\mathring{x}_4)\times \mathbb{P}_{\alpha_1}+\mathcal{Z}_{\alpha_2}(\mathring{x}_1,\mathring{x}_2,\mathring{x}_3,\mathring{x}_4)\times \mathbb{P}_{\alpha_2},
	\end{equation*}
where $\mathbb{P}_{\alpha_j}$ is the so-called global $2$-multiple $\SLE_6$ on $B_1(0)$ with marked points $(x_1,x_2,x_3,x_4)$ and link pattern $\alpha_j$, whose definition is given by~\cite[Definition~1.1  and Section~1.3]{BeffaraPeltolaWuUniqueness}, uniqueness is given by~\cite[Theorem~A.1]{NonSimple} and whose existence is given by~\cite[Corollary~B.2]{LiuPeltolaWuUST}.
\end{lemma}
\begin{proof}
This is a consequence of~\cite[Theorem~B.1 and Corollary~B.2]{LiuPeltolaWuUST}.
\end{proof}

\begin{lemma} \label{lem::four_aux12}
	Assume the same setup as in Theorem~\ref{thm::pivotal}. Then there exists a universal constant $C_{9}\in (0,\infty)$ such that for $1\leq j\leq 2$, we have 
	\begin{equation}
	\lim_{m\to \infty}	\delta_m^{-\frac{5}{4}}\mathbb{P}_{\alpha_j}\left[d(0,\gamma_k)\leq \delta_m,\enspace 1\leq k\leq 4\right]=C_{9}\times \frac{\prod_{1\leq k<r\leq 4} |\sin(\theta^k-\theta^r)|^{\frac{1}{3}}}{\mathcal{Z}_{\alpha_j}\left(\mathring{x}_1,\mathring{x}_2,\mathring{x}_3,\mathring{x}_4\right)}.
	\end{equation}
\end{lemma}
\begin{proof}
	This is a special case of~\cite[Theorem~1.1]{ZhanGreen2SLE} with $\kappa=6$, $D=\unitD$ and $z_0=0$.
\end{proof}

\begin{proof}[Proof of Lemma~\ref{lem::four_aux1}]
We have 
\begin{align*}
& \lim_{m\to \infty} \delta_m^{-\frac{5}{4}}\times \mathbb{P}\left[\mathcal{F}_{\delta_m}\left(x_1,x_2,x_3,x_4\right)\right] = \lim_{m\to \infty}\delta_m^{-\frac{5}{4}}\lim_{a\to 0} \mathbb{P}^a\left[\mathcal{F}_{\delta_m}^a\left(x_1,x_2,x_3,x_4\right)\right] \\
& \qquad \qquad = \lim_{m\to \infty}\delta_m^{-\frac{5}{4}}\lim_{a\to 0} \mathbb{P}^a\left[d(0,\gamma_j^a)\leq \delta_m, \enspace 1\leq j\leq 4\right] \\
& \qquad \qquad = \lim_{m\to \infty}\delta_{m}^{-\frac{5}{4}} \left(\sum_{j=1}^2 \mathcal{Z}_{\alpha_j}(\mathring{x}_1,\mathring{x}_2,\mathring{x}_3,\mathring{x}_4)\times \mathbb{P}_{\alpha_j}\left[d(0,\gamma_k)\leq \delta_m,\enspace 1\leq k\leq 4 \right]\right) \\
& \qquad \qquad = C_6\prod_{1\leq j<k\leq 4}|\sin(\theta^j-\theta^k)|^{\frac{1}{3}},
\end{align*}
where $C_6:=2C_{9}$, the second equality is due to the observation~\eqref{eqn::proba_discre_to_conti}, the third equality is due to Lemma~\ref{lem::four_aux11}, and the last equality is due to Lemma~\ref{lem::four_aux12}.
\end{proof}

\begin{proof}[Proof of Lemma~\ref{lem::four_aux2}]

Let $z_1^a,z_2^a,z_3^a,z_4^a \in a\mathcal{T}$ denote four hexagons which intersect $\partial B_1(0)$ and whose centers are close to $1,\ii,-1,-\ii$, respectively.
To prove the lemma, we need to compare the conditional probabilities $\mathbb{P}^a\left[\mathcal{E}^a(x_1^a,x_2^a,x_3^a,x_4^a)| \mathcal{F}^a_{\delta_m}(x_1^a,x_2^a,x_3^a,x_4^a)\right]$ and $\mathbb{P}^a\left[\mathcal{E}^a(z_1^a,z_2^a,z_3^a,z_4^a)| \mathcal{F}^a_{\delta_m}(z_1^a,z_2^a,z_3^a,z_4^a)\right]$, where the events $\mathcal{E}^a\left(x_1^a,x_2^a,x_3^a,x_4^a\right)$ and $\mathcal{E}^a\left(z_1^a,z_2^a,z_3^a,z_4^a\right)$ depend on the labels of vertices inside the annulus $B_{1}(0)\setminus B_{\delta_m}(0)$, including vertices that are very close to $\partial B_{1}(0)$.
To do this, we will couple the two conditional measures $\mathbb{P}^a\left[\, \cdot \, |\mathcal{F}^a_{\delta_m}(x_1^a,x_2^a,x_3^a,x_4^a)\right]$ and $\mathbb{P}^a\left[\, \cdot \,  |\mathcal{F}_{\delta_m}^a(z_1^a,z_2^a,z_3^a,z_4^a)\right]$.
To this end, we will use the notion, introduced in~\cite{GarbanPeteSchrammPivotalClusterInterfacePercolation}, of ``faces" induced by arms of alternating labels. We briefly recall its definition.

 Let $r\in (0,1)$ and, for $1\leq j\leq 4$, let $z_j\in\partial B_r(0)$ be on the boundary of some hexagon $\tilde{z}_j\in a\mathcal{T}$ which intersects $\partial B_r(0)$, where $z_1,\ldots,z_4$ are chosen in counterclockwise order. A \textit{configuration of faces} $\boldsymbol{\eta}$ around the circle $\partial B_r(0)$ with endpoints $z_1,\ldots,z_4$ is a collection of four oriented simple paths $(\eta_1,\eta_2,\eta_3,\eta_4)$ consisting of hexagons of $a\mathcal{H}$ such that, for $j=1,2,3,4$ (with the convention that $z_5=z_1$),
		\begin{itemize}
			\item $z_j$ is on the boundary of the first hexagon of the path $\eta_j$ and $z_{j+1}$ is on the boundary of the last hexagon of $\eta_j$;
			\item $\eta_j$ is a path consisting of black (resp., white) hexagons if $j$ is odd (resp., even);
			\item there are no hexagons in $\eta_j$ that are entirely contained in $B_r(0)$.
		\end{itemize}
		
Now recall that $y_j^a$ is the starting point of the interface $\gamma_j^a$ and note that the assumption~\eqref{eqn::decla} induces a configuration of faces around $\partial B_1(0)$ with endpoints $(y_1^a,\ldots,y_4^a)$. We denote this configuration of faces by $\boldsymbol{\eta}_1(x_1^a,x_2^a,x_3^a,x_4^a)$.

Let $r\in (0,1)$ and note that, on the event $\mathcal{F}_r(x_1^a,x_2^a,x_3^a,x_4^a)$, $\gamma_j^a$ intersects $\partial B_r(0)$.
For $j=1,2,3,4$,  we denote by $\tilde{\gamma}_j^a$ the portion of $\gamma_j^a$ that runs from $y_j^a$ until $\gamma_j^a$ first hits $\partial B_r(0)$, and we call $y_j^{(r,a)}$ this first hitting point.
Let $H$ be the set of hexagons of $a\mathcal{H}$ that are adjacent to $\cup_{j=1}^4 \tilde{\gamma}_j^a$ and define $\tilde{H}$ to be the union of
\begin{itemize}
    \item black hexagons in $a\mathcal{H}\cap A_{r,1}(0)$ that are connected to some black hexagon in $H$ by a black path inside $a\mathcal{H}\cap A_{r,1}(0)$ and
    \item white hexagons in $a\mathcal{H}\cap A_{r,1}(0)$ that are connected to some white hexagon in $H$ by a white path inside $a\mathcal{H}\cap A_{r,1}(0)$.
\end{itemize}
Then, the connected component of $\mathbb{C}\setminus \left(H\cup \tilde{H}\right)$ containing $0$ is a bounded domain and the hexagons in $a\mathcal{H}$ that lie on the boundary of this domain form a configuration of faces around $\partial B_r(0)$ with endpoints $y_1^{(r,a)},\ldots,y_4^{(r,a)}$. We let
		\begin{equation*}
			\boldsymbol{\eta}_r(x_1^a,x_2^a,x_3^a,x_4^a)=\Big(\eta_1^{(r,a)},\eta_2^{(r,a)},\eta_2^{(r,a)},\eta_4^{(r,a)}\Big)
		\end{equation*}
		denote this configuration of faces. Similarly, one can define 
		\begin{equation*}
		\boldsymbol{\eta}_r(z_1^a,z_2^a,z_3^a,z_4^a)=\Big(\tilde{\eta}_1^{(r,a)},\tilde{\eta}_2^{(r,a)},\tilde{\eta}_3^{(r,a)},\tilde{\eta}_4^{(r,a)}\Big),
		\end{equation*}
starting with the configuration of faces by $\boldsymbol{\eta}_1(z_1^a,z_2^a,z_3^a,z_4^a)$ determined by the boundary condition \eqref{eqn::decla} with the points $(x_1^a,x_2^a,x_3^a,x_4^a)$ replaced by $(z_1^a,z_2^a,z_3^a,z_4^a)$.

The following coupling result can be proved using the same strategy as in~\cite[Proof of Proposition~3.1]{GarbanPeteSchrammPivotalClusterInterfacePercolation}, therefore we omit its proof. 
 \begin{lemma} \label{lem::faces}
	Assume the same setup as in Theorem~\ref{thm::pivotal}. 
	Then there exists a universal constant $c_2\in (0,\infty)$ such that the following holds. There exists a coupling $\mathbb{P}_m^a$, between $\mathbb{P}^a\left[\, \cdot \,  |\mathcal{F}_{\delta_m}^a(x_1^a,x_2^a,x_3^a,x_4^a)\right]$ and $\mathbb{P}^a\left[\, \cdot \,  |\mathcal{F}_{\delta_m}^a(z_1^a,z_2^a,z_3^a,z_4^a)\right]$, that is, a joint distribution on pairs $(\tilde{\Lambda}^a,\hat{\Lambda}^a)$ such that $\tilde{\Lambda}^a$ and $\hat{\Lambda}^a$ are distributed according to $\mathbb{P}^a\left[\, \cdot \, |\mathcal{F}_{\delta_m}^a(x_1^a,x_2^a,x_3^a,x_4^a)\right]$ and $\mathbb{P}^a\left[\, \cdot \, |\mathcal{F}_{\delta_m}^a(z_1^a,z_2^a,z_3^a,z_4^a)\right]$, respectively, such that the following holds:
	\begin{equation*}
		\mathbb{P}_{m}^a\left[\boldsymbol{\eta}_{\delta_m}(x_1^a,x_2^a,x_3^a,x_4^a)=\boldsymbol{\eta}_{\delta_m}(z_1^a,z_2^a,z_3^a,z_4^a)\subseteq A_{\delta_m,\sqrt{\delta_m}}(0)\right]\geq 1-\delta^{c_2}_{m}
	\end{equation*}
	and, if the event $\mathcal{G}=\{\boldsymbol{\eta}_{\delta_m}(x_1^a,x_2^a,x_3^a,x_4^a)=\boldsymbol{\eta}_{\delta_m}(z_1^a,z_2^a,z_3^a,z_4^a)\subseteq A_{\delta_m,\sqrt{\delta_m}}(0)\}$ happens, then $\tilde{\Lambda}^a$ and $\hat{\Lambda}^a$ coincide inside the domain enclosed by $\boldsymbol{\eta}_{\delta_m}(x_1^a,x_2^a,x_3^a,x_4^a)=\boldsymbol{\eta}_{\delta_m}(z_1^a,z_2^a,z_3^a,z_4^a)$.
\end{lemma}

Going back to the proof of Lemma~\ref{lem::four_aux2}, using Lemma~\ref{lem::faces} above in the third equality below, on event $\mathcal{G}$, we have 
\begin{align*}
& \mathbb{P}^a\left[\mathcal{E}^a(x_1^a,x_2^a,x_3^a,x_4^a) |\mathcal{F}_{\delta_m}^a(x_1^a,x_2^a,x_3^a,x_4^a), \boldsymbol{\eta}_{\delta_m}(x_1^a,x_2^a,x_3^a,x_4^a)\right]\\
& \qquad = \frac{\mathbb{P}^a\left[\mathcal{E}^a(x_1^a,x_2^a,x_3^a,x_4^a),\mathcal{F}_{\delta_m}^a(x_1^a,x_2^a,x_3^a,x_4^a)|\boldsymbol{\eta}_{\delta_m}(x_1^a,x_2^a,x_3^a,x_4^a)\right]}{\mathbb{P}^a\left[\mathcal{F}_{\delta_m}^a(x_1^a,x_2^a,x_3^a,x_4)|\boldsymbol{\eta}_{\delta_m}(x_1^a,x_2^a,x_3^a,x_4^a)\right]}\\
& \qquad = \mathbb{P}^a\left[\eta_1^{(\delta_m,a)},\eta_3^{(\delta_m,a)}\xlongleftrightarrow{B}0,\enspace\eta_2^{(\delta_m,a)},\eta_4^{(\delta_m,a)}\xlongleftrightarrow{W}0|\boldsymbol{\eta}_{\delta_m}(x_1^a,x_2^a,x_3^a,x_4^a)\right]\\
& \qquad = \mathbb{P}^a\left[\tilde{\eta}_1^{(\delta_m,a)},\tilde{\eta}_3^{(\delta_m,a)}\xlongleftrightarrow{B}0,\enspace\tilde{\eta}_2^{(\delta_m,a)},\tilde{\eta}_4^{(\delta_m,a)}\xlongleftrightarrow{W}0|\boldsymbol{\eta}_{\delta_m}(z_1^a,z_2^a,z_3^a,z_4^a)\right]\\
& \qquad = \mathbb{P}^a\left[\mathcal{E}^a(z_1^a,z_2^a,z_3^a,z_4^a) |\mathcal{F}_{\delta_m}^a(z_1^a,z_2^a,z_3^a,z_4^a), \boldsymbol{\eta}_{\delta_m}(z_1^a,z_2^a,z_3^a,z_4^a)\right].
	\end{align*}
Thus, 
\begin{align*}
	\left|\frac{\mathbb{P}^a\left[\mathcal{E}^a(x_1^a,x_2^a,x_3^a,x_4^a)| \mathcal{F}^a_{\delta_m}\left(x_1^a,x_2^a,x_3^a,x_4^a\right)\right]}{\mathbb{P}^a\left[0\xlongleftrightarrow{BWBW} \partial B_{\delta_m}(0)\right]}-\frac{\mathbb{P}^a\left[\mathcal{E}^a(z_1^a,z_2^a,z_3^a,z_4^a)| \mathcal{F}^a_{\delta_m}\left(z_1^a,z_2^a,z_3^a,z_4^a\right)\right]}{\mathbb{P}^a\left[0\xlongleftrightarrow{BWBW} \partial B_{\delta_m}(0)\right]}\right|\leq \mathbb{P}_m^{a}\left[\mathcal{G}^c\right]\leq \delta_m^{c_2}.
\end{align*}
Letting $a\to 0$ gives
\begin{equation*}
	|f(x_1,x_2,x_3,x_4;\delta_m)-f(1,\ii,-1,-\ii;\delta_m)|\leq \delta_m^{c_2},
\end{equation*}
which implies the desired result readily. 
\end{proof}
\subsection{Percolation backbone: Proof of Theorem~\ref{thm::backbone}} \label{sec::thm::backbone}
We start with a coupling result. For $z\in \mathbb{C}$, we denote by $\mathcal{A}^a_{\eta,\epsilon}(z)$ the event that there is a black path  connecting $\partial B_{\eta}(z)$ to $\partial B_{\epsilon}(z)$ 
and denote by $\mathcal{B}_{\eta,\epsilon}^a(z)$ the event that there are two disjoint black paths  connecting $\partial B_{\eta}(z)$ to $\partial B_{\eta}(z)$. 

The two events $\mathcal{A}^a_{\eta,\epsilon}(z), \mathcal{B}_{\eta,\epsilon}^a(z)$ can be expressed in terms of interfaces.  The event $\mathcal{A}^a_{\eta,\epsilon}(z)$ means that there is no interface loop between black and white hexagons separating $B_{\eta}(z)$ from the complement of $B_{\epsilon}(z)$ and that one of the two following events occurs:
	\begin{itemize}
		\item there is a counterclockwise interface intersecting both $B_{\eta}(z)$ and the complement of $B_{\epsilon}(z)$;
		\item the innermost interface surrounding $B_{\eta}(z)$ is oriented counterclockwise.
\end{itemize}
Thus, we can define the natural analog of event $\mathcal{A}^a_{\eta,\epsilon}(z)$ in the continuum for the full scaling limit $\Lambda$, which we denote by $\mathcal{A}_{\eta,\epsilon}(z)$.
Moreover, since the boundary polychromatic $3$-arm exponent for critical site percolation is strictly larger than $1$~\cite{SmirnovWernerCriticalExponents}, using Lemma~6.1 of~\cite{CamiaNewmanPercolationFull}, one can conclude that $\mathcal{A}_{\eta,\epsilon}(z)$ is a continuity event for $\mathbb{P}$.
Therefore, we have 
\begin{equation*}
\lim_{a\to 0}\mathbb{P}^a\left[\mathcal{A}_{\eta,\epsilon}^a(z)\right]=\mathbb{P}\left[\mathcal{A}_{\eta,\epsilon}(z)\right].
\end{equation*}

The situation for $\mathcal{B}_{\eta,\epsilon}^a(z)$ is similar.
The event $\mathcal{B}_{\eta,\epsilon}^a(z)$ occurs if $\mathcal{A}_{\eta,\epsilon}^a(z)$ occurs and there is no pivotal vertex for $\mathcal{A}_{\eta,\epsilon}^a(z)$.
The last condition implies that one of the three following events occurs:
\begin{itemize}
	\item there are at least $2$ counterclockwise interfaces intersecting both $B_{\eta}(z)$ and the complement of $B_{\epsilon}(z)$;
	\item there is only one counterclockwise interface intersecting both $B_{\eta}(z)$ and the complement of $B_{\epsilon}(z)$ and, if we call $H_1$ the black cluster that is adjacent to the left-hand side of this interface, then there is no black hexagon $x^a\in H_1\cap A_{\eta,\epsilon}(z)$ such that $H_1\setminus \{x^a\}$ does not connect $\partial B_{\eta}(z)$ to $\partial B_{\epsilon}(z)$;
	\item the innermost interface surrounding $B_{\eta}(z)$ is oriented counterclockwise and, if we call $H_2$ the black cluster that is adjacent to the left-hand side of this interface, then there is no black hexagon $x^a\in H_2\cap A_{\eta,\epsilon}(z)$ such that $H_2\setminus \{x^a\}$ does not connect $\partial B_{\eta}(z)$ to $\partial B_{\epsilon}(z)$. 
\end{itemize}
Now we can define the analog of event $\mathcal{B}^a_{\eta,\epsilon}(z)$ in the continuum, which we denote by $\mathcal{B}_{\eta,\epsilon}(z)$.
The definition of $\mathcal{B}_{\eta,\epsilon}(z)$ is similar to that of $\mathcal{B}^a_{\eta,\epsilon}(z)$, but with $\mathcal{A}^a_{\eta,\epsilon}(z)$ replaced by $\mathcal{A}_{\eta,\epsilon}(z)$ and the events in the last two bullet points above replaced by:
\begin{itemize}
\item there is only one counterclockwise interface intersecting both $B_{\eta}(z)$ and the complement of $B_{\epsilon}(z)$; we call the closure of the union of bounded domains surrounded by this interface and the rightmost (clockwise) interface that is on the left of it $H_1$; then there is no $x\in H_1\cap A_{\eta,\epsilon}(z)$ such that $H_1\setminus \{x\}$ does not connect $\partial B_{\eta}(z)$ to $\partial B_{\epsilon}(z)$; 
\item the innermost interface surrounding $B_{\eta}(z)$ is oriented counterclockwise; we call the closure of the union of bounded domains surrounded by this interface and the rightmost (clockwise) interface that is on the left of it $H_2$ ; then there is no $x\in H_2\cap A_{\eta,\epsilon}(z)$ such that $H_2\setminus \{x\}$ does not connect $\partial B_{\eta}(z)$ to $\partial B_{\epsilon}(z)$. 
\end{itemize}
Since for critical site percolation, the polychromatic boundary $3$-arm exponent is strictly larger than $1$ and the polychromatic interior $6$-arm exponent is strictly larger than $2$~\cite{SmirnovWernerCriticalExponents}, using Lemma~6.1 of~\cite{CamiaNewmanPercolationFull}, one can conclude that $\mathcal{B}_{\eta,\epsilon}(z)$ is also a continuity event for $\mathbb{P}$. Therefore, we have 
\begin{equation*}
	\lim_{a\to 0}\mathbb{P}^a\left[\mathcal{B}_{\eta,\epsilon}^a(z)\right]=\mathbb{P}\left[\mathcal{B}_{\eta,\epsilon}(z)\right].
\end{equation*}

\begin{lemma} \label{lem::backbone}
	Fix $\epsilon>\delta>a$. For any $\delta>\eta>a$, there exist an event $\mathcal{S}$ and a coupling, $\mathbb{P}_{\eta}^a$, between $\mathbb{P}^a\left[\, \cdot \, |0\xlongleftrightarrow{BB}\partial B_{\epsilon}(0)\right]$ and $\mathbb{P}^a\left[\, \cdot \, | \mathcal{B}_{\eta,\epsilon}^a(0)\right]$ (that is, a joint distribution on pairs $(\tilde{\Lambda}^a,\hat{\Lambda}^a)$ such that $\tilde{\Lambda}^a$ and $\hat{\Lambda}^a$ are distributed according to $\mathbb{P}^a\left[\, \cdot \, |0\xlongleftrightarrow{BB}\partial B_{\epsilon}(0)\right]$ and $\mathbb{P}^a\left[\, \cdot \, | \mathcal{B}_{\eta,\epsilon}^a(0)\right]$, respectively) with the following properties:
	\begin{equation} \label{eqn::proba_lower_bound_open_circuit}
\mathbb{P}_{\eta}^a\left[\mathcal{S}\right]\geq \mathbb{P}^a\left[\exists \text{ black circuit in }A_{\eta,\delta}(0) \text{ surrounding }0\right],
	\end{equation}
and, for any event $\mathcal{A}$ that depends only on the states of hexagons of a single percolation configuration outside $B_{\delta}(0)$,
\begin{equation} \label{eqn::banckbone_proba_coupling}
	\mathbb{P}_{\eta}^a\left[\tilde{\Lambda}^a\in \mathcal{A}|\mathcal{S}\right]=\mathbb{P}_{\eta}^a\left[\hat{\Lambda}^a\in \mathcal{A}|\mathcal{S}\right].
\end{equation} 
\end{lemma}
\begin{proof}
Since the conditioning in the two probabilities that need to be coupled involves crossings of a single label, one can basically repeat the argument in~\cite[Proof of Lemma~2.1]{Cam23}, which we explain briefly below. 

We start by generating a critical percolation configuration, $\Lambda^a$. Then, one can proceed as in~\cite[Proof of Lemma~2.1]{Cam23} to explore $a\mathcal{H}$ starting from the innermost circuit $G_0$ of hexagons in $\mathcal{A}_{\eta,\epsilon}$ surrounding $0$ and, thanks to the stochastic domination given by the FKG inequality, use $\Lambda^a$ to construct two configurations, $\tilde{\Lambda}\sim \mathbb{P}^a\left[\cdot |0\xlongleftrightarrow{BB}\partial B_{\epsilon}(0\right]$ and $\hat{\Lambda}^a\sim \mathbb{P}^a\left[\cdot | \mathcal{B}_{\eta,\epsilon}^a(0\right]$, such that 
\begin{equation}\label{eqn::configu_domination}
    \mathbb{1}_{\{x_i^a \text{ is open in }\Lambda^a\}}\leq \mathbb{1}_{\{x_i^a \text{ is open in }\tilde{\Lambda}^a\}}, \mathbb{1}_{\{x_i^a \text{ is open in }\hat{\Lambda}^a\}},\quad \text{for all }x_{i}^a \in a\mathcal{H}, 
\end{equation}
where $\mathbb{1}_{\{\cdot\}}$ is the indicator function. 

We denote by $\hat{\mathcal{S}}$ the event that $\Lambda^a$ has an open circuit surrounding $0$ that is fully contained in $\mathcal{A}_{\eta,\epsilon}$ and by $\mathcal{S}$ the event that there is a common open circuit in $\tilde{\Lambda}^a$ and $\hat{\Lambda}^a$ that is fully contained in $\mathcal{A}_{\eta,\epsilon}$.
Thanks to the relation~\eqref{eqn::configu_domination}, we have $\hat{\mathcal{S}}\subseteq \mathcal{S}$, which implies~\eqref{eqn::proba_lower_bound_open_circuit}.
If $\mathcal{S}$ occurs, we denote by $\gamma^a$ the innermost common black circuit in $\tilde{\Lambda}^a$ and $\hat{\Lambda}^a$ that is fully contained in $\mathcal{A}_{\eta,\epsilon}$.

Assuming that $\mathcal{S}$ occurs and given $\gamma^a$, let $\tilde{\omega}_{\gamma^a}$ and $\hat{\omega}_{\gamma^a}$ be the configurations generated inside $\gamma^a$ for $\tilde{\Lambda}^a$ and $\hat{\Lambda}^a$, respectively. For any event $\mathcal{A}$ that depends only on the labels of hexagons of a single percolation configuration outside $B_{\delta}(0)$, we have 
\begin{align*}
\mathbb{P}^a_{\eta}\left[\tilde{\Lambda}^a\in\mathcal{A}|\gamma^a, \tilde{\omega}_{\gamma^a}\right] = &\frac{\mathbb{P}^a\left[\mathcal{A}, 0\xlongleftrightarrow{BB}\partial B_{\epsilon}(0)|\gamma^a,\tilde{\omega}_{\gamma^a}\right]}{\mathbb{P}^a\left[ 0\xlongleftrightarrow{BB}\partial B_{\epsilon}(0)|\gamma^a,\tilde{\omega}_{\gamma^a}\right]}\\
=& \frac{\mathbb{P}^a\left[\mathcal{A},\gamma^a\xlongleftrightarrow{BB}\partial B_{\epsilon}(0)|\gamma^a,\tilde{\omega}_{\gamma^a}\right]}{\mathbb{P}^a\left[\gamma^a\xlongleftrightarrow{BB}\partial B_{\epsilon}(0)|\gamma^a,\tilde{\omega}_{\gamma^a}\right]}\\
=&\mathbb{P}^a\left[\mathcal{A}| \gamma^a\xlongleftrightarrow{BB}\partial B_{\epsilon}(0)\right],
\end{align*}
where the second equality is due to the fact that, since we assume that $\mathcal{S}$ occurs, $\tilde{\omega}_{\gamma^a}$ must contain two disjoint black paths connecting $0$ to $\gamma^a$, and the last equality follows from the independence of percolation and the fact that $\mathcal{A}$ depends only on the labels of hexagons outside $B_{\delta}(0)$. Similarly, we have 
\begin{equation*}
\mathbb{P}^a_{\eta}\left[\hat{\Lambda}^a\in\mathcal{A}|\gamma^a, \hat{\omega}_{\gamma^a}\right]= \mathbb{P}^a\left[\mathcal{A}| \gamma^a\xlongleftrightarrow{BB}\partial B_{\epsilon}(0)\right].
\end{equation*}
Consequently, we obtain~\eqref{eqn::banckbone_proba_coupling}.
\end{proof}

According to~\cite[Theorem~1]{nolin2023backbone}, 
\begin{equation}\label{eqn::backbone_o1}
\lim_{a\to 0} \frac{\log \rho_a}{\log a}=\lim_{a\to 0} \frac{\log \mathbb{P}^a\left[0\xlongleftrightarrow{BB}\partial B_{1}(0)\right]}{\log a}=\xi.
\end{equation}
With Lemma~\ref{lem::backbone} and~\eqref{eqn::backbone_o1} at hand, one can conclude that
\begin{equation} \label{eqn::backbone_scal}
	\lim_{a\to 0} \frac{\mathbb{P}^a\left[0\xlongleftrightarrow{BB}\partial B_{\epsilon}(0)\right]}{\mathbb{P}^a\left[0\xlongleftrightarrow{BB}\partial B_{1}(0)\right]} = \epsilon^{-\xi},\quad \forall \epsilon>0. 
\end{equation}
Indeed,  without loss of generality, we may assume that $\epsilon\in(0,1)$.  With~\cite[Lemma~2.1]{Cam23} replaced by our Lemma~\ref{lem::backbone}, one can proceed as in~\cite[Proof of Theorem~1.1]{Cam23} to show that 
\begin{equation}\label{eqn::backbone_aux1}
	\lim_{a\to 0} \frac{\mathbb{P}^a\left[0\xlongleftrightarrow{BB}\partial B_{\epsilon}(0)\right]}{\mathbb{P}^a\left[0\xlongleftrightarrow{BB}\partial B_{1}(0)\right]}= \frac{1}{\mathbb{P}\left[ 0\xlongleftrightarrow{BB} \partial B_1(0)| 0\xlongleftrightarrow{BB} \partial B_\epsilon(0)\right]},
\end{equation}
where
\begin{equation*}
	\mathbb{P}\left[ 0\xlongleftrightarrow{BB} \partial B_1(0)| 0\xlongleftrightarrow{BB} \partial B_\epsilon(0)\right]:=\lim_{m\to \infty}\lim_{\eta\to 0} \mathbb{P}\left[ \partial B_{\delta_m}(0)\xlongleftrightarrow{BB}\partial B_{1}(0)|\mathcal{B}_{\eta,\epsilon}(0)\right]. 
\end{equation*}
Combining~\eqref{eqn::backbone_o1} with~\eqref{eqn::backbone_aux1}, the value of the limit in~\eqref{eqn::backbone_aux1} must be $\epsilon^{-\xi}$.

To see this, let $C_{BB}$ denote the limit in~\eqref{eqn::backbone_aux1} and note that, for any $n\geq1$, we can write
\begin{equation*}
    \mathbb{P}^{\epsilon^n}\left[0\xlongleftrightarrow{BB}\partial B_1(0)\right]=\frac{\mathbb{P}^{\epsilon^n}\left[0\xlongleftrightarrow{BB}\partial B_1(0)\right]}{\mathbb{P}^{\epsilon^n}\left[0\xlongleftrightarrow{BB}\partial B_{\epsilon}(0)\right]} \frac{\mathbb{P}^{\epsilon^{n-1}}\left[0\xlongleftrightarrow{BB}\partial B_1(0)\right]}{\mathbb{P}^{\epsilon^{n-1}}\left[0\xlongleftrightarrow{BB}\partial B_{\epsilon}(0)\right]} \ldots \frac{\mathbb{P}^{\epsilon}\left[0\xlongleftrightarrow{BB}\partial B_1(0)\right]}{1} \, ,
\end{equation*}
which implies
\begin{equation}\label{eqn::backbone_aux3}
    \frac{1}{n}\log\mathbb{P}^{\epsilon^n}\left[0\xlongleftrightarrow{BB}\partial B_1(0)\right] = \frac{1}{n}\sum_{k=1}^{n}\log\frac{\mathbb{P}^{\epsilon^k}\left[0\xlongleftrightarrow{BB}\partial B_1(0)\right]}{\mathbb{P}^{\epsilon^k}\left[0\xlongleftrightarrow{BB}\partial B_{\epsilon}(0)\right]} \, ,
\end{equation}
where we have conventionally set $\mathbb{P}^{\epsilon}\left[0\xlongleftrightarrow{BB}\partial B_{\epsilon}(0)\right]=1$.
Since $\epsilon^{n}\to 0$ as $n\to \infty$, \eqref{eqn::backbone_aux1} implies 
\begin{equation*}
    \lim_{n\to \infty} \frac{\mathbb{P}^{\epsilon^n}\left[0\xlongleftrightarrow{BB}\partial B_1(0)\right]}{\mathbb{P}^{\epsilon^n}\left[0\xlongleftrightarrow{BB}\partial B_{\epsilon}(0)\right]} = \frac{1}{C_{BB}} \, . 
\end{equation*}
Using this fact together with~\eqref{eqn::backbone_aux3}, the convergence of the Ces\`aro mean gives 
\begin{equation}\label{eqn::backbone_aux4}
    \lim_{n\to \infty} \frac{1}{n} \log \mathbb{P}^{\epsilon^n}\left[0\xlongleftrightarrow{BB}\partial B_1(0)\right]= -\log C_{BB}. 
\end{equation}
Comparing~\eqref{eqn::backbone_o1} with~\eqref{eqn::backbone_aux4}, we conclude that $C_{BB}=\epsilon^{-\xi}$.

Similarly, we have
\begin{equation} \label{eqn::bd_one_scal}
\lim_{a\to 0}\frac{\mathbb{P}^a\left[0\xlongleftrightarrow[\mathbb{H}]{B} \partial B_{\epsilon}(0)\right]}{\mathbb{P}^a\left[0\xlongleftrightarrow[\mathbb{H}]{B} \partial B_{1}(0)\right]}=\epsilon^{-\frac{1}{3}}, \quad \lim_{a\to 0} \frac{\mathbb{P}^a\left[0\xlongleftrightarrow{B}\partial B_{\epsilon}(0)\right]}{\mathbb{P}^a\left[0\xlongleftrightarrow{B} \partial B_1(0)\right]}=\epsilon^{-\frac{5}{48}}.
\end{equation}
We denote by $\mathbb{P}_{\mathbb{H}}$ the law of nested $\mathrm{CLE}_6$ in $\mathbb{H}$. Now, we are ready to prove Theorem~\ref{thm::backbone}.

\begin{proof}[Proof of Theorem~\ref{thm::backbone}]
A standard application of RSW estimates (see, e.g., the proofs of Lemmas~2.1 and~2.2 of~\cite{CamiaNewman2009ising}) implies that there exist constants $0<K_1<K_2<\infty$, independent of $a$, such that 
	\begin{equation} \label{eqn::RSW_1}
		\opi_a^{-2}\rho_a^{-1}\times \mathbb{P}^a \left[\{x_1^a\xlongleftrightarrow[\mathbb{H}]{B}z^a\}\circ \{x_2^a\xlongleftrightarrow[\mathbb{H}]{B}z^a\}\right]\in [K_1,K_2]
	\end{equation} 
which shows that all subsequential limits of the left-hand side of~\eqref{eqn::RSW_1} belong to $(0,\infty)$.  

We let $\epsilon \in (0,\frac{1}{100}\min\{|z-x_1|,|z-x_2|,|x_1-x_3|\} )$. Then, for small enough $a$, we can write 
\begin{align*}
	&\opi_a^{-2}\rho_a^{-1}\times \mathbb{P}^a \left[\{x_1^a\xlongleftrightarrow[\mathbb{H}]{B}z^a\}\circ \{x_2^a\xlongleftrightarrow[\mathbb{H}]{B}z^a\}\right] \\
	& \qquad = \underbrace{\mathbb{P}^a \left[\{x_1^a\xlongleftrightarrow[\mathbb{H}]{B}z^a\}\circ \{x_2^a\xlongleftrightarrow[\mathbb{H}]{B}z^a\}| z^a\xlongleftrightarrow{BB}\partial B_{\epsilon}(z^a), x_j^a\xlongleftrightarrow[\mathbb{H}]{B}\partial B_{\epsilon}(x_j^a),\enspace j\in \{1,2\}\right]}_{T_1} \\
	& \qquad \quad \times \underbrace{\left(\opi_a^{-2}\times \left(\mathbb{P}^a\left[0\xlongleftrightarrow[\mathbb{H}]{B}\partial B_{\epsilon}(0)\right]\right)^2\right)}_{T_2} \times \underbrace{\left(\rho_a^{-1}\times \mathbb{P}^a\left[0\xlongleftrightarrow{BB}\partial B_{\epsilon}(0)\right]\right) }_{T_3}.
\end{align*}
According to~\eqref{eqn::backbone_scal} and~\eqref{eqn::bd_one_scal}, we have 
\begin{equation*}
	\lim_{a\to 0}T_2= \epsilon^{-\frac{2}{3}},\quad \lim_{a\to 0} T_3=\epsilon^{-\xi}.
\end{equation*}
For $T_1$, with~\cite[Lemma~2.1]{Cam23} replaced by Lemma~\ref{lem::backbone}, we can proceed as in~\cite[Proof of Theorem~1.1]{Cam23} to show that 
\begin{align*}
	&\lim_{a\to 0} T_1 = \mathbb{P}_{\mathbb{H}}\left[ \{x_1\xlongleftrightarrow{B}z\}\circ \{x_2\xlongleftrightarrow{B}z\}| z\xlongleftrightarrow{BB}\partial B_{\epsilon}(z), x_j\xlongleftrightarrow{B}\partial B_{\epsilon}(x_j), \enspace j\in \{1,2\}\right]\\ 
	& \quad = \lim_{m\to \infty}\lim_{\eta\to 0}\mathbb{P}_{\mathbb{H}}\left[ \{\partial B_{\delta_m}(x_1)\xlongleftrightarrow{B} \partial B_{\delta_m}(z)\}\circ \{\partial B_{\delta_m}(x_2)\xlongleftarrow{B}\partial B_{\delta_m}(z)\}| \mathcal{B}_{\eta,\epsilon}(z), \mathcal{A}_{\eta,\epsilon}(x_j), \enspace j\in \{1,2\}\right].
\end{align*}
Combining all of these observations together, we get the existence of the limit
\begin{equation*}
	P(x_1,x_2,z):=\lim_{a\to 0} \opi_a^{-2}\rho_a^{-1}\times \mathbb{P}^a \left[\{x_1^a\xlongleftrightarrow[\mathbb{H}]{B}z^a\}\circ \{x_2^a\xlongleftrightarrow[\mathbb{H}]{B}z^a\}\right].
\end{equation*}

The desired conformal covariance property for $P(x_1,x_2,z)$ can be derived using essentially the same method as in~\cite[Proof of Theorems~1.1\&1.4]{Cam23}, with the help of~\eqref{eqn::backbone_scal} and~\eqref{eqn::bd_one_scal} to get the correct exponents. 
\end{proof}

\section{Operator product expansions and logarithmic correlations} \label{sec::3}
In this section, we prove Theorems~\ref{thm::expansion}, \ref{thm::4arm_expansion}, \ref{thm::three_arm} and \ref{thm::log2}.
\subsection{Logarithmic correlations of bulk fields: Proof of Theorem~\ref{thm::expansion} modulo a key lemma}
Fix $\epsilon\in (0,1)$ such that $\epsilon<\frac{\min\{|x_3-x_4|,|x-x_3|,|x-x_4|\}}{20}$.
First, we consider 
\begin{equation*}
P\left(x_1\xlongleftrightarrow{B}x_2, x_3\xlongleftrightarrow{B}x_4\right):= P(x_1\xlongleftrightarrow{B}x_2\xlongleftrightarrow{B}x_3\xlongleftrightarrow{B}x_4)+P(x_1\xlongleftrightarrow{B}x_2\centernot{\xlongleftrightarrow{B}}x_3\xlongleftrightarrow{B} x_4).
\end{equation*}
According to~\cite[Eq.~(1.3)]{Cam23}, there exists a universal constant $C_*\in (0,\infty)$ such that 
\begin{equation} \label{eqn::twopoint}
\begin{split}
    	P(x_1\xlongleftrightarrow{B}x_2):=\lim_{a\to 0} \pi_a^{-2}\times \mathbb{P}^2\left[x_1^a\xlongleftrightarrow{B}x_2^a\right]=C_*|x_2-x_1|^{-\frac{5}{24}},\\
     P(x_3\xlongleftrightarrow{B}x_4):=\lim_{a\to 0} \pi_a^{-2}\times \mathbb{P}^2\left[x_3^a\xlongleftrightarrow{B}x_4^a\right]=C_*|x_4-x_3|^{-\frac{5}{24}}.
\end{split}
\end{equation}
\begin{lemma} \label{lem::expansion_aux1}
		Assume the same setup as in Theorem~\ref{thm::expansion}. Then there exists a universal constant $C_{10}\in (0,\infty)$ such that 
	\begin{align}\label{eqn::expansion_aux1_g1}
    \begin{split}
	   & \lim_{x_1,x_2\to x}\frac{P\left(x_1\xlongleftrightarrow{B}x_2, x_3\xlongleftrightarrow{B}x_4\right)-P(x_1\xlongleftrightarrow{B}x_2)P(x_3\xlongleftrightarrow{B}x_4)}{|x_2-x_1|^{\frac{25}{24}}\times \left|\log\left|x_2-x_1\right|\right|} \\
      & \qquad = C_{10} |x-x_3|^{-\frac{5}{4}}|x-x_4|^{-\frac{5}{4}}|x_3-x_4|^{\frac{25}{24}}.
	\end{split}
    \end{align}
\end{lemma}

We postpone the proof of Lemma~\ref{lem::expansion_aux1}  to Section~\ref{sec::tech_log}. Next, we consider $P(x_1\xlongleftrightarrow{B}x_3\centernot{\xlongleftrightarrow{B}}x_2\xlongleftrightarrow{B}x_4)$ and $P(x_1\xlongleftrightarrow{B}x_4\centernot{\xlongleftrightarrow{B}}x_2\xlongleftrightarrow{B}x_3)$. In the limit $x_1,x_2\to0$, these two terms clearly behave in the same way, so it suffices to consider the first one. 

\begin{lemma} \label{lem::expansion_aux3}
		Assume the same setup as in Theorem~\ref{thm::expansion}. Then when $x_1,x_2$ are close enough to $x$, we have 
		\begin{equation*}
			P(x_1\xlongleftrightarrow{B}x_3\centernot{\xlongleftrightarrow{B}}x_2\xlongleftrightarrow{B}x_4)\leq 2C_8 \left(2 \frac{|x_2-x_1|}{\epsilon}\right)^{\frac{5}{4}}\times \epsilon^{-\frac{5}{24}}\times \left(\frac{|x_2-x_1|}{2}\right)^{-\frac{5}{24}},
		\end{equation*}
		where $C_8$ is the constant in Corollary~\ref{coro::four_arm}.
\end{lemma}
\begin{proof}
	Let $x_1,x_2$ be close enough to $x$ so that $\max\{|x_1-x|,|x_2-x|\} \leq \frac{\epsilon}{2} $. 
	Recall that
	\begin{equation} \label{eqn::expansion_aux31}
		P(x_1\xlongleftrightarrow{B}x_3\centernot{\xlongleftrightarrow{B}}x_2\xlongleftrightarrow{B}x_4)=\lim_{a\to 0} \pi_a^{-4} \times \mathbb{P}^a\left[x_1^a\xlongleftrightarrow{B}x_3^a\centernot{\xlongleftrightarrow{B}}x_2^a\xlongleftrightarrow{B}x_4^a\right]. 
	\end{equation}
	When $a>0$ is small enough, we have
	\begin{equation}\label{eqn::expansion_aux32}
\begin{split}
			&\{x_1^a\xlongleftrightarrow{B}x_3^a\centernot{\xlongleftrightarrow{B}}x_2^a\xlongleftrightarrow{B}x_4^a\}\\
	&\subseteq 	\{x_j^a\xlongleftrightarrow{B}\partial B_{\frac{\vert x_2^a-x_1^a \vert}{2}}(x_j^a), \enspace j=1,2\}\cap 	\{x_k^a\xlongleftrightarrow{B}\partial B_{\epsilon}(x_k^a), \enspace k=3,4\}\cap \mathcal{F}^a_{2\vert x_2^a-x_1^a \vert,\epsilon}\left(\frac{x_1^a+x_2^a}{2}\right).
\end{split}
	\end{equation}
Combining~\eqref{eqn::expansion_aux31},~\eqref{eqn::expansion_aux32} and Corollary~\ref{coro::four_arm}, and using the second limit in~\eqref{eqn::bd_one_scal} , gives the desired inequality.
\end{proof}
Theorem~\ref{thm::expansion} is a direct consequence of Lemmas~\ref{lem::expansion_aux1} and~\ref{lem::expansion_aux3}, as shown below.
\begin{proof}[Proof of Theorem~\ref{thm::expansion}]
Let $C_1=C_*^2$, where $C_*\in (0,\infty)$ is the constant in~\eqref{eqn::twopoint}. Note that 
\begin{equation} \label{eqn::expansion_proof}
	\begin{split}
	&	\frac{C(x_1,x_2,x_3,x_4)-C_1|x_3-x_4|^{-\frac{5}{24}}\times |x_2-x_1|^{-\frac{5}{24}}}{|x_2-x_1|^{\frac{25}{24}}\times \left| \log\left|x_2-x_1\right|\right|}\\
		=&\frac{P\left(x_1\xlongleftrightarrow{B}x_2, x_3\xlongleftrightarrow{B}x_4\right)-P(x_1\xlongleftrightarrow{B}x_2)P(x_3\xlongleftrightarrow{B}x_4)}{|x_2-x_1|^{\frac{25}{24}}\times \left|\log\left|x_2-x_1\right|\right|}+\frac{P\left(x_1\xlongleftrightarrow{B}x_3\centernot{\xlongleftrightarrow{B}}x_2\xlongleftrightarrow{B}x_4\right)}{|x_2-x_1|^{\frac{25}{24}}\times \left|\log\left|x_2-x_1\right|\right|}\\
		&+\frac{P\left(x_1\xlongleftrightarrow{B}x_4\centernot{\xlongleftrightarrow{B}}x_2\xlongleftrightarrow{B}x_3\right)}{|x_2-x_1|^{\frac{25}{24}}\times \left|\log\left|x_2-x_1\right|\right|} \to C_{10}|x-x_3|^{-\frac{5}{4}}|x-x_4|^{-\frac{5}{4}}|x_3-x_4|^{\frac{25}{24}}, \enspace \text{as }x_1,x_2\to x,
	\end{split}
\end{equation}
where we used Lemmas~\ref{lem::expansion_aux1} and~\ref{lem::expansion_aux3}, and $C_{10}\in (0,\infty)$ is the constant in Lemma~\ref{lem::expansion_aux1}.
This proves Theorem~\ref{thm::expansion} with $C_2=C_{10}/C_1$.
\end{proof}

\subsection{The origin of the logarithmic correction: Proof of Lemma~\ref{lem::expansion_aux1}} \label{sec::tech_log}
Without loss of generality, we may assume that $|x_2-x_1|$ is small enough in this section. Recall that $\epsilon\in (0,1)$ such that $\epsilon<\frac{\min\{|x_3-x_4|,|x-x_3|,|x-x_4|\}}{20}$ is fixed.

Let $M=\lfloor \log_2 \frac{\epsilon}{|x_2-x_1|}\rfloor$ and write $B^a_m=B_{2^m|x_2^a-x_1^a|}(\frac{x_1^a+x_2^a}{2})$, $B_m=B_{2^m|x_2-x_1|}(\frac{x_1+x_2}{2})$ for $m\in \{1,2,\ldots,M\}$.
Given two subsets of the plane, $C$ and $D$, and two vertices $z_1^a,z_2^a\in a\mathcal{T}$, we consider the following events:
\begin{itemize}
\item $\{z_1^a\xlongleftrightarrow{B;C}z_2^a\}$: there is a black path connecting $z_1^a$ to $z_2^a$ contained in $C$;
\item $\{z_1^a\xlongleftrightarrow[D]{B}z_2^a\}$: $z_1^a$ and $z_2^a$ belong to the same black cluster but there is no black path fully contained in $D$;
\item $\{z_1^a\xlongleftrightarrow[D]{B;C}z_2^a\}$: there is a black path connecting $z_1^a$ to $z_2^a$ contained in $C$ but no black path fully contained in $D$. 
\end{itemize}
We also let $(B_m^a)^c:=\mathbb{C}\setminus B_m^a, B_m^c:=\mathbb{C}\setminus B_m$ and use the convention that 
\[B_0^a=B_0:=\emptyset,\quad  B_{M+1}^a=B_{M+1}:=\mathbb{C}. \]
Using the strategy in~\cite[Proof of Theorem~1.5]{Cam23}, one can show that for $m\in \{1,2,\ldots,M,M+1\}$, the following limits exist and belong to $(0,\infty)$:
\begin{align*}
	P\left(x_1\xlongleftrightarrow[B_{m-1}]{B;B_m}x_2\right):=&\lim_{a\to 0} \pi_a^{-2}\times\mathbb{P}\left[x_1^a\xlongleftrightarrow[B_{m-1}]{B;B_{m}}x_2^a\right],\\
	P\left(x_1\xlongleftrightarrow[B_{m-1}]{B;B_{m}}x_2, x_3\xlongleftrightarrow[B_m^c]{B}x_4\right):=&\lim_{a\to 0} \pi_a^{-4}\times \mathbb{P}^a\left[x_1^a\xlongleftrightarrow[B_{m-1}^a]{B;B_{m}^a}x_2^a, x_3^a\xlongleftrightarrow[(B_m^a)^c]{B}x_4^a\right],\\
	P\left(x_1\xlongleftrightarrow[B_M]{B}x_2, x_3\xlongleftrightarrow{B}x_4\right):=&\lim_{a\to 0} \pi_a^{-4}\times \mathbb{P}^a\left[x_1^a\xlongleftrightarrow[B_M^a]{B}x_2^a, x_3^a\xlongleftrightarrow{B}x_4^a\right].
\end{align*}

\begin{lemma}\label{lem::expansion_aux11}
	We have 
	\begin{equation} \label{eqn::decomp}
		\begin{split}
			& P(x_1\xlongleftrightarrow{B}x_2,x_3\xlongleftrightarrow{B}x_4)- P\left(x_1\xlongleftrightarrow{B}x_2\right)P\left(x_3\xlongleftrightarrow{B}x_4\right) \\
			& \quad = \underbrace{P\left(x_1\xlongleftrightarrow{B;B_1}x_2,x_3\xlongleftrightarrow[B_1^c]{B}x_4\right)-P\left(x_1\xlongleftrightarrow{B;B_1}x_2\right)P\left(x_3\xlongleftrightarrow[B_1^c]{B}x_4\right)}_{T_1} \\
			& \qquad + \sum_{m=2}^M \underbrace{P\left(x_1\xlongleftrightarrow[B_{m-1}]{B;B_m}x_2,x_3\xlongleftrightarrow[B_m^c]{B}x_4\right)-P\left(x_1\xlongleftrightarrow[B_{m-1}]{B;B_m}x_2\right) P\left(x_3\xlongleftrightarrow[B_m^c]{B}x_4\right)}_{T_m} \\
			& \qquad + \underbrace{P\left(x_1\xlongleftrightarrow[B_M]{B}x_2,x_3\xlongleftrightarrow{B}x_4\right)-P\left(x_1\xlongleftrightarrow[B_M]{B}x_2\right)P\left(x_3\xlongleftrightarrow{B}x_4\right)}_{T_{M+1}}.
		\end{split}
	\end{equation}
\end{lemma}
\begin{proof}
	First of all, we observe that
	\begin{align*}
		& \mathbb{P}^a\left[x_1^a\xlongleftrightarrow{B}x_2^a,x_3^a\xlongleftrightarrow{B}x_4^a\right] = \mathbb{P}^a\left[x_1^a\xlongleftrightarrow{B;B_1^a}x_2^a,x_3^a\xlongleftrightarrow{B}x_4^a\right]+\mathbb{P}^a\left[x_1^a\xlongleftrightarrow[B_M^a]{B}x_2^a,x_3^a\xlongleftrightarrow{B}x_4^a\right] \\
		& \qquad \qquad \qquad \qquad \qquad +\sum_{m=2}^M \mathbb{P}^a\left[x_1^a\xlongleftrightarrow[B_{m-1}^a]{B;B_m^a}x_2^a,x_3^a\xlongleftrightarrow{B}x_4^a\right], \\
		& \mathbb{P}^a\left[x_1^a\xlongleftrightarrow{B}x_2^a\right]\mathbb{P}^a\left[x_3^a\xlongleftrightarrow{B}x_4^a\right] = \mathbb{P}^a\left[x_1^a\xlongleftrightarrow{B;B_1^a}x_2^a\right]\mathbb{P}^a\left[x_3^a\xlongleftrightarrow{B}x_4^a\right] \\
        & \qquad \qquad \qquad \qquad \qquad + \mathbb{P}^a\left[x_1^a\xlongleftrightarrow[B_M^a]{B}x_2^a\right]\mathbb{P}^a\left[x_3^a\xlongleftrightarrow{B}x_4^a\right] + \sum_{m=2}^M \mathbb{P}^a\left[x_1^a\xlongleftrightarrow[B_{m-1}^a]{B;B_m^a} x_2^a\right]\mathbb{P}^a\left[ x_3^a\xlongleftrightarrow{B}x_4^a\right].
	\end{align*}
	Now note that 
	\begin{align*}
		& \mathbb{P}^a\left[x_1^a\xlongleftrightarrow{B;B_1^a}x_2^a,x_3^a\xlongleftrightarrow{B}x_4^a\right]-\mathbb{P}^a\left[x_1^a\xlongleftrightarrow{B;B_1^a}x_2^a\right]\mathbb{P}^a\left[x_3^a\xlongleftrightarrow{B}x_4^a\right] \\
		& \quad = \mathbb{P}^a\left[x_1^a\xlongleftrightarrow{B;B_1^a}x_2^a,x_3^a\xlongleftrightarrow[(B_1^a)^c]{B}x_4^a\right]- \mathbb{P}^a\left[x_1^a\xlongleftrightarrow{B;B_1^a}x_2^a\right]\mathbb{P}^a\left[x_3^a\xlongleftrightarrow[(B_1^a)^c]{B}x_4^a\right],
	\end{align*}
	where we used the fact that 
	\begin{equation} \label{eqn::independence_log}
		\mathbb{P}^a\left[x_1^a\xlongleftrightarrow{B;B_1^a}x_2^a,x_3^a\xlongleftrightarrow{B;(B_1^a)^c}x_4^a\right]= \mathbb{P}^a\left[x_1^a\xlongleftrightarrow{B;B_1^a}x_2^a\right]\mathbb{P}^a\left[x_3^a\xlongleftarrow{B;(B_1^a)^c}x_4^a\right].
	\end{equation}
	Similarly, we have 
	\begin{align*}
& \mathbb{P}^a\left[x_1^a\xlongleftrightarrow[B_{m-1}^a]{B;B_m^a},x_3^a\xlongleftrightarrow{B}x_4^a\right]- \mathbb{P}^a\left[x_1^a\xlongleftrightarrow[B_{m-1}^a]{B;B_m^a}\right]\mathbb{P}^a\left[x_3^a\xlongleftrightarrow{B}x_4^a\right] \\
& \quad = \mathbb{P}^a\left[x_1^a\xlongleftrightarrow[B_{m-1}^a]{B;B_m^a},x_3^a\xlongleftrightarrow[(B_m^a)^c]{B}x_4^a\right]- \mathbb{P}^a\left[x_1^a\xlongleftrightarrow[B_{m-1}^a]{B;B_m^a}\right]\mathbb{P}^a\left[x_3^a\xlongleftrightarrow[(B_m^a)^c]{B}x_4^a\right],
	\end{align*}
	for $m\in \{2,3,\ldots,M\}$. 
	Combining these observations together and letting $a\to 0$ gives~\eqref{eqn::decomp}.
\end{proof}

\begin{lemma}\label{lem::expansion_aux12}
	Consider the terms $T_1$ and $T_{M+1}$ defined in~\eqref{eqn::decomp}. Then there exists a constant $\hat{C}\in (0,\infty)$ that may depend on $\epsilon$ such that when $x_1,x_2$ are close enough to $x$, we have 
	\begin{equation*}
		|T_1|+|T_{M+1}|\leq \hat{C} |x_2-x_1|^{\frac{25}{24}}.
	\end{equation*}
\end{lemma}
\begin{proof}
		Let $x_1,x_2$ be close enough to $x$ so that $\max\{|x_1-x|,|x_2-x|\} \leq \frac{\epsilon}{2} $.
        Then, when $a>0$ is small enough, we have 
		\begin{align*}
			\{x_3^a\xlongleftrightarrow[(B_1^a)^c]{B}x_4^a\}\subseteq & \; \mathcal{F}_{2|x_2^a-x_1^a|,\epsilon}\left(\frac{x_1^a+x_2^a}{2}\right)\cap \{x_k^a\xlongleftrightarrow{B}\partial B_{\epsilon}(x_k^a), \enspace k\in \{3,4\}\},\\
			 \{x_1^a\xlongleftrightarrow[B_M^a]{B}x_2^a\}\subseteq & \; \mathcal{F}_{2|x_2^a-x_1^a|,\epsilon}\left(\frac{x_1^a+x_2^a}{2}\right)\cap \{x_j^a\xlongleftrightarrow{B} \partial B_{\frac{|x_2^a-x_1^a|}{2}}(x_j^a),\enspace j\in\{1,2\}\}.
		\end{align*}
		Letting $a\to 0$ gives 
		\begin{equation*} \label{eqn::esti_M01_aux1}
			|T_1|+|T_{M+1}|\leq 4 \left(\frac{|x_2-x_1|}{2}\right)^{-\frac{5}{24}}\times \epsilon^{-\frac{5}{24}}\times \mathbb{P}\left[\mathcal{F}_{2|x_2-x_1|,\epsilon}(\frac{x_1+x_2}{2})\right].
		\end{equation*}
		Combining the last inequality with Corollary~\ref{coro::four_arm} gives the desired inequality with some constant $\hat{C}\in (0,\infty)$, which may depend on $\epsilon$. 
\end{proof}
Let $m\in \{2,\ldots,M\}$. We define $\mathcal{F}^{(o,a)}_m$ to be the following event: (1) there is no black path in $(B_m^a)^c$ connecting $x_3^a$ to $x_4^a$; and (2) there are two black paths connecting $x_3^a$ and $x_4^a$ to $\partial B_m^a$, respectively. We define $\mathcal{F}^{(i,a)}_m$ to be the following event: (1) there is no black path in $B_{m}^a$ connecting $x_1^a$ to $x_2^a$; and (2) there are two black paths connecting $x_1^a$ and $x_2^a$ to $\partial B_m^a$, respectively.
Furthermore, we let
\begin{align*}
\hat{\mathcal{F}}_m^{(o,a)}:=& \; \mathcal{F}^a_{2^m|x_2^a-x_1^a|,\epsilon}\left(\frac{x_1^a+x_2^a}{2}\right)\cap \{x_k^a\xlongleftrightarrow{B}\partial B_{\epsilon}(x_k^a),\enspace k\in \{3,4\}\},\\
\hat{\mathcal{F}}_m^{(i,a)}:=& \; \mathcal{F}^a_{2|x^a_2-x_1^a|,2^m|x_2^a-x_1^a|}\left(\frac{x_1^a+x_2^a}{2}\right)\cap \{x_j^a\xlongleftrightarrow{B}\partial B_{\frac{|x_2^a-x_1^a|}{2}}(x_j^a),\enspace j\in \{1,2\}\}.
\end{align*}
See Fig.~\ref{fig::fourarmvariants1} for an illustration of these events.

\begin{figure} 
\begin{subfigure}{1\textwidth}
		\begin{center}
			\includegraphics[width=0.5\textwidth]{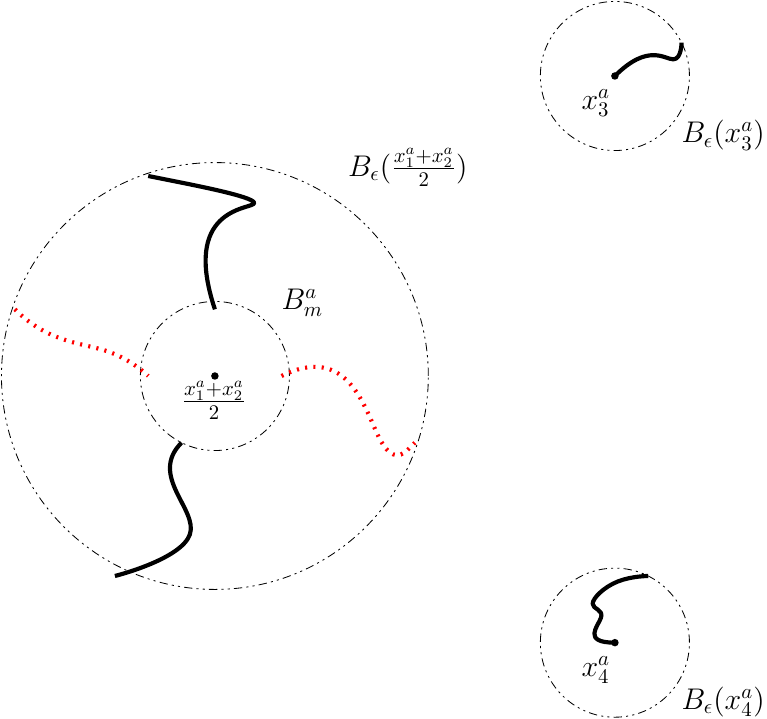}
		\end{center}
		\caption{The event $\hat{\mathcal{F}}_{m}^{(o,a)}$.}
	\end{subfigure}
	\begin{subfigure}{0.43\textwidth}
		\begin{center}
			\includegraphics[width=0.65\textwidth]{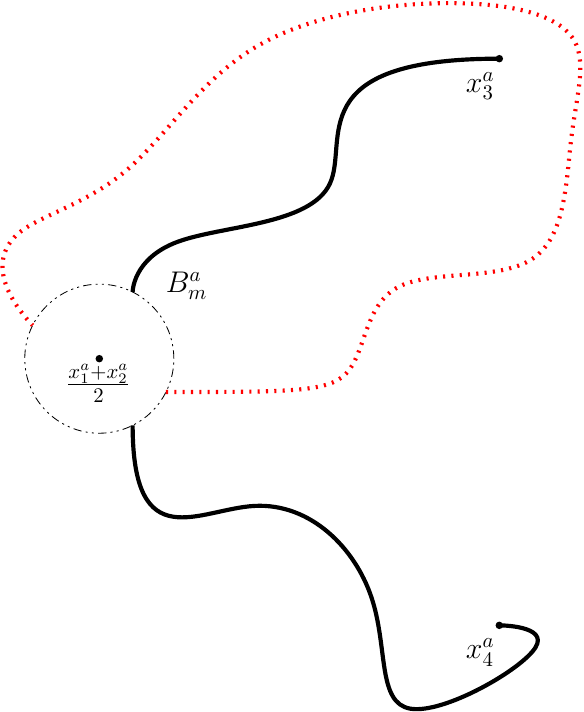}
		\end{center}
		\caption{The event $\mathcal{F}_m^{(o,a)}$.}
	\end{subfigure}
	\begin{subfigure}{0.43\textwidth}
		\begin{center}
			\includegraphics[width=0.75\textwidth]{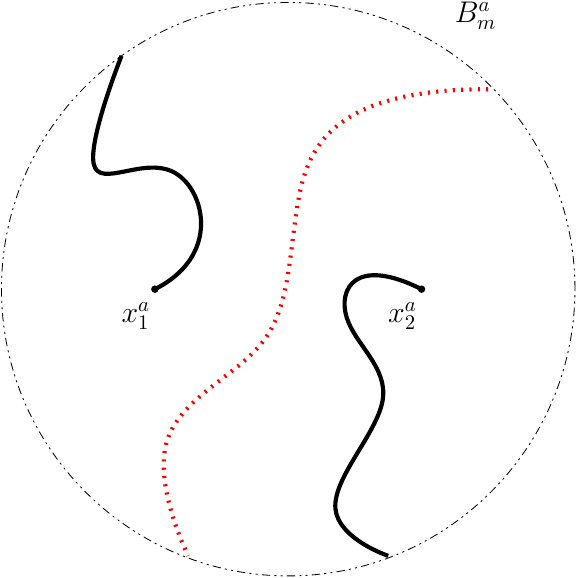}
		\end{center}
		\caption{The event $\mathcal{F}_{m}^{(i,a)}$.}
	\end{subfigure}
	\begin{subfigure}{0.45\textwidth}
		\begin{center}
			\includegraphics[width=0.65\textwidth]{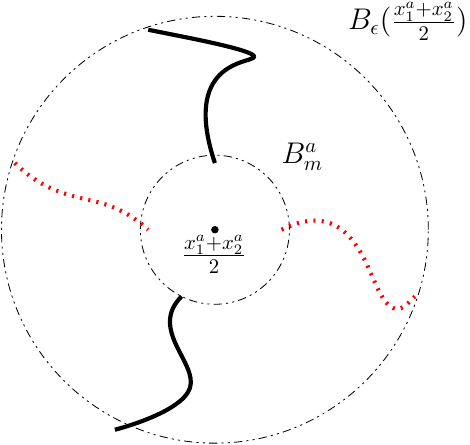}
		\end{center}
		\caption{The event $\mathcal{F}_{2^m|x_2^a-x_1^a|,\epsilon}\big( \frac{x_1^a+x_2^a}{2}\big)$.}
	\end{subfigure}
	\begin{subfigure}{0.45\textwidth}
		\begin{center}
			\includegraphics[width=0.6\textwidth]{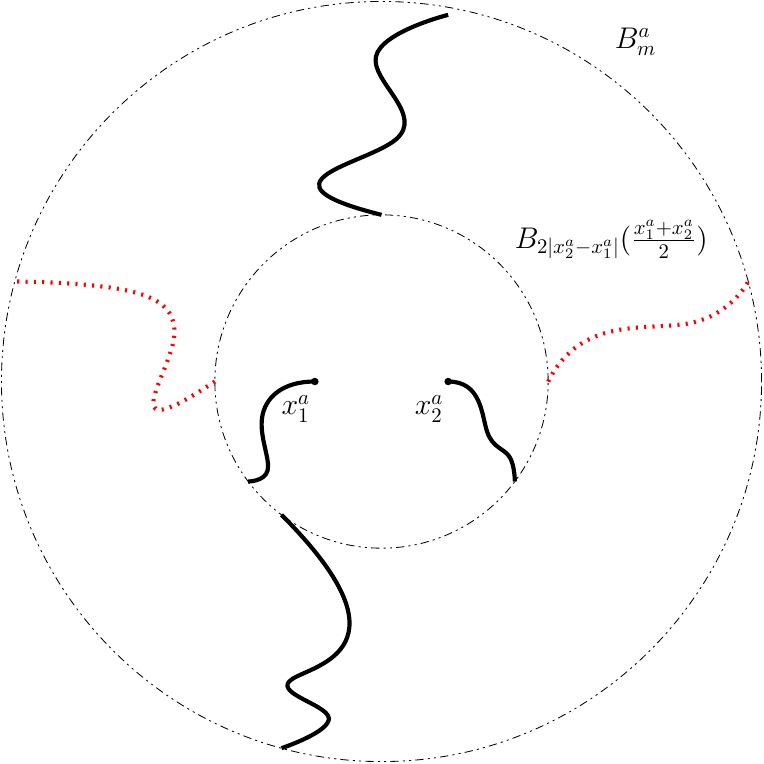}
		\end{center}
		\caption{The event $\hat{\mathcal{F}}_{m}^{(i,a)}$.}
	\end{subfigure}
	\caption{The events $\hat{\mathcal{F}}_m^{(o,a)}$, $\mathcal{F}_m^{(o,a)}$, $\mathcal{F}_{m}^{(i,a)}$, $\mathcal{F}_{2^m|x_2^a-x_1^a|,\epsilon}\big(\frac{x_1^a+x_2^a}{2}\big)$ and $\hat{\mathcal{F}}_m^{(i,a)}$. The black, solid lines represent black paths, and the red, dotted lines represent white paths.}
	\label{fig::fourarmvariants1}
\end{figure}

It follows from~\eqref{eqn::bd_one_scal} and Corollary~\ref{coro::four_arm} that 
\begin{equation}\label{eqn::asy_modif_four_arm}
\begin{split}
		\lim_{x_1,x_2\to x} \left(\frac{2^m|x_2-x_1|}{\epsilon}\right)^{-\frac{5}{4}}\times\lim_{a\to 0} \pi_a^{-2}\times  \mathbb{P}^a\left[\hat{\mathcal{F}}^{(o,a)}_m\right] = & \; C_8 \, \epsilon^{-\frac{5}{24}},\\
		\quad \lim_{x_1,x_2\to x} 2^{\frac{5}{4}(m-2)}\times \left(\frac{x_2-x_1}{2}\right)^{\frac{5}{24}}\times \lim_{a\to 0} \pi_a^{-2}\times \mathbb{P}^a\left[\hat{\mathcal{F}}^{(i,a)}_{m-1}\right] = & \; C_8,
\end{split}
\end{equation}
where $C_8$ is the constant in Corollary~\ref{coro::four_arm}. 
Note that when $x_1^a,x_2^a$ are close enough to $x$, we have 
\begin{equation} \label{eqn::F_inclusion_1}
\{x_3^a\xlongleftrightarrow[(B_m^a)^c]{B}x_4^a\}\subseteq \mathcal{F}^{(o,a)}_m\subseteq\hat{\mathcal{F}}^{(o,a)}_m,
\end{equation}
\begin{equation}
 \{x_1^a\xlongleftrightarrow[B_m^a]{B;B_{m+1}^a}x_2^a\}\subseteq \mathcal{F}^{(i,a)}_{m}\subseteq \hat{\mathcal{F}}_m^{(i,a)}.
\end{equation}
For $m\in \{2,3,\ldots, M\}$, using the strategy in~\cite[Proof of Theorem~1.5]{Cam23}, one  can show that the following limits exist and belong to $(0,1)$:
\begin{align*}
	P\left(\mathcal{F}^o_m| \hat{\mathcal{F}}_m^{o}\right) := & \; \lim_{a\to 0} \mathbb{P}^a\left[\mathcal{F}^{(o,a)}_m| \hat{\mathcal{F}}_m^{(o,a)}\right],\\
P\left(x_3\xlongleftrightarrow[B_m^c]{B}x_4| \mathcal{F}^{o}_m\right) := & \; \lim_{a\to 0}  \mathbb{P}^a\left[x_3^a\xlongleftrightarrow[(B_m^a)^c]{B}x_4^a| \mathcal{F}^{(o,a)}_m\right],\\
P\left(x_3\xlongleftrightarrow[B_m^c]{B}x_4| x_1\xlongleftrightarrow[B_{m-1}]{B;B_m}x_2,\mathcal{F}_m^o\right) := & \; \lim_{a\to 0} \mathbb{P}^a\left[x_3^a\xlongleftrightarrow[(B_m^a)^c]{B}x_4^a| x_1^a\xlongleftrightarrow[(B_{m-1}^a)^c]{B;B_m^a}x_2^a,\mathcal{F}_m^{(o,a)}\right],\\
P\left(x_3\xlongleftrightarrow[B_m^c]{B}x_4| \mathcal{F}_{m-1}^i,\mathcal{F}_m^o\right) := & \; \lim_{a\to 0} \mathbb{P}^a\left[x_3^a\xlongleftrightarrow[(B_m^a)^c]{B}x_4^a| \mathcal{F}_{m-1}^{(i,a)},\mathcal{F}_m^{(o,a)}\right],\\
P\left(x_1\xlongleftrightarrow[B_{m-1}]{B;B_m}x_2|\hat{\mathcal{F}}^i_{m-1}\right) := & \; \lim_{a\to 0} \mathbb{P}^a\left[x_1^a\xlongleftrightarrow[B_{m-1}^a]{B;B_m^a}x_2^a|\hat{\mathcal{F}}^{(i,a)}_{m-1}\right].
\end{align*}

\begin{lemma} \label{lem::expansion_aux13}
There are universal constants $\tilde{C},\tilde{c}\in (0,\infty)$ such that the following holds. There exist universal constants $V_2,V_3,V_4\in (0,\infty)$ and a function $V_1(x,x_3,x_4,\epsilon)$ taking values in $(0,\infty)$ such that for any $\delta\in(0,1/2)$, any $y\in [\delta,1-\delta]$ and $m_y\in \{\lfloor y\cdot M\rfloor, \lceil y\cdot M\rceil\}$, we have 
\begin{align}
	\left| P\left(\mathcal{F}_{m_y}^o|\hat{\mathcal{F}}_{m_y}^o\right)-V_1(x,x_3,x_4,\epsilon)\right| \leq & \; \tilde{C} \left(\frac{|x_2-x_1|}{\epsilon}\right)^{\tilde{c}\delta}, \label{eqn::expansion_aux131}\\
	\left| P\left(x_3\xlongleftrightarrow[B_{m_y}^c]{B}x_4| \mathcal{F}^{o}_{m_y}\right)-V_2\right| \leq & \; \tilde{C} \left(\frac{|x_2-x_1|}{\epsilon}\right)^{\tilde{c}\delta},\label{eqn::expansion_aux132}\\
	\left| P\left(x_3\xlongleftrightarrow[B_{m_y}^c]{B}x_4| x_1\xlongleftrightarrow[B_{m_y-1}]{B;B_{m_y}}x_2,\mathcal{F}_{m_y}^o\right)-V_3\right| \leq & \; \tilde{C} \left(\frac{|x_2-x_1|}{\epsilon}\right)^{\tilde{c}\delta},\\
	\left|P\left(x_1\xlongleftrightarrow[B_{m_y-1}]{B;B_{m_y}}x_2|\hat{\mathcal{F}}^i_{m_y-1}\right)-V_4\right| \leq & \; \tilde{C} \left(\frac{|x_2-x_1|}{\epsilon}\right)^{\tilde{c}\delta}.\label{eqn::expansion_aux133}
\end{align}
\end{lemma}
\begin{lemma}\label{lem::expansion_aux14}
	The constants $V_2,V_3$ in Lemma~\ref{lem::expansion_aux13} satisfy 
	\begin{equation}\label{eqn::lem::expansion_aux141}
		V_3>V_2=\frac{1}{2}.
	\end{equation}
\end{lemma}

We postpone the proof of Lemmas~\ref{lem::expansion_aux13} and~\ref{lem::expansion_aux14} to the end of this section.
First, we prove Lemma~\ref{lem::expansion_aux1} with the help of Lemmas~\ref{lem::expansion_aux11}-\ref{lem::expansion_aux14}.
\begin{proof}[Proof of Lemma~\ref{lem::expansion_aux1}]

For each $m\in \{2,3,\ldots,M\}$, we write 
	\begin{align*}
		T_m = & \; P\left(x_1\xlongleftrightarrow[B_{m-1}]{B;B_m}x_2,x_3\xlongleftrightarrow[B_m^c]{B}x_4\right)-P\left(x_1\xlongleftrightarrow[B_{m-1}]{B;B_m}x_2\right) P\left(x_3\xlongleftrightarrow[B_m^c]{B}x_4\right)\\
		= & \; \underbrace{P\left(\mathcal{F}_m^o|\hat{\mathcal{F}}_m^o\right)}_{T_{m,1}}\times\underbrace{ \lim_{a\to 0} \pi_a^{-2}\times \mathbb{P}^a\left[\hat{\mathcal{F}}^{(o,a)}_m\right]}_{T_{m,2}}\times \underbrace{P\left(x_1\xlongleftrightarrow[B_{m-1}]{B;B_{m}}x_2|\hat{\mathcal{F}}^i_{m-1}\right)}_{T_{m,3}}\times \underbrace{\lim_{a\to 0} \pi_a^{-2} \times \mathbb{P}^a\left[\hat{\mathcal{F}}^{(i,a)}_m\right]}_{T_{m,4}}\\
		& \times \left(\underbrace{P\left(x_3\xlongleftrightarrow[B_{m}^c]{B}x_4| x_1\xlongleftrightarrow[B_{m-1}]{B;B_{m}}x_2,\mathcal{F}_{m}^o\right)}_{T_{m,5}}-\underbrace{P\left(x_3\xlongleftrightarrow[B_m^c]{B}x_4| \mathcal{F}^{o}_{m}\right)}_{T_{m,6}}\right).
	\end{align*}
Thanks to~\eqref{eqn::asy_modif_four_arm}, we have 
	\begin{equation*}
		\lim_{x_1,x_2\to x} \frac{T_{m,2}\times T_{m,4}}{|x_2-x_1|^{\frac{25}{24}}} =  C_8^2 \, 2^{\frac{65}{24}} \, \epsilon^{-\frac{35}{24}}.
	\end{equation*}
		Let $\delta\in (0,1/2)$.  On the one hand, since $T_{m,j}\in (0,1)$ for $j\in \{1,3,5,6\}$, we have 
		\begin{equation}\label{eqn::expansion_aux1_g2}
			\limsup_{x_1,x_2\to x} \left|\frac{\sum_{m\in \{2,\ldots,M\}\setminus [\delta M, (1-\delta)M]}T_m}{|x_2-x_1|^{\frac{25}{24}}\times |\log|x_2-x_1||}\right| \leq 2 \,\vert \log_2 \epsilon \vert \, \delta \, C_8^2 \, 2^{\frac{65}{24}} \, \epsilon^{-\frac{35}{24}}/\log 2 \, .
		\end{equation}  
	On the other hand, thanks to Lemma~\ref{lem::expansion_aux13}, we have 
	\begin{equation} \label{eqn::expansion_aux1_g3}
		\lim_{x_1,x_2\to x}\frac{\sum_{m\in \{2,3,\ldots,M\}\cap [\delta M,(1-\delta)M]}T_m}{|x_2-x_1|^{\frac{25}{24}}\times |\log|x_2-x_1||}= \frac{(1-2\delta)}{\log 2}  \vert \log_2 \epsilon \vert \, C_8^2 \, 2^{\frac{65}{24}} \, \epsilon^{-\frac{35}{24}} V_1(x,x_3,x_4,\epsilon)V_4\left(V_3-V_2\right).
	\end{equation}
	Combining Lemmas~\ref{lem::expansion_aux11},~\ref{lem::expansion_aux12},~\eqref{eqn::expansion_aux1_g2} with~\eqref{eqn::expansion_aux1_g3} and letting $\delta\to 0$ gives 
	\begin{equation} \label{eqn::expansion_aux1_g4}
\begin{split}
			G(x,x_3,x_4) := & \; \lim_{x_1,x_2\to x} \frac{P\left(x_1\xlongleftrightarrow{B}x_2,x_3\xlongleftrightarrow{B}x_4\right)-P\left(x_1\xlongleftrightarrow{B}x_2\right)P\left(x_3\xlongleftrightarrow{B}x_4\right)}{|x_2-x_1|^{\frac{25}{24}}\times \left|\log\left|x_2-x_1\right|\right|} \\
	        = & \; \vert \log_2 \epsilon \vert \, C_8^2 \, 2^{\frac{65}{24}} \, \epsilon^{-\frac{35}{24}} V_1(x,x_3,x_4,\epsilon)V_4\left(V_3-V_2\right)/\log2.
\end{split}
	\end{equation}
	Thanks to Lemma~\ref{lem::expansion_aux14}, we have $V_3-V_2>0$, which implies 
	\begin{equation}\label{eqn::expansion_aux1_g6}
		G(x,x_3,x_4)>0.
	\end{equation}
	
	 According to~\cite[Theorem~1.5]{Cam23}, for any non-constant M\"obius transformation $\varphi: \mathbb{C}\cup \{\infty\} \to \mathbb{C}\cup \{\infty\}$ with $\varphi(x), \varphi(x_3),\varphi(x_4)\neq \infty$, we have (when $x_1, x_2$ are close enough to $x$)
	 \begin{equation} \label{eqn::expansion_aux1_g5}
\begin{split}
	&P\left(\varphi(x_1)\xlongleftrightarrow{B}\varphi(x_2),\varphi
	(x_3)\xlongleftrightarrow{B}\varphi(x_4)\right)-P\left(\varphi(x_1)\xlongleftrightarrow{B}\varphi(x_2)\right)P\left(\varphi(x_3)\xlongleftrightarrow{B}\varphi(x_4)\right)\\
	& \quad = \left(P\left(x_1\xlongleftrightarrow{B}x_2,x_3\xlongleftrightarrow{B}x_4\right)-P\left(x_1\xlongleftrightarrow{B}x_2\right)P\left(x_3\xlongleftrightarrow{B}x_4\right)\right)\times \prod_{j=1}^4 |\varphi'(x_j)|^{-\frac{5}{48}}. 
\end{split}
	 \end{equation}
	 Combining~\eqref{eqn::expansion_aux1_g4} with~\eqref{eqn::expansion_aux1_g6} and~\eqref{eqn::expansion_aux1_g5}, we have 
	 \begin{align*}
	 & G(\varphi(x),\varphi(x_3),\varphi(x_4))\\
	 & =\lim_{x_1,x_2\to x} \frac{P\left(\varphi(x_1)\xlongleftrightarrow{B}\varphi(x_2),\varphi
	 		(x_3)\xlongleftrightarrow{B}\varphi(x_4)\right)-P\left(\varphi(x_1)\xlongleftrightarrow{B}\varphi(x_2)\right)P\left(\varphi(x_3)\xlongleftrightarrow{B}\varphi(x_4)\right)}{|\varphi(x_2)-\varphi(x_1)|^{\frac{25}{24}}\times |\log|\varphi(x_2)-\varphi(x_1)||}\\
	 		& = G(x,x_3,x_4)\times |\varphi'(x)|^{-\frac{5}{4}}\times \prod_{j=3}^4 |\varphi'(x_j)|^{-\frac{5}{48}}>0,
	 \end{align*}
	 which gives~\eqref{eqn::expansion_aux1_g1} with some $C_{10}\in (0,\infty)$, as we set out to prove.
\end{proof}
\begin{proof}[Proof of Lemma~\ref{lem::expansion_aux13}]
We only sketch the proof, as the key argument here is very similar to that in the proof of Lemma~\ref{lem::four_aux2}.  We start with~\eqref{eqn::expansion_aux131}. To keep track of the influence of $x_1$ and $x_2$ on $P\left(\mathcal{F}_m^o|\hat{\mathcal{F}}_m^o\right)$, we write 
\begin{equation*}
	P^{(x_1,x_2)}\left(\mathcal{F}_m^o|\hat{\mathcal{F}}_m^o\right)=P\left(\mathcal{F}_m^o|\hat{\mathcal{F}}_m^o\right).
\end{equation*}	
	
Let $m_y\in \{\lfloor y\cdot M\rfloor, \lceil y\cdot M\rceil\}$. Choose a sequence $\{(x_1^{(k)},x_2^{(k)})\}_{k=1}^{\infty}$ such that $x_1^{(k)}\neq x_2^{(k)}$ and that $\lim_{k\to \infty}x_1^{(k)}=\lim_{k\to \infty} x_2^{(k)}=x$. Then, one can proceed as in the proof of Lemma~\ref{lem::four_aux2} to show that there exist universal constants $\tilde{C},\tilde{c}\in (0,\infty)$ such that, for any $K\geq 1$ and $k_1,k_2\geq K$, we have 
\begin{equation*}
	\left| P^{(x_1^{(k_1)},x_2^{(k_1)})}\left(\mathcal{F}_{m_y}^{o}| \hat{\mathcal{F}}_{m_y}^o\right)-P^{(x_1^{(k_2)},x_2^{(k_2)})}\left(\mathcal{F}_{m_y}^{o}| \hat{\mathcal{F}}_{m_y}^o\right)\right| \leq \tilde{C}\left(\frac{\sup_{k\geq K}|x_1^{(k)}-x_2^{(k)}|}{\epsilon}\right)^{\tilde{c}\delta},
\end{equation*}
which gives~\eqref{eqn::expansion_aux131} with some function $V_1(x,x_3,x_4,\epsilon)$ taking values on $[0,1]$. A standard application of RSW estimates (see, e.g., the proofs of Lemmas~2.1 and~2.2 of~\cite{CamiaNewman2009ising}) shows that $V_1(x,x_3,x_4,\epsilon)>0$, as desired. 

The proof of~\eqref{eqn::expansion_aux132}-\eqref{eqn::expansion_aux133} is similar.
One additional comment is that the arguments in the proof of Lemma~\ref{lem::four_aux2} imply that the limits $V_j$ for $2\leq j\leq 4$ in~\eqref{eqn::expansion_aux132}-\eqref{eqn::expansion_aux133} do not depend on $x,x_3,x_4$ or on $\epsilon$. 
\end{proof}

Finally, we prove~Lemma~\ref{lem::expansion_aux14}.
Let $m\in \{2,\ldots, M\}$ and define ${E}_{b}^{a}$ to be the following event: the event $\mathcal{F}_{m-1}^{(i,a)} \cap \mathcal{F}_{m}^{(o,a)}$ occurs and either (1) the two black paths in the definition of $\mathcal{F}^{(o,a)}_m$ are in the same black cluster or (2) {each black path in the definition of $\mathcal{F}^{(o,a)}_m$ is connected by a black path to one of the black paths in the definition of $\mathcal{F}_{m-1}^{(i,a)}$.}
{Furthermore, let ${E}_{w}^{a}$ denote the event that $\mathcal{F}_{m-1}^{(i,a)} \cap \mathcal{F}_{m}^{(o,a)}$ occurs and the white cluster in the definition of $\mathcal{F}_{m}^{(o,a)}$ continues inside $B^a_m$ in such a way that the two black paths in the definition of $\mathcal{F}_{m}^{(o,a)}$ belong to different clusters (i.e., they are not connected inside $B^a_m$).}
Using the strategy in~\cite[Proof of Theorem~1.5]{Cam23}, one  can show that the following limits exist and belong to $(0,1)$:
\begin{align*}
	P\left(E_b | \mathcal{F}_{m-1}^{i}, \mathcal{F}_m^o\right) := & \; \lim_{a\to 0} \mathbb{P}^a\left[E_b^a| \mathcal{F}_{m-1}^{(i,a)}, \mathcal{F}_m^{(o,a)}\right],\\
		P\left(E_w| \mathcal{F}_{m-1}^{i}, \mathcal{F}_m^o\right) := & \; \lim_{a\to 0} \mathbb{P}^a\left[E_w^a| \mathcal{F}_{m-1}^{(i,a)}, \mathcal{F}_m^{(o,a)}\right].
\end{align*}

\begin{figure} 
\begin{subfigure}{0.45\textwidth}
		\begin{center}
\includegraphics[width=1\textwidth]{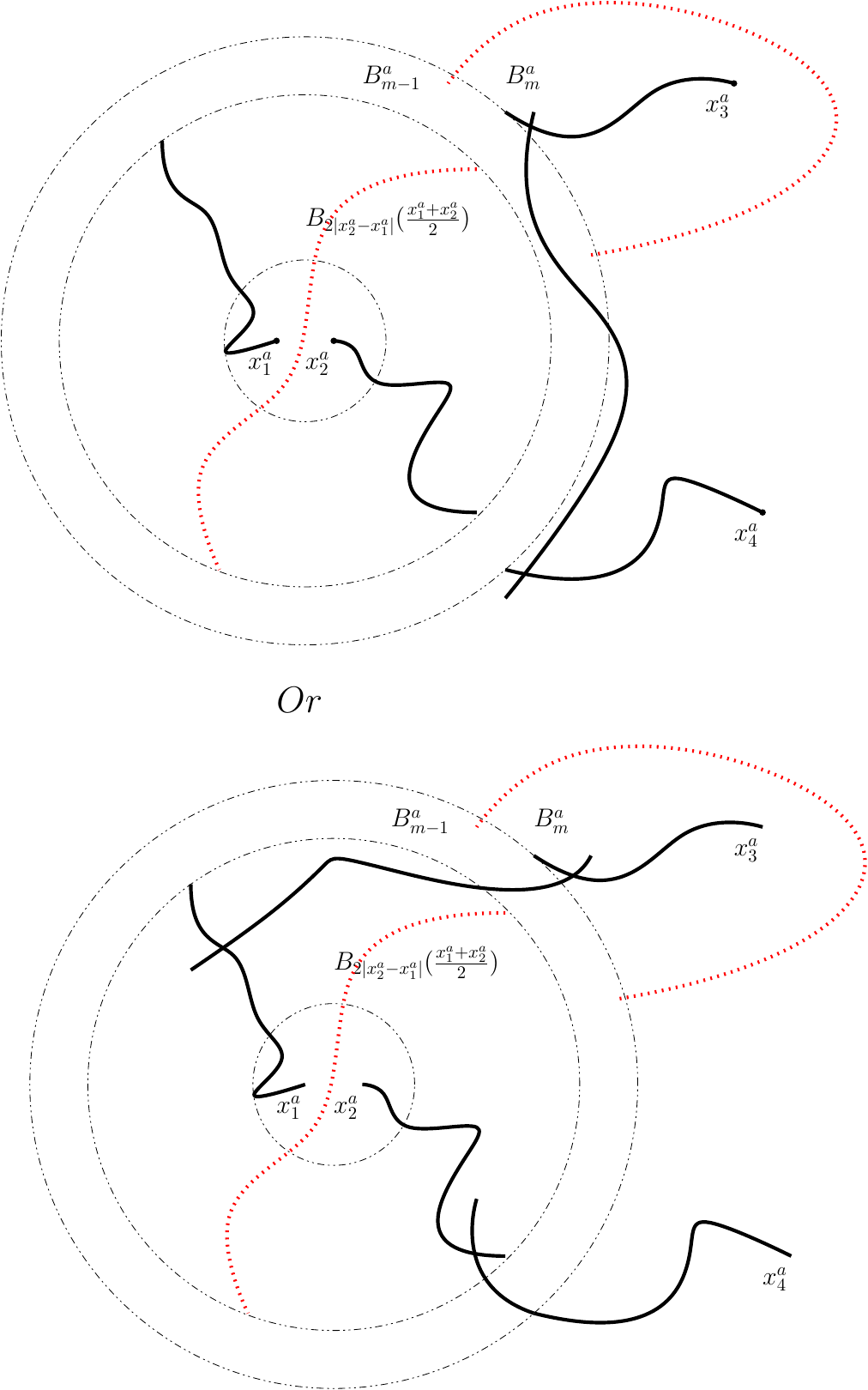}
		\end{center}
		\caption{The event $E_b^a$; }
	\end{subfigure}
 \begin{subfigure}{0.45\textwidth}
		\begin{center}
\includegraphics[width=1\textwidth]{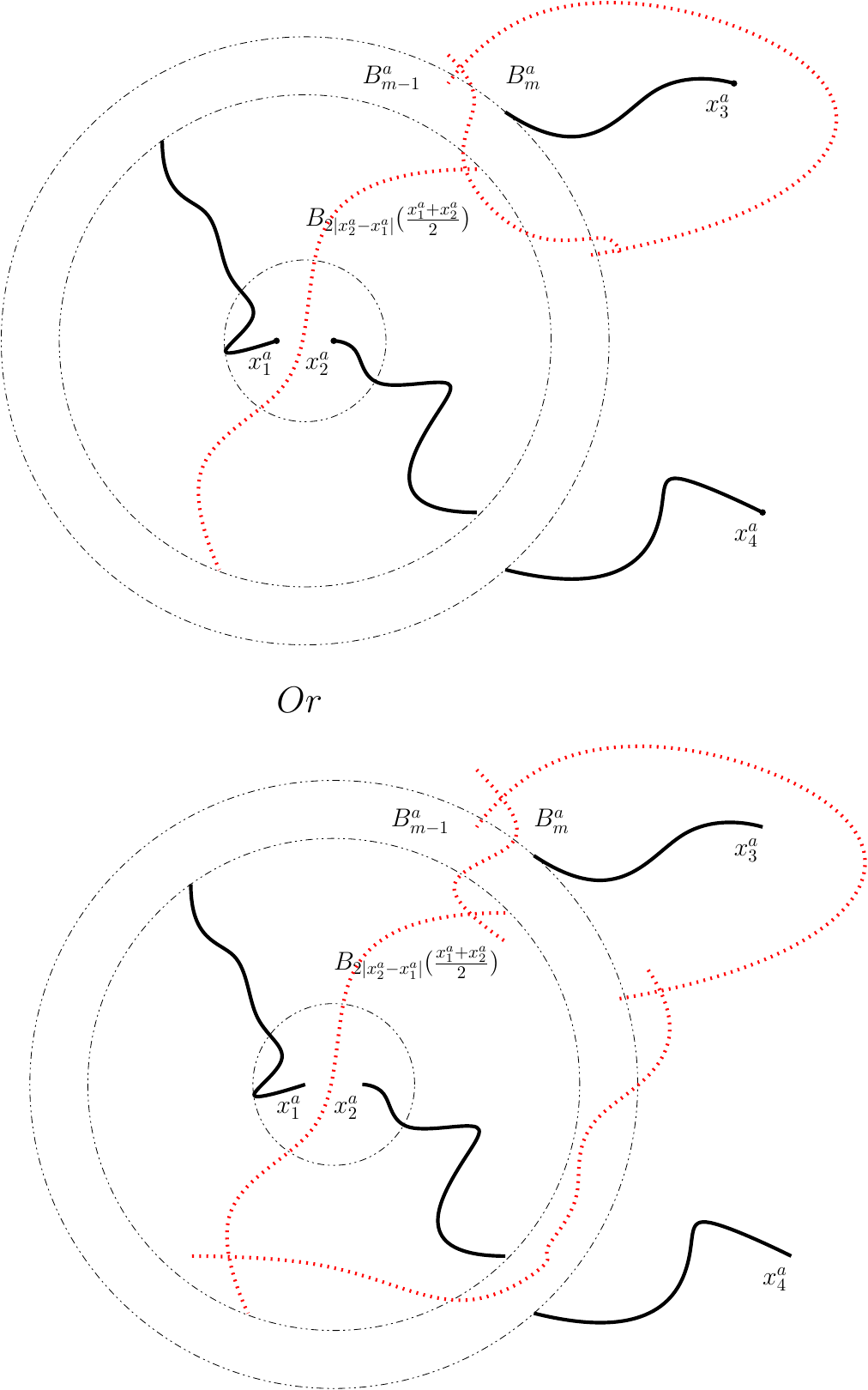}
		\end{center}
		\caption{The event $E_w^a$; }
	\end{subfigure}
	\begin{subfigure}{0.45\textwidth}
		\begin{center}
			\includegraphics[width=0.85\textwidth]{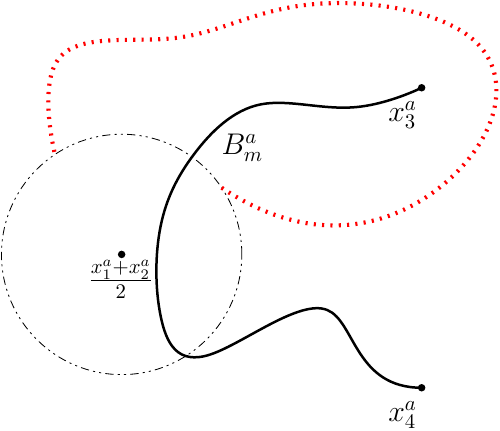}
		\end{center}
		\caption{The event $E_{1,b}^a$;}
	\end{subfigure}
	\begin{subfigure}{0.45\textwidth}
		\begin{center}
\includegraphics[width=0.85\textwidth]{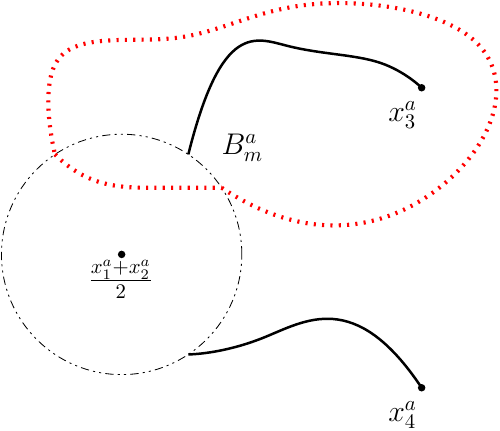}
		\end{center}
		\caption{The event $E_{1,w}^a$. }
	\end{subfigure}
	\caption{The events $E_b^a$, $E_w^a$, $E_{1,b}^a$ and $E_{1,w}^a$. The black, solid lines represent black paths, and the red, dotted lines represent white paths.}
	\label{fig::fourarmvariants2}
\end{figure}

\begin{figure}
		\begin{center}
\includegraphics[width=0.7\textwidth]{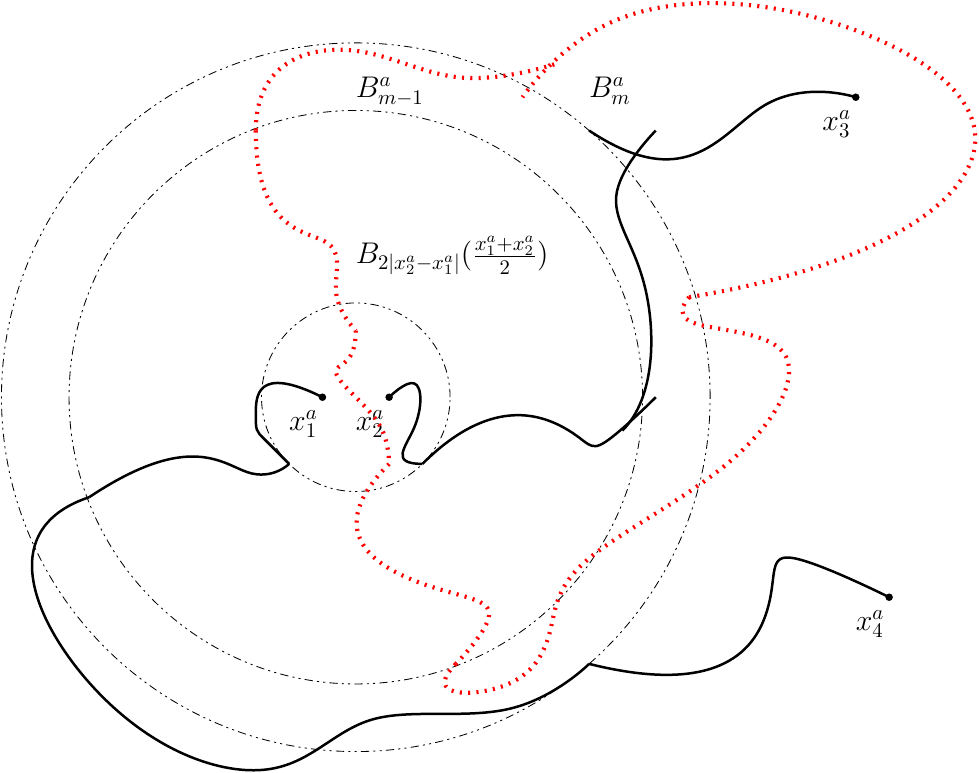}
		\end{center}
		\caption{The event $E^a$. The black, solid lines represent black paths, and the red, dotted lines represent white paths.}
  \label{fig::Ea}
\end{figure}

\begin{proof}[Proof of Lemma~\ref{lem::expansion_aux14}]
	Fix $\delta\in (0,1/2)$, $y\in [\delta,1-\delta]$. Let $m\in \{\lfloor y\cdot M\rfloor, \lceil y\cdot M\rceil\}$. 
	First, we show that 
	\begin{equation} \label{eqn::expansion_aux141}
		V_2=\frac{1}{2}.
	\end{equation}
	Define $E_{1,b}^{a}$ (resp., $E_{1,w}^{a}$) to be the event that the two black (resp., white) paths in the definition of $\mathcal{F}_m^{(o,a)}$ are in the same black (resp., white) cluster. Note that
	\begin{align} \label{eqn::expansion_aux1411}
\begin{split}
			\mathbb{P}^a\left[x_3^a\xlongleftrightarrow[(B_m^a)^c]{B}x_4^a|\mathcal{F}^{(o,a)}_{m}\right] & = \, \mathbb{P}^a\left[E_{1,b}^a|\mathcal{F}^{(o,a)}_{m}\right],\\
				\mathbb{P}^a\left[E_{1,b}^a|\mathcal{F}^{(o,a)}_{m}\right]+\mathbb{P}^a\left[E_{1,w}^a|\mathcal{F}^{(o,a)}_{m}\right] & = \, 1.
\end{split}
	\end{align}
	Using the strategy in~\cite[Proof of Theorem~1.5]{Cam23}, one  can show that the limits $\lim_{a\to 0}\mathbb{P}^a\left[E_{1,b}^a|\mathcal{F}^{(o,a)}_{m}\right]$ and $\lim_{a\to 0}\mathbb{P}^a\left[E_{1,w}^a|\mathcal{F}^{(o,a)}_{m}\right]$ exist and belong to $(0,1)$. Using the argument in the proof of Lemma~\ref{lem::four_aux2}, one can show that 
	\begin{equation}\label{eqn::expansion_aux1412}
		\lim_{x_1,x_2\to x}\lim_{a\to 0}\mathbb{P}^a\left[E_{1,w}^a|\mathcal{F}^{(o,a)}_{m}\right]=\lim_{x_1,x_2\to x}\lim_{a\to 0}				\mathbb{P}^a\left[E_{1,b}^a|\mathcal{F}^{(o,a)}_{m}\right],
	\end{equation}
	which means that the two limits in~\eqref{eqn::expansion_aux1412} exist and that they are equal. 
	Combining~\eqref{eqn::expansion_aux1411} with~\eqref{eqn::expansion_aux1412} gives~\eqref{eqn::expansion_aux141}, as desired. 
	
	Second, we show that 
	\begin{equation} \label{eqn::expansion_aux142}
\lim_{x_1,x_2\to x}	P\left(E_b| \mathcal{F}_{m-1}^{i}, \mathcal{F}_m^o\right)>\frac{1}{2}.
	\end{equation}
On the one hand, note that 
\begin{equation} \label{eqn::expansion_aux1421}
 \mathbb{P}^a\left[E_w^a\cup E_b^a| \mathcal{F}_{m-1}^{(i,a)}, \mathcal{F}_m^{(o,a)}\right]=1,\quad  \mathbb{P}^a\left[E_w^a\cap E_b^a| \mathcal{F}_{m-1}^{(i,a)}, \mathcal{F}_m^{(o,a)}\right]\geq  \mathbb{P}^a\left[E^a| \mathcal{F}_{m-1}^{(i,a)}, \mathcal{F}_m^{(o,a)}\right],
\end{equation}
where $E^a$ is defined as follows.
First of all, note that $\mathcal{F}_{m}^{(o,a)}$ implies the presence of two black paths connecting $x^a_3$ and $x^a_4$ to $\partial B^a_m$ and a white cluster separating them, and $\mathcal{F}_{m-1}^{(i,a)}$ implies that the are two black paths connecting $x^a_1$ and $x^a_2$ to $\partial B^a_{m-1}$ and two white paths crossing the annulus $B^a_{m-1} \setminus B_{2 \vert x^a_2-x^a_1 \vert}\Big(\frac{x^a_1+x^a_2}{2}\Big)$ (see Fig.~\ref{fig::fourarmvariants1}).
Then, $E^a$ is the event that both black paths from $\mathcal{F}_{m}^{(o,a)}$ are connected to the union of the two black paths from $\mathcal{F}_{m-1}^{(i,a)}$ by black paths and the white cluster from $\mathcal{F}_{m}^{(o,a)}$ is connected to both of the white paths crossing $B^a_{m-1} \setminus B_{2 \vert x^a_2-x^a_1 \vert}\Big(\frac{x^a_1+x^a_2}{2}\Big)$ from $\mathcal{F}_{m-1}^{(i,a)}$ (see Fig.~\ref{fig::Ea}).

A standard application of RSW estimates and separation of arms (see, e.g.,~\cite[Theorem~11]{NolinNearCriticalPercolation}) implies the existence of a constant $C_{11}>0$ that may depend only on $\delta$ and $\epsilon$ such that for all $a>0$ small enough, we have  
\begin{equation} \label{eqn::expansion_aux1422}
	\mathbb{P}^a\left[E^a| \mathcal{F}_{m-1}^{(i,a)}, \mathcal{F}_m^{(o,a)}\right]\geq C_{11}. 
\end{equation}
On the other hand, using the argument in the proof of Lemma~\ref{lem::four_aux2}, one can show that 
\begin{equation} \label{eqn::expansion_aux1423}
	\lim_{x_1,x_2\to x} 	P\left(E_b | \mathcal{F}_{m-1}^{i}, \mathcal{F}_m^o\right)= \lim_{x_1,x_2\to x}	P\left(E_w | \mathcal{F}_{m-1}^{i}, \mathcal{F}_m^o\right),
\end{equation}
which means that the two limits in~\eqref{eqn::expansion_aux1423} exist and that they are equal. Combining~\eqref{eqn::expansion_aux1421}-\eqref{eqn::expansion_aux1423} yields
\begin{align*}
& \lim_{x_1,x_2\to x} P\left(E_b | \mathcal{F}_{m-1}^{i}, \mathcal{F}_m^o\right) = \frac{\lim_{x_1,x_2\to x}	P\left(E_b | \mathcal{F}_{m-1}^{i}, \mathcal{F}_m^o\right)+	\lim_{x_1,x_2\to x} 	P\left(E_w | \mathcal{F}_{m-1}^{i}, \mathcal{F}_m^o\right)}{2}\\
& \qquad \qquad \qquad = \lim_{x_1,x_2\to x} \frac{P\left(E_b\cup E_w| \mathcal{F}_{m-1}^{i}, \mathcal{F}_m^o\right)+P\left(E_b \cap E_w| \mathcal{F}_{m-1}^{i}, \mathcal{F}_m^o\right)}{2}\\
& \qquad \qquad \qquad \geq \frac{1}{2}+\frac{C_{11}}{2}>\frac{1}{2}, 
\end{align*}
as desired. 
	
	Third, we show that 
	\begin{equation} \label{eqn::expansion_aux143}
V_3\geq \lim_{x_1,x_2\to x}	P\left(E_b | \mathcal{F}_{m-1}^{i}, \mathcal{F}_m^o\right). 
	\end{equation}
	Since 
	\begin{equation*}
	\mathbb{P}^a\left[x_3^a\xlongleftrightarrow[(B_m^a)^c]{B}x_4^a \bigg| x_1^a\xlongleftrightarrow[B_{m-1}^a]{B;B_m^a}x_2^a,\mathcal{F}_m^{(o,a)}\right]=\mathbb{P}^a\left[E_b^a \bigg| x_1^a\xlongleftrightarrow[B_{m-1}^a]{B;B_m^a}x_2^a,\mathcal{F}_m^{(o,a)}\right],
	\end{equation*}
	 it suffices to show that 
	 \begin{equation}\label{eqn::expansion_aux1431}
	 \mathbb{P}^a\left[E_b^a \bigg| x_1^a\xlongleftrightarrow[B_{m-1}^a]{B;B_m^a}x_2^a,\mathcal{F}_m^{(o,a)}\right]\geq \mathbb{P}^a\left[E_b^a \bigg| \mathcal{F}_{m-1}^{(i,a)}, \mathcal{F}_m^{(o,a)}\right]. 
	 \end{equation}
	 
Given a percolation configuration $\Lambda^a$. For each vertex $v\in a\mathcal{T} $, we define $\Lambda^a(v)$ to be $1$ (resp., $0$) if the label at $v$ under $\Lambda^a$ is black (reps., white). For two percolation configurations $\Lambda^a$, $\hat{\Lambda}^a$ , we say $\Lambda\leq \hat{\Lambda}^a$  if  $\Lambda^a(v)\leq \hat{\Lambda}^a(v)$ for all $v$.  We say an event $E$ is increasing if $\mathbb{1}_{E}(\Lambda^a)\leq \mathbb{1}_{E}(\hat{\Lambda}^a)$ whenever $\Lambda^a\leq \hat{\Lambda}^a$, where $\mathbb{1}_E$ is the indicator function of the event $E$. 
Note that:
\begin{itemize}
	\item the events $\mathcal{F}_{m}^{(o,a)}$ and $\mathcal{F}_{m-1}^{(i,a)}$ depend only on the labels of vertices in $a\mathcal{T}\cap (B_m^a)^c$ and $a\mathcal{T}\cap B_{m-1}^a$, respectively;
	\item given the labels of vertices in $a\mathcal{T}\cap \left(\left(B_{m}^a\right)^c\cup B_{m-1}^a\right)$ such that $\mathcal{F}_{m}^{(o,a)}\cap \mathcal{F}_{m-1}^{(i,a)}$ happens, then the event $\{x_1^a\xlongleftrightarrow[B_{m-1}^a]{B;B_{m}^a}x_2^a\}$ is an increasing event if we view it as an event for percolation of black vertices on $a\mathcal{T}\cap \left(B_m^a\setminus B_{m-1}^a\right)$.
\end{itemize}

The two observations above and the FKG inequality imply that there is a coupling $(\Lambda^a,\hat{\Lambda}^a)$ such that $\Lambda^a\sim \mathbb{P}^a\left[\cdot | \mathcal{F}_{m-1}^{(i,a)},\mathcal{F}_m^{(o,a)}\right]$, $\hat{\Lambda}^a\sim \mathbb{P}^a\left[\cdot |x_1^a\xlongleftrightarrow[B_{m-1}^a]{B;B_{m}^a}x_2^a, \mathcal{F}_m^{(o,a)}\right]$ and that 
\begin{equation*}
	\begin{cases}
		\Lambda^a(v)= \hat{\Lambda}^a(v), &\text{if }v\in a\mathcal{T}\cap \left(\left(B_{m}^a\right)^c\cup B_{m-1}^a\right),\\
		\Lambda^a(v)\leq \hat{\Lambda}^a(v),& \text{if }v\in a\mathcal{T}\cap \left(B_m^a\setminus B_{m-1}^a\right).
	\end{cases}
\end{equation*}
The existence of such a coupling implies readily~\eqref{eqn::expansion_aux1431}, which implies \eqref{eqn::expansion_aux143}.	
Combining~\eqref{eqn::expansion_aux141} and \eqref{eqn::expansion_aux142} with~\eqref{eqn::expansion_aux143} gives the desired result. 
\end{proof}

\subsection{
Correlations and fusion of four-arm events: Proof of Theorem~\ref{thm::4arm_expansion}}

\begin{figure} 
	\begin{subfigure}{1\textwidth}
		\begin{center}
			\includegraphics[width=0.6\textwidth]{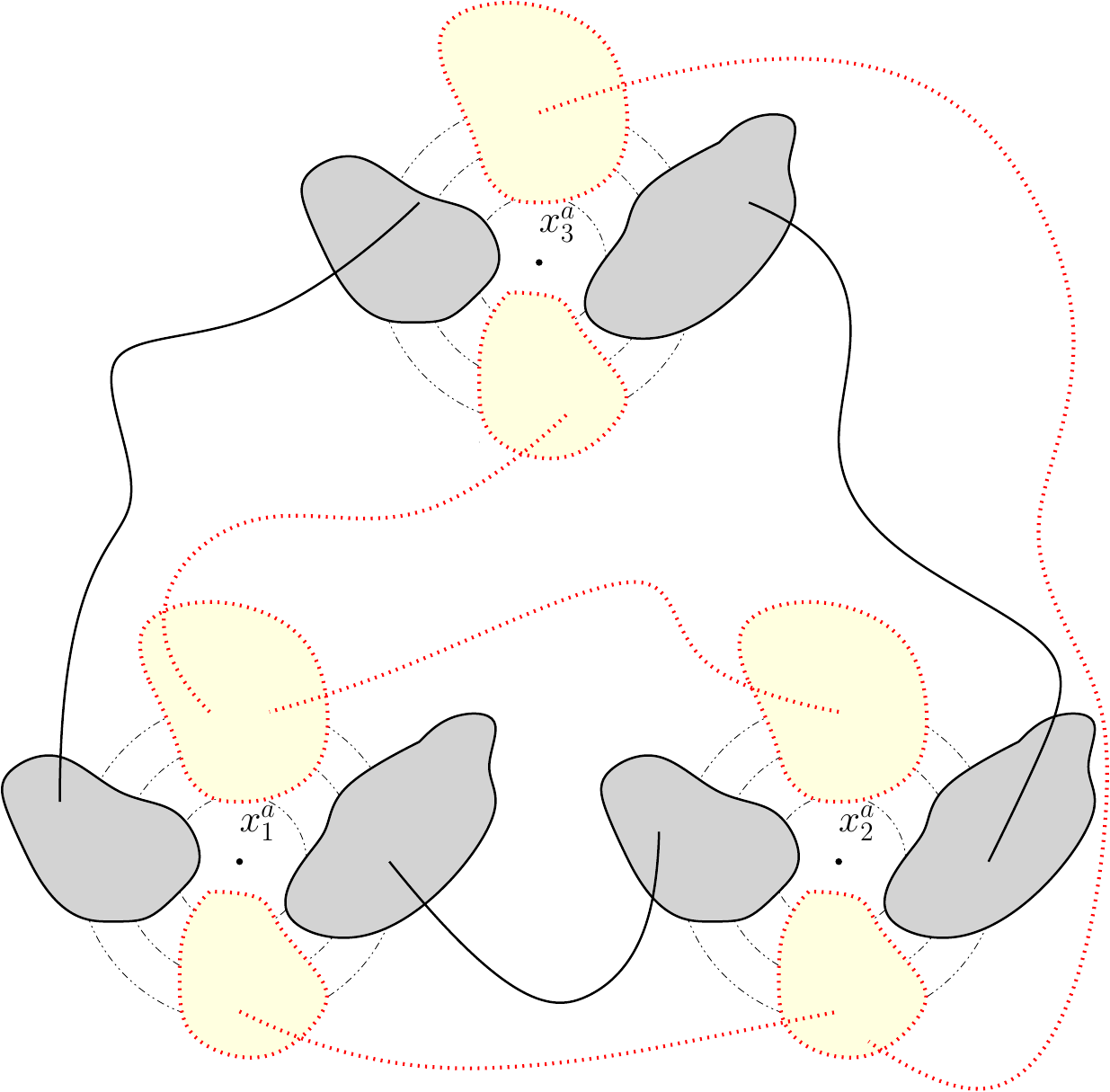}
		\end{center}
		\caption{The event $\mathcal{R}^{(a,\eta,\delta_m,\epsilon)}_3(x_1^a,x_2^a,x_3^a)$;}
	\end{subfigure}
	\begin{subfigure}{1\textwidth}
		\begin{center}
			\includegraphics[width=0.6\textwidth]{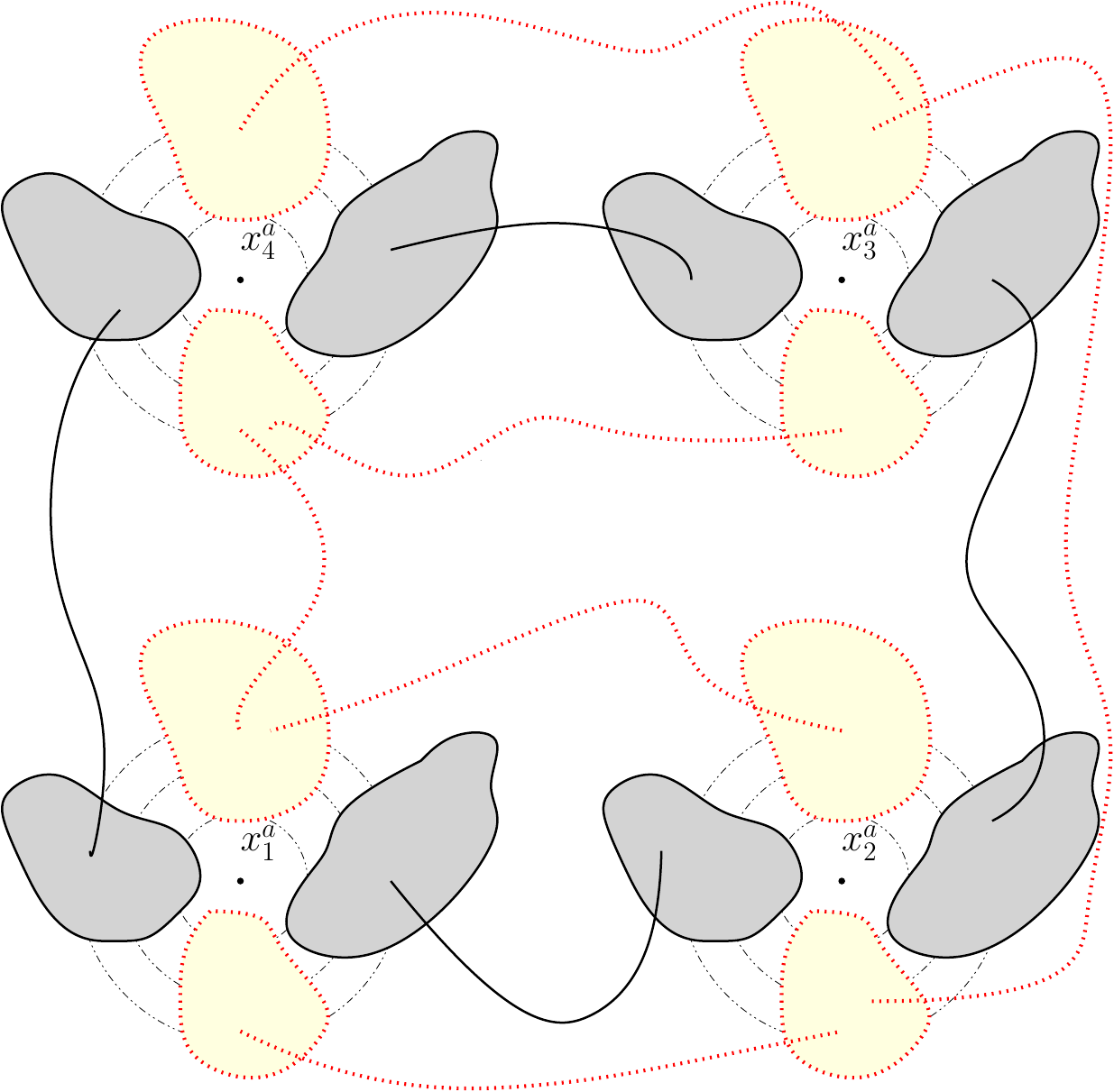}
		\end{center}
		\caption{The event $\mathcal{R}^{(a,\eta,\delta_m,\epsilon)}_4(x_1^a,x_2^a,x_3^a,x_4^a)$. }
	\end{subfigure}
	\caption{The events $\mathcal{R}_i^{{(a,\eta,\delta_m,\epsilon)}}(x_1^a,\ldots,x_i^{a})$ for $i \in\{3,4\}$. For each $1\leq j\leq 4$, there are $3$ concentric circles centered at $x_j^a$, with radii $\eta<\delta_m<\epsilon$, respectively. The black, solid lines represent black paths and the red, dotted lines represent white paths. The grey domains represent parts of black clusters, while the yellow domains represent parts of white clusters.}
	\label{fig::Rclusterfigures}
\end{figure}

The goal of this section is to prove Theorem~\ref{thm::4arm_expansion}. 
\begin{proof}[Proof of Theorem~\ref{thm::4arm_expansion} (\ref{item::existence_limit}).]
We first observe that standard RSW arguments imply that, if the limit in~\eqref{eqn::def_R} exists, it must be in $(0,\infty)$.
Therefore, we focus on the proof of the existence of the limit, which follows from the coupling result in Lemma~\ref{lem::four_arm_coupling_inner}.

More precisely, we  let $0<\eta<\delta_m<\epsilon<\min_{1\leq j<k\leq 4}\frac{|x_j-x_k|}{10}$ and define  $\mathcal{R}_4^{(a,\eta,\delta_m,\epsilon)}(x_1^a,x_2^a,x_3^a,x_4^a)$ to be the following event (with the convention that $x_5^a=x_1^a$):
\begin{enumerate}
    \item for $1\leq j\leq 4$, there are four paths, with labels black, white, black, white, in counterclockwise order, connecting $\partial B_{\eta}(x_j^a)$ to $\partial B_{\epsilon}(x_j^a)$ (we say that these four paths are \emph{adjacent to} $x_j^a$);
    \item for $1\leq j\leq 4$, a black path adjacent to $x_j^a$ and a black path adjacent to $x_{j+1}^a$ are connected by a black path inside $\mathbb{C}\setminus \left(\cup_{k=1}^4 B_{\delta_m}(x_k^a)\right)$;
    \item for $1\leq j\leq 4$, one of the two white paths adjacent to $x_j^a$ is connected to one of the two white paths adjacent to $x_{j+1}^a$ by a white path inside $\mathbb{C}\setminus \left(\cup_{k=1}^4 B_{\delta_m}(x_k^a)\right)$, and the other white path adjacent to $x_j^a$ is connected to the other white path adjacent to $x_{j+1}^a$ by a white path inside $\mathbb{C}\setminus \left(\cup_{k=1}^4 B_{\delta_m}(x_k^a)\right)$.
\end{enumerate}
We define the event  $\mathcal{R}_3^{(a,\eta,\delta_m,\epsilon)}(x_1^a.x_2^a,x_3^a)$ in a similar way. See Figure~\ref{fig::Rclusterfigures} for a schematic illustration of these two events.

As in Sections~\ref{sec::pivotal} and~\ref{sec::thm::backbone}, one can express the events $\mathcal{R}_4^{(a,\eta,\delta_m,\epsilon)}(x^a_1,x^a_2,x^a_3,x^a_4)$ and $\mathcal{R}_3^{(a,\eta,\delta_m,\epsilon)}(x^a_1,x^a_2,x^a_3)$ in terms of percolation interface loops and define analogous events in the continuum for the full scaling limit $\Lambda$.
The latter events will be denoted $\mathcal{R}_4^{(\eta,\delta_m,\epsilon)}(x_1,x_2,x_3,x_4)$ and $\mathcal{R}_3^{(\eta,\delta_m,\epsilon)}(x_1,x_2,x_3)$ and it is easy to see, using arguments described earlier for similar events, that they are continuity events for $\mathbb{P}$.
Then, one can proceed as in the proof of Theorem~\ref{thm::backbone}, with Lemma~\ref{lem::backbone} replaced by Lemma~\ref{lem::four_arm_coupling_inner}, to show that, for $i\in\{3,4\}$,
\begin{align} \label{eqn::cvg_fourarm_four}
\begin{split}
& R_i(x_1,\ldots,x_i):=\lim_{a\to 0} \overline{\rho}_a^{-i} \times \mathbb{P}^a\left[\mathcal{R}_i^a\left(x_1^a,\ldots,x_i^a\right)\right] \\
& \qquad \qquad = \epsilon^{-\frac{5i}{4}} \times \lim_{m\to \infty}\lim_{\eta\to 0} \mathbb{P}\left[\mathcal{R}_i^{{(\eta,\delta_m,\epsilon)}}(x_1,\ldots,x_i)|\mathcal{F}_{\eta,\epsilon}(x_j), \; 1\leq j\leq i\right].
\end{split}
\end{align}
\end{proof}

\begin{proof}[Proof of Theorem~\ref{thm::4arm_expansion} (\ref{item::conformal_cova})]
	The goal is to show that, for $i\in \{3,4\}$, the functions $R_i(x_1,\ldots,x_3)$ satisfy the desired M\"obius covariance property. This can be done using essentially the same method as in~\cite[Proofs of Theorems~1.1\&1.4]{Cam23}, with the help of~\eqref{eqn::scal_four} to get the correct exponent. The explicit functional expression for $R_3$ follows directly from the M\"obius covariance property. 
\end{proof}

\begin{proof}[Proof of Theorem~\ref{thm::4arm_expansion} (\ref{item::expansion})]
We assume the same setup as in the proof of Theorem~\ref{thm::4arm_expansion} (\ref{item::existence_limit}). Let $x\in \mathbb{C}\setminus \{x_1,x_2\}$ and assume that $x_3,x_4$ are much closer to $x$ than to $x_1$ or $x_2$. We fix a number $\epsilon\in (0,\min\{\frac{|x_1-x_2|}{10},  \frac{|x_1-x|}{10}, \frac{|x_2-x|}{10}\})$ sufficiently small and define the event $\overline{\mathcal{R}}_4^{(\eta,\delta_m,\epsilon)}(x_1,x_2,x_3,x_4)$ in the same way as $\mathcal{R}_4^{(\eta,\delta_m,\epsilon)}$, except that we replace the disks around $x_3$ and $x_4$ with smaller disks, with radii $\varkappa\eta,\varkappa\delta_m,\varkappa\epsilon$, where $\varkappa:=\frac{|x_4-x_3|}{2}$. We write $\mathcal{F}_{\varkappa}$ for the event $\{\mathcal{F}_{\eta,\epsilon}(x_j), \mathcal{F}_{\varkappa\eta,\varkappa\epsilon}(x_k), \; j=1,2, \, k=3,4\}$.
Then, the same proof as for \eqref{eqn::cvg_fourarm_four} gives
\begin{align} \label{eqn::fourarm_expansion_aux1}
	R_4(x_1,x_2,x_3,x_4)=\lim_{m\to \infty}\lim_{\eta\to 0} \mathbb{P}\left[\overline{\mathcal{R}}_4^{(\eta,\delta_m,\epsilon)}(x_1,x_2,x_3,x_4)|\mathcal{F}_{\varkappa} \right]
	\times \epsilon^{-5}\times \left(\frac{|{x_4-x_3}|}{2}\right)^{-\frac{5}{2}} .
\end{align}
We claim that
\begin{equation} \label{eqn::fourarm_expansion_aux2}
	\lim_{x_3,x_4\to x} |{x_4-x_3}|^{-\frac{5}{4}}\times \lim_{m\to \infty}\lim_{\eta\to 0} \mathbb{P}\left[\overline{\mathcal{R}}_4^{(\eta,\delta_m,\epsilon)}(x_1,x_2,x_3,x_4)| \mathcal{F}_{\varkappa}\right] \textit{ exists  and belongs to }(0,\infty). 
\end{equation}
Combining~\eqref{eqn::fourarm_expansion_aux1} with~\eqref{eqn::fourarm_expansion_aux2}, we conclude that 
\begin{equation*}
\overline{R}_3(x_1,x_2,x):=	\lim_{x_3,x_4\to x} \frac{R_4(x_1,x_2,x_3,x_4)}{|x_4-x_3|^{-\frac{5}{4}}}
\end{equation*}
exists and belongs to $(0,\infty)$. Moreover, for any non-constant M\"obius transformation $\varphi$ such that 
\[\varphi(x),\varphi(x_1),\varphi(x_2)\neq \infty,\]
we have 
\begin{align*}
& \overline{R}_3(\varphi(x_1), \varphi(x_2),\varphi(x)) = \lim_{x_3,x_4\to x} \frac{R_4(\varphi(x_1),\varphi(x_2),\varphi(x_3),\varphi(x_4))}{|\varphi(x_4)-\varphi(x_3)|^{-\frac{5}{4}}} \\
& \qquad \qquad \qquad = \overline{R}_3(x_1,x_2,x)\times |\varphi'(x)|^{-\frac{5}{4}}\times \prod_{j=1}^2 |\varphi'(x_j)|^{-\frac{5}{4}},
\end{align*}
where the last equality follows from Theorem~\ref{thm::4arm_expansion}~(\ref{item::conformal_cova}).
As a consequence, there exists a universal constant $C_4\in (0,\infty)$ such that 
\begin{equation*}
	\overline{R}_3(x_1,x_2,x)=C_4 R_3(x_1,x_2,x),
\end{equation*}
as desired.

Now, we prove~\eqref{eqn::fourarm_expansion_aux2}. We denote by $\mathcal{E}_{\varkappa}$ the event $\mathcal{F}_{10\varkappa,5\epsilon}(x)$. Note that, for $x_3$ and $x_4$ sufficiently close to $x$ {and $\eta$ sufficiently small}, we have 
\begin{equation*}
\overline{\mathcal{R}}_4^{(\eta,\delta_m,\epsilon)}(x_1,x_2,x_3,x_4)\subseteq 	\mathcal{E}_{\varkappa}.
\end{equation*}
Then, {taking $\epsilon$ sufficiently small,} we can write 
\begin{equation*}
	|x_4-x_3|^{-\frac{5}{4}} \times \mathbb{P}\left[\overline{\mathcal{R}}_4^{(\eta,\delta_m,\epsilon)}(x_1,x_2,x_3,x_4)|\mathcal{F}_{\varkappa}\right]= 	\underbrace{|x_4-x_3|^{-\frac{5}{4}} \times \mathbb{P}\left[\mathcal{E}_{\varkappa}\right]}_{T_1}\times \underbrace{\mathbb{P}\left[\overline{\mathcal{R}}_4^{(\eta,\delta_m,\epsilon)}(x_1,x_2,x_3,x_4)|\mathcal{F}_{\varkappa}, \mathcal{E}_{\varkappa}\right]}_{T_2}.
\end{equation*}
For the term $T_1$, thanks to the M\"obius invariance of $\mathrm{CLE}_6$ \cite{CamiaNewmanPercolationFull,GwynneMillerQianCLE} and~\eqref{eqn::coro_aux2}, we have 
\begin{equation*}
	\lim_{x_3,x_4\to x} T_1= \lim_{x_3,x_4\to x} |x_4-x_3|^{-\frac{5}{4}}\times \mathbb{P}\left[\mathcal{F}_{2\varkappa,\epsilon}(0)\right]= C_8\epsilon^{-\frac{5}{4}},  
\end{equation*}
where $C_8\in (0,\infty)$ is the universal constant in~\eqref{eqn::coro_aux2}.

It remains to show that $\lim_{x_3,x_4\to x}\lim_{m\to \infty}\lim_{\eta\to 0} T_2$ exists\footnote{Once the existence of this limit is proved, the fact that it belongs to $(0,\infty)$ can be derived using standard RSW arguments.}. The idea is to use once again Lemmas~\ref{lem::four_arm_coupling_inner} and~\ref{lem::faces}.
{As before, we assume that $x_3$ and $x_4$ are sufficiently close to $x$ and that $\eta$ and $\epsilon$ are sufficiently small.}
We define three auxiliary events\footnote{We choose to define these events in terms of paths and clusters in the continuum, since this might be a more intuitive way. It is clear that these events can also be expressed in terms of loops. {To this effect, we say that two sets, $A$ and $B$, are connected by a black (resp., white) path if there is no loop such that one of the two sets is contained in its interior and the other in the complement and the smallest loop containing both in its interior is oriented counterclockwise (resp., clockwise).} Moreover, since the polychromatic boundary $3$-arm exponent is strictly larger than $1$~\cite{SmirnovWernerCriticalExponents}, these events are actually continuity events for $\mathbb{P}$, {by the same argument used earlier for similar events (e.g., $\mathcal{A}_{\eta,\epsilon}(z)$ and $\mathcal{B}_{\eta,\epsilon}(z)$)}.}, see Figure~\ref{fig::Rfusion} for an illustration. 
\begin{itemize}
	\item First, we define $\tilde{R}_4^{(\mathrm{int})}$ to be the following event:
 \begin{enumerate}
     \item $\mathcal{E}_{\varkappa}\cap \mathcal{F}_{\varkappa}$ happens.
     \item {One black path adjacent to $x_3$ and one black path adjacent to $x_4$ from $\mathcal{F}_{\varkappa}$ are connected to each other by a black path contained inside $B_{10\sqrt{\varkappa}}(x)\setminus (\cup_{j=3,4}B_{\varkappa\delta_m}(x_j))$}.
     \item The second black path from $\mathcal{F}_{\varkappa}$ adjacent to $x_3$ is connected to one of the black paths from the event $\mathcal{E}_{\varkappa}$ by a black path contained inside $B_{10\sqrt{\varkappa}}(x)\setminus (\cup_{j=3,4}B_{\varkappa\delta_m}(x_j))$ and the second black path from $\mathcal{F}_{\varkappa}$ adjacent to $x_4$ is connected to the other black path from the event $\mathcal{E}_{\varkappa}$ by a black path contained inside $B_{10\sqrt{\varkappa}}(x)\setminus (\cup_{j=3,4}B_{\varkappa\delta_m}(x_j))$.
     \item One white path adjacent to $x_3$ and one white path adjacent to $x_4$ from $\mathcal{F}_{\varkappa}$ are connected to each other by a white path contained inside $B_{10\sqrt{\varkappa}(x)}\setminus (\cup_{j=3,4}B_{\varkappa\delta_m}(x_j))$.
     Moreover, this white path is connected to one of the white paths defining the event $\mathcal{E}_{\varkappa}$ by a white path contained inside $B_{10\sqrt{\varkappa}}(x)\setminus (\cup_{j=3,4}B_{\varkappa\delta_m}(x_j))$.
     \item The other white path adjacent to $x_3$ from $\mathcal{F}_{\varkappa}$ is connected to the other white path adjacent to $x_4$ from $\mathcal{F}_{\varkappa}$ by a white path contained inside $B_{10\sqrt{\varkappa}}(x)\setminus (\cup_{j=3,4}B_{\varkappa\delta_m}(x_j))$.
     Moreover, this white path is connected to the second white path defining the event $\mathcal{E}_{\varkappa}$ by a white path contained inside $B_{10\sqrt{\varkappa}}(x)\setminus (\cup_{j=3,4}B_{\varkappa\delta_m}(x_j))$.
 \end{enumerate}  
	\item Second, we define the event $\tilde{R}_4^{(\mathrm{ext})}$ in a similar way to $\tilde{R}_4^{(\mathrm{int})}$, but with $(x_3,x_4)$ replaced by $(x_1, x_2)$, 
 and $B_{10\sqrt{\varkappa}}(x)\setminus (\cup_{j=3,4}B_{\varkappa\delta_m}(x_j))$ replaced by $\{z: |z-x|>10\varkappa^{1/3}\}\setminus (\cup_{j=1,2} B_{\delta_m}(x_j))$.
	\item Third, we define $\tilde{\mathcal{R}}_4(x_1,x_2,x_3,x_4)=\tilde{R}_4^{(\mathrm{int})}\cap \tilde{R}_4^{(\mathrm{ext})}$. 
\end{itemize}

Then, one can use Lemma~\ref{lem::four_arm_coupling_inner} and Lemma~\ref{lem::faces} to prove the following observations.
\begin{itemize}
	\item There exist constants $c_3,c_4\in (0,\infty)$, independent of $\varkappa,\eta,\delta_m$, such that
	\begin{align}
		\left| T_2- \mathbb{P}\left[\tilde{\mathcal{R}}_4(x_1,x_2,x_3,x_4)| \mathcal{F}_{\varkappa},\mathcal{E}_{\varkappa}\right] \right|& \leq c_3\varkappa^{c_4},\label{eqn::four_arm_proof_aux1}\\
		\left|\mathbb{P}\left[\tilde{\mathcal{R}}_4(x_1,x_2,x_3,x_4)| \mathcal{F}_{\varkappa},\mathcal{E}_{\varkappa}\right] -\mathbb{P}\left[\tilde{\mathcal{R}}_4^{(\mathrm{int})}| \mathcal{F}_{\varkappa},\mathcal{E}_{\varkappa}\right] \times \mathbb{P}\left[\tilde{\mathcal{R}}_4^{(\mathrm{\ext})}| \mathcal{F}_{\varkappa},\mathcal{E}_{\varkappa}\right] \right| & \leq c_3\varkappa^{c_4},\label{eqn::four_arm_proof_aux2}
		\end{align}
  where~\eqref{eqn::four_arm_proof_aux1} is a consequence of the fact that
  \begin{itemize}
      \item $\tilde{\mathcal{R}}_4(x_1,x_2,x_3,x_4)$ implies $\overline{\mathcal{R}}^{\eta,\delta_m,\epsilon}_4(x_1,x_2,x_3,x_4)$ and
      \item the difference $\overline{\mathcal{R}}^{\eta,\delta_m,\epsilon}_4(x_1,x_2,x_3,x_4) \setminus \tilde{\mathcal{R}}_4(x_1,x_2,x_3,x_4)$ implies a polychromatic five-arm event\footnote{This can happen if some of the paths from $\mathcal{F}_{\varkappa}$ are not connected to each other in the way required by $\tilde{\mathcal{R}}_4(x_1,x_2,x_3,x_4)$ inside $B_{10\sqrt{\varkappa}}(x)$.} in the annulus $A_{3\varkappa,10\sqrt{\varkappa}}(x)$, whose conditional probability given $\mathcal{F}_{\varkappa},\mathcal{E}_{\varkappa}$ decays polynomially in $\varkappa$, as $x_3,x_4\to x$, because the five-arm exponent is 2~\cite{KestenAlmostAllWords};
  \end{itemize}
  and~\eqref{eqn::four_arm_proof_aux2} can be proven as follows: write 
  \begin{align*}
\mathbb{P}\left[\tilde{\mathcal{R}}_4(x_1,x_2,x_3,x_4)| \mathcal{F}_{\varkappa},\mathcal{E}_{\varkappa}\right]=&\mathbb{P}\left[\tilde{\mathcal{R}}_4^{(\mathrm{int})}| \mathcal{F}_{\varkappa},\mathcal{E}_{\varkappa}\right] \times \mathbb{P}\left[\tilde{R}_4(x_1,x_2,x_3,x_4)|\tilde{R}_4^{(\mathrm{int})},\mathcal{F}_{\varkappa},\mathcal{E}_{\varkappa}\right]\\
=&\mathbb{P}\left[\tilde{\mathcal{R}}_4^{(\mathrm{int})}| \mathcal{F}_{\varkappa},\mathcal{E}_{\varkappa}\right] \times \mathbb{P}\left[\tilde{R}_4^{(\mathrm{ext})}|\tilde{R}_4^{(\mathrm{int})},\mathcal{F}_{\varkappa},\mathcal{E}_{\varkappa}\right],
  \end{align*}
  then by investigating the ``faces"\footnote{Compared with the proof of Lemma~\ref{lem::four_aux2}, there is some additional work to do, due to the lack of initial faces. However, the strategy is still clear, see~\cite[Sketch of the proof of Proposition~3.6]{GarbanPeteSchrammPivotalClusterInterfacePercolation}.} (introduced in the proof of Lemma~\ref{lem::four_aux2}) formed by four alternating paths crossing the annulus $A_{10\varkappa,5\epsilon}$, as in the proof of Lemma~\ref{lem::four_aux2}, one can obtain a suitable coupling between $\mathbb{P}\left[\, \cdot \, |\tilde{R}_4^{\mathrm{int}},\mathcal{F}_{\varkappa},\mathcal{E}_{\varkappa}\right]$ and $\mathbb{P}\left[\, \cdot \, |\mathcal{F}_{\varkappa},\mathcal{E}_{\varkappa}\right]$ and derive~\eqref{eqn::four_arm_proof_aux2}.

		\item The following limits exist:
		\begin{align} \label{eqn::def_R_4}
		R_4^{(\mathrm{int})}(\varkappa):=	\lim_{m\to \infty}\lim_{\eta\to 0} \mathbb{P}\left[\tilde{\mathcal{R}}_4^{(\mathrm{int})}| \mathcal{F}_{\varkappa},\mathcal{E}_{\varkappa}\right] \quad \textit{and}\quad R_4^{(\mathrm{ext})}(\varkappa):= \lim_{m\to \infty}\lim_{\eta\to 0} \mathbb{P}\left[\tilde{\mathcal{R}}_4^{(\mathrm{ext})}| \mathcal{F}_{\varkappa},\mathcal{E}_{\varkappa}\right] .
		\end{align}
		\item If $\{\varkappa_n\}_{n=1}^{\infty}$ is a decreasing sequence with $\lim_{n\to\infty} \varkappa_n=0 $, then $\{R_4^{(\mathrm{int})}(\varkappa_n)\}_{n=1}^{\infty}, 	\{R_4^{(\mathrm{ext})}(\varkappa_n)\}_{n=1}^{\infty} $ are two Cauchy sequences. As a consequence, the limits $\lim_{x_3,x_4\to x} R_4^{(\mathrm{int})}(\varkappa)$ and $ \lim_{x_3,x_4\to x} R_4^{(\mathrm{ext})}(\varkappa) $ exist. 
\end{itemize}

Combining all these observations together, we conclude that 
\begin{equation*}
	\lim_{x_3,x_4\to x}\lim_{m\to \infty}\lim_{\eta\to 0} T_2 =\lim_{x_3,x_4\to x} R_4^{(\mathrm{int})}(\varkappa)\times \lim_{x_3,x_4\to x} R_4^{(\mathrm{ext})}(\varkappa)
\end{equation*}
exists, as we set out to prove. 
\end{proof}

\begin{figure}
	\begin{center}
		\includegraphics[width=0.80\textwidth]{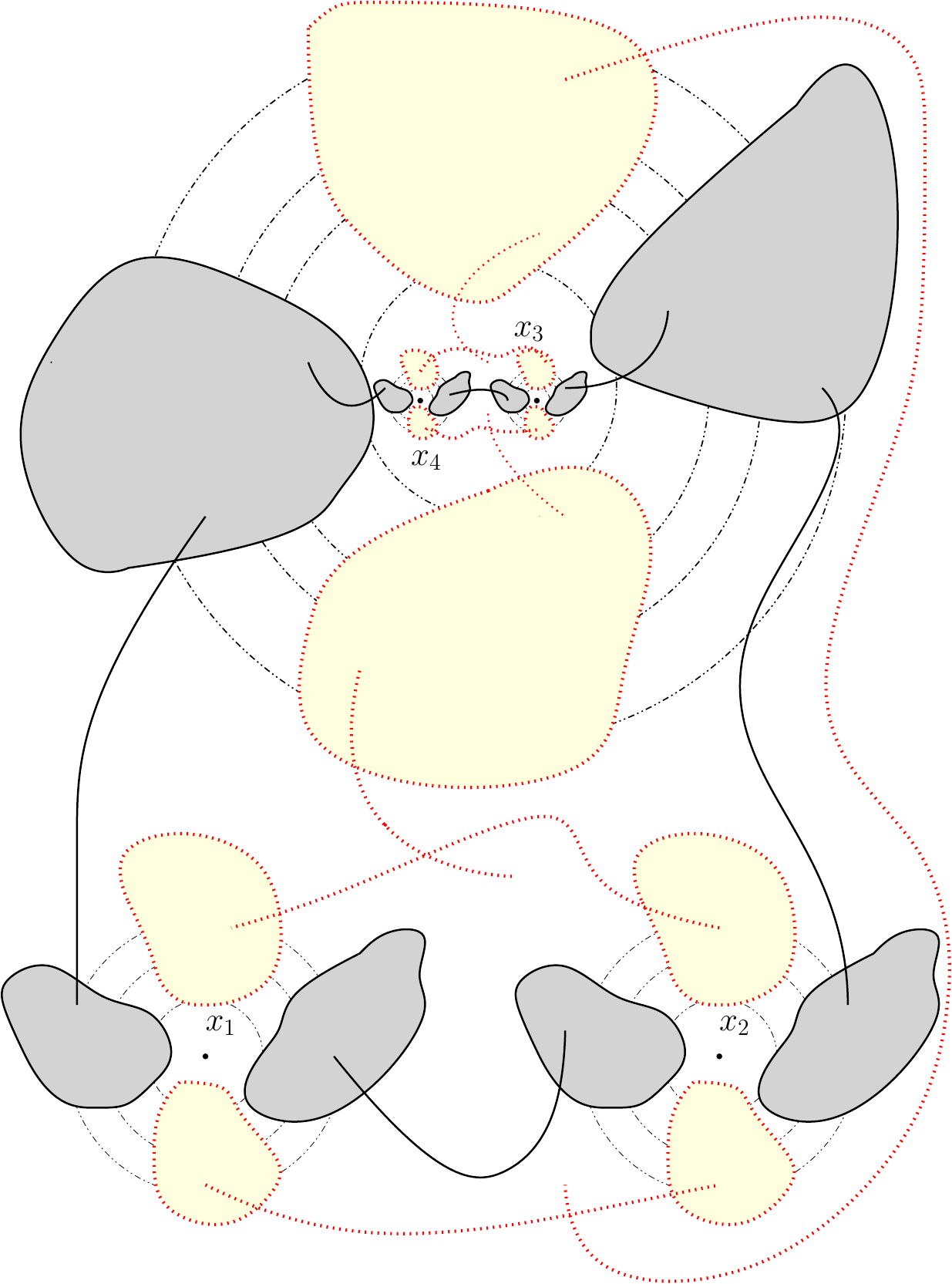}
	\end{center}
	\caption{The event $\tilde{\mathcal{R}}_4(x_1,x_2,x_3,x_4)$ when $x_3$ and $x_4$ are close to each other. For $1\leq j\leq 2$, there are $3$ concentric circles centered at $x_j$ with radius $\eta<\delta_m<\epsilon$, respectively, where $\epsilon\in (0,\min\{\frac{|x_1-x_2|}{10},  \frac{|x_1-x|}{10}, \frac{|x_2-x|}{10}\})$.  For $3\leq j\leq 4$, there are $3$ concentric circles centered at $x_j$ with radius $\varkappa\eta<\varkappa\delta_m<\varkappa\epsilon$, respectively, where $\varkappa:=\frac{|x_4-x_3|}{2}$. Moreover, there are $4$ concentric circles centered at $x$ with radius $10\varkappa<10\sqrt{\varkappa}<10\varkappa^{1/3}<5\epsilon$, respectively. Note that, as $x_3,x_4\to x$, standard RSW estimates imply that the probabilities of the event that two {black paths near $x_3$ and $x_4$ belong to the same black cluster and the event that two pairs of white clusters around $x_3$ and $x_4$ belong to the same white cluster} are uniformly bounded from below. }
	\label{fig::Rfusion}
\end{figure}

\subsection{
Correlations and fusion of boundary two-arm events: Proof of Theorem~\ref{thm::three_arm}}
	For $0<a<\eta<\epsilon$, we let $\{\partial B_{\eta}(0) \xlongleftrightarrow[\mathbb{H}]{BW} \partial B_{\epsilon}(0)\}$  denote the event that there are a black path and a white path in the upper half-plane, {in clockwise order,} connecting $\partial B_{\eta}(0)$  to $\partial B_{\epsilon}(0)$.
 Furthermore, we let $\{\partial B_{\eta}(0) \xlongleftrightarrow[\mathbb{H}]{BWB} \partial B_{\epsilon}(0)\}$ (resp., $\{0 \xlongleftrightarrow[\mathbb{H}]{BWB} \partial B_{\epsilon}(0)\}$) denote the event that there are two black paths and one white path in the upper half-plane connecting $\partial B_{\eta}(0)$ {(resp., $-a,0,a$)} to $\partial B_{\epsilon}(0)$, {with the white path between the two black paths}. 
\begin{lemma} \label{lem::two_arm_coupling_half}
The conclusions of Lemma~\ref{lem::four_arm_coupling_inner} still hold if we replace  $\{0\xlongleftrightarrow{BWBW}\partial B_{\epsilon}(0)\}$ by $\{0 \xlongleftrightarrow[\mathbb{H}]{BW} \partial B_{\epsilon}(0)\}$ and replace $\{\partial B_{\eta}(0) \xlongleftrightarrow{BWBW}\partial B_{\epsilon}(0)\}$  by $\{\partial B_{\eta}(0) \xlongleftrightarrow[\mathbb{H}]{BW} \partial B_{\epsilon}(0)\}$.
\end{lemma}

\begin{lemma}\label{lem::three_arm_coupling}
The conclusions of Lemma~\ref{lem::four_arm_coupling_inner} still hold if we replace  $\{0\xlongleftrightarrow{BWBW}\partial B_{\epsilon}(0)\}$ by $\{0 \xlongleftrightarrow[\mathbb{H}]{BWB} \partial B_{\epsilon}(0)\}$ and replace $\{\partial B_{\eta}(0) \xlongleftrightarrow{BWBW}\partial B_{\epsilon}(0)\}$  by $\{\partial B_{\eta}(0) \xlongleftrightarrow[\mathbb{H}]{BWB} \partial B_{\epsilon}(0)\}$.
\end{lemma}
\begin{proof}[Proof of Lemmas~\ref{lem::two_arm_coupling_half} and~\ref{lem::three_arm_coupling}]
The two lemmas can be proved in the same way as Lemma~\ref{lem::four_arm_coupling_inner}, using the strategy in the proof of Proposition~3.6 of~\cite{GarbanPeteSchrammPivotalClusterInterfacePercolation}, with an exploration process that starts at the origin and moves outwards.
In fact, the presence of a boundary makes the proof on the upper half-plane easier than on the plane.
\end{proof}

Let
\begin{equation*}
	\overline{\iota}_a:=\mathbb{P}^a\left[0 \xlongleftrightarrow[\mathbb{H}]{BWB} \partial B_{1}(0)\right].
\end{equation*}
It was shown in~\cite{SmirnovWernerCriticalExponents} that $\overline{\iota}_a=a^{2+o(1)}$.
As in Sections~\ref{sec::thm::backbone} and~\ref{sec::pivotal}, one can express the event $\{\partial B_{\eta}(0) \xlongleftrightarrow[\mathbb{H}]{BWB} \partial B_{\epsilon}(0)\}$ in terms of interface loops. Thus, we can define $\{\partial B_{\eta}(0) \xlongleftrightarrow[\mathbb{H}]{BWB} \partial B_{\epsilon}(0)\}$ in the continuum, for the full scaling limit $\Lambda$.
One can proceed as in the proof of Corollary~\ref{coro::four_arm}, with Lemma~\ref{lem::four_aux12} replaced by ~\cite[Theorem~1.1]{ZhanGreen2SLEboundary}, to prove the following result.
\begin{lemma} \label{lem::three_arm_scaling}
	There exists a universal constant $C_{12}\in (0,\infty)$  such that 
	\begin{equation*}
		\lim_{m\to \infty} \delta_{m}^{-2}\times \mathbb{P}\left[\partial B_{\delta_{m}}(0)\xlongleftrightarrow[\mathbb{H}]{BWB} \partial B_1(0)\right]=C_{12}.
	\end{equation*}
\end{lemma}

\begin{proof}[Proof of Theorem~\ref{thm::three_arm}]
{Theorem~\ref{thm::three_arm} can be proved using the strategy in the proof of Theorem~\ref{thm::4arm_expansion}, with Lemma~\ref{lem::four_arm_coupling_inner} replaced by Lemmas~\ref{lem::two_arm_coupling_half} and~\ref{lem::three_arm_coupling}, Lemma~\ref{lem::faces} replaced by an (easier) half-plane version involving the polychromatic three-arm event, and~\eqref{eqn::coro_aux2} replaced by Lemma~\ref{lem::three_arm_scaling}.
In particular, the same strategy used to prove \eqref{eqn::fourarm_expansion_aux1} and \eqref{eqn::fourarm_expansion_aux2} can be applied to $K(x_1,x_2,x_3,x_4)$, {\it mutatis mutandis}, to show that there is a constant $C_K \in (0,\infty)$ such that
\begin{equation*}
\lim_{x_2,x_3\to x} \frac{K(x_1,x_2,x_3,x_4)}{(x_3-x_2)^{1+1-2}} = C_K (x-x_1)^{-(2+1-1)}(x_4-x)^{-(1+2-1)}(x_4-x_1)^{-(1+1-2)},
\end{equation*}
which should be compared to~\eqref{eq:OPE} and~\eqref{eq:OPE-term}, and where the values of the exponents emerge from combining the boundary two-arm exponent, $1$, and the alternating boundary three-arm exponent, $2$.}
\end{proof}

\subsection{Logarithmic correlations of boundary fields: Proof of Theorem~\ref{thm::log2}}	
Suppose that $x_1^a<x_2^a<x_3^a<y^a<x_4^a$ are five vertices on $a\mathcal{T}\cap \mathbb{R}$. We let $\mathcal{M}^a(x_1^a,x_2^a,x_3^a,x_4^a)$ denote the event that there exist two disjoint black paths, $\ell_1^a,\ell_2^a$, and a white path, $\ell_3^a$, in the upper half-plane, such that: (1) $\ell_1^a$ connects $x_1^a$ to the segment $[x_3^a,x_4^a]$; (2) $\ell_2^a$ connects $x_1^a$ to $x_2^a$; (3) $\ell_3^a$ connects $x_2^a+a$ to the segment $[x_3^a,x_4^a]$; (4) {the lowest white path from {$x^a_2+a$} to $[x^a_3,x^a_4]$ is connected by a black path to the black path from $x^a_1$ to $[x^a_3,x^a_4]$.}
Moreover, we let $\hat{\mathcal{M}}^a(x_1^a,x_2^a,y^a)$ denote the following event: (1) there are two disjoint {(except for $x^a_1$)} black paths in the upper half-plane connecting $x_1^a$ to $x_2^a$ and to $y^a$, respectively; (2) there is a white path in the upper half-plane connecting $x_2^a+a$ to $y^a-a$; (3) the lowest white path from {$x^a_2+a$} to $[x^a_3,x^a_4]$ is connected by a black path to the black path from $x^a_1$ to $[x^a_3,x^a_4]$. See Figure~\ref{fig::L3} for an illustration of these two events.
 		
		\begin{figure} 
  \begin{subfigure}[b]{0.45\textwidth}
          			\begin{center}
		\includegraphics[width=0.9\textwidth]{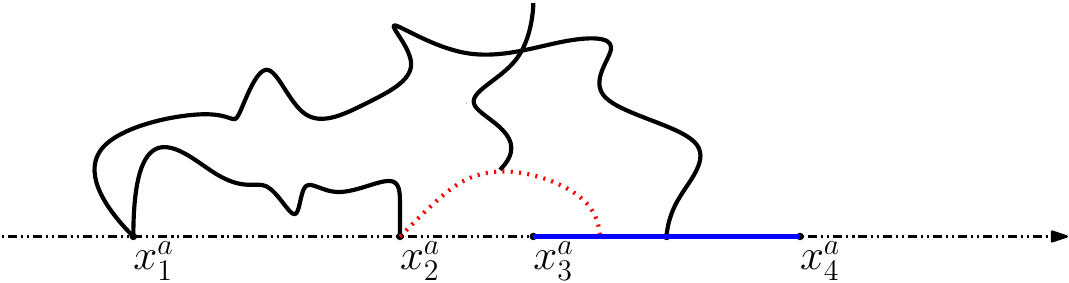}
			\end{center}
  \end{subfigure}
    \begin{subfigure}[b]{0.45\textwidth}
          			\begin{center}
		\includegraphics[width=0.9\textwidth]{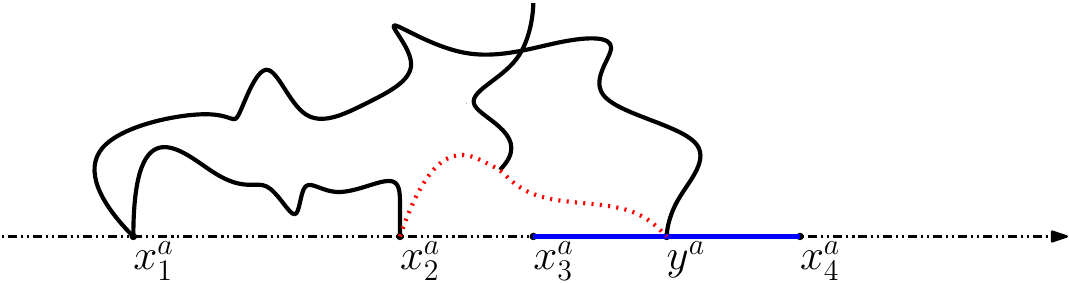}
			\end{center}
  \end{subfigure}

			\caption{The events $\mathcal{M}^a(x_1^a,x_2^a,x_3^a,x_4^a)$ and $\hat{\mathcal{
   M}}^a(x_1^a,x_2^a,y^a)$. The black, solid lines represent black paths and the red, dotted lines represent white paths.}
			\label{fig::L3}
		\end{figure}
	
\begin{lemma} \label{lem::equal_proba}
    Suppose that $x_1^a<x_2^a<x_3^a<x_4^a$ are four vertices of $a\mathcal{T}\cap \mathbb{R}$. Then we have 
    \begin{equation} \label{eqn::equal_proba}
        \mathbb{P}^a\left[\mathcal{L}^a(x_1^a,x_2^a,x_3^a,x_4^a)\right]= \mathbb{P}^a\left[\mathcal{M}^a(x_1^a,x_2^a,x_3^a,x_4^a)\right].
    \end{equation}
\end{lemma}	
\begin{proof}
    We define a {probability-preserving} bijection $\tau: \mathcal{L}^a(x_1^a,x_2^a,x_3^a,x_4^a)\to \mathcal{M}^a(x_1^a,x_2^a,x_3^a,x_4^a)$ between percolation configurations, as follows. 
    
   We let $V(a\mathcal{T})$ denote the vertex set of $a\mathcal{T}$. For a given configuration $\omega\in \mathcal{L}^a(x_1^a,x_2^a,x_3^a,x_4^a)$, we can write $\omega=(\omega_v)_{v\in \{V(a\mathcal{T})\}}\in \{0,1\}^{V(a\mathcal{T})}$, where $\omega_v=1$ (resp., $\omega_v=0$) means that the label at $v$ is black (resp., white). We explore the region enclosed by $\mathbb{R}$ and the lowest black path in the upper half-plane connecting $x_1^a$ to $x_2^a$, and the region enclosed by $\mathbb{R}$ and the lowest white path in the upper half-plane connecting {$x_2^a+a$} to the segment $[x_3^a,x_4^a]$. Let $D(\omega)$ denote the set of vertices contained in the explored regions and in the two lowest paths. {It is a standard percolation result that one can perform the exploration without gaining any information about the regions above the lowest paths.} We define
   \begin{equation*}
       \tau(\omega)_v:= \begin{cases}
           \omega_v, &\text{if }v\in D(\omega),\\
           1-\omega_v, & \text{if }v\in V(a\mathcal{T})\setminus D(\omega). 
       \end{cases}
   \end{equation*}
   Then $\tau: \mathcal{L}^a(x_1^a,x_2^a,x_3^a,x_4^a)\to \mathcal{M}^a(x_1^a,x_2^a,x_3^a,x_4^a)$ is a bijection, which gives~\eqref{eqn::equal_proba}, as desired.
\end{proof}

\begin{lemma}  \label{lem::L3}
Let $-\infty<x_1<x_2<y<\infty$ be three distinct points. Suppose that $x_1^a<x_2^a<y^a\in a\mathcal{T}\cap \mathbb{R}$ satisfy $x_1^a\to x_1$, $x_2^a\to x_2$ and $y^a\to y$ as $a\to 0$. Then there exists a universal constant $C_{13}\in (0,\infty)$ such that 
\begin{equation}\label{eqn::L}
	\hat{M}(x_1,x_2,y):=\iota_a^{-3}\times \mathbb{P}^a\left[\hat{\mathcal{M}}(x_1^a,x_2^a,y^a)\right]= C_{13}(y-x_1)^{-1}(y-x_2)^{-1}(x_2-x_1)^{-1}.
\end{equation}
\end{lemma}
	
\begin{proof}
{It is clear that the proof of Lemma~\ref{lem::backbone} works also on the upper half-plane, that is,} if we replace the event $\{0\xlongleftrightarrow{BB}\partial B_{\epsilon}(0)\}$ by the event $\{0\xlongleftrightarrow[\mathbb{H}]{BB}\partial B_{\epsilon}(0)\}$ and the event $\mathcal{B}_{\eta,\epsilon}^a(0)$ by the event that there are two disjoint black paths in the upper half-plane connecting $\partial B_{\eta}(0)$ to $\partial B_{\epsilon}(0)$. Then, thanks to this observation and Lemma~\ref{lem::two_arm_coupling_half}, one can use the strategy in the proof of Theorem~\ref{thm::backbone} to show that
\begin{equation*}
	\hat{M}(x_1,x_2,x_3):={\lim_{a \to 0}} \, \iota_a^{-3} \, \mathbb{P}^a\left[\hat{\mathcal{M}}^a(x_1^a,x_2^a,x_3^a)\right]
\end{equation*}
exists and belongs to $(0,\infty)$ and that, for any non-constant M\"obius transformation $\varphi:\mathbb{H}\to\mathbb{H}$ satisfying $\varphi(x_1),\varphi(x_2),\varphi(y)\neq \infty$, we have
\begin{equation*}
	\hat{M}(\varphi(x_1),\varphi(x_2),\varphi(y))=\hat{M}(x_1,x_2,y)\times |\varphi'(x_1)|^{-1}\times|\varphi'(x_2)|^{-1}\times |\varphi'(y)|^{-1}.
\end{equation*}
As a consequence, there exists a universal constant $C_{13}\in (0,\infty)$ such that~\eqref{eqn::L} holds, as desired.
\end{proof}

\begin{proof}[Proof of Theorem~\ref{thm::log2}]
We claim that 
\begin{equation} \label{eqn::log_axu1}
	\mathcal{M}^a(x_1^a,x_2^a,x_3^a,x_4^a)=\cup_{x_3^a < y^a \leq x_4^a} \hat{\mathcal{M}}^a(x_1^a,x_2^a,y^a).
\end{equation}
Indeed, on the one hand, it is clear that $\cup_{x_3^a< y^a\leq x_4^a} \hat{\mathcal{M}}^a(x_1^a,x_2^a,y^a)\subseteq \mathcal{M}^a(x_1^a,x_2^a,x_3^a,x_4^a)$. {On the other hand, $\mathcal{M}^a(x_1^a,x_2^a,x_3^a,x_4^a)$ implies the existence of a white path in the upper half-plane connecting $x_2^a+a$ to the segment $[x_3^a,x_4^a]$ and, if $\mathcal{M}^a(x_1^a,x_2^a,x_3^a,x_4^a)$ occurs, there is a rightmost vertex in $[x_3^a,x_4^a]$ connected to $x_2^a+a$ by a white path in the upper half-plane.
Denoting this rightmost vertex by $y_*^a$, we see that $\mathcal{M}^a(x_1^a,x_2^a,x_3^a,x_4^a)$ implies $\hat{\mathcal{M}}^a(x_1^a,x_2^a,y^a_*+a)$, for some $y_*^a\in a\mathcal{T}\cap [x_3^a,x_4^a-a]$.
Therefore, $\mathcal{M}^a(x_1^a,x_2^a,x_3^a,x_4^a)\subseteq \cup_{x_3^a< y^a\leq x_4^a} \hat{\mathcal{M}}^a(x_1^a,x_2^a,y^a)$.}

Furthermore, it is clear that, due to topological constraints,
\begin{equation} \label{eqn::log_aux2}
	\hat{\mathcal{M}}^a(x_1^a,x_2^a,y_1^a)\cap 	\hat{\mathcal{M}}^a(x_1^a,x_2^a,y_2^a)=\emptyset, \quad \forall y_1^a\neq y_2^a \in a\mathcal{T}\cap(x_2^a+a,\infty).
\end{equation}
Combining~\eqref{eqn::log_axu1} and~\eqref{eqn::log_aux2} with Lemma~\ref{lem::equal_proba}, we obtain
\begin{equation*}
\mathbb{P}^a\left[\mathcal{L}^a(x_1^a,x_2^a,x_3^a,x_4^a)\right]=\sum_{x_3^a< y^a\leq x_4^a}\mathbb{P}^a\left[\hat{\mathcal{M}}^a(x_1^a,x_2^a,y^a)\right].
\end{equation*}
Thus, using Remark~\ref{rem::sharp_arm} and Lemma~\ref{lem::L3}, we have
\begin{align*}
	& L(x_1,x_2,x_3,x_4) := \lim_{a\to 0} \iota_a^{-2} \times \mathbb{P}^a\left[\mathcal{L}^a(x_1^a,x_2^a,x_3^a,x_4^a)\right] \\
	& \qquad \qquad \qquad = \lim_{a\to 0} \frac{\iota_a}{a}\sum_{x_3^a< y^a\leq x_4^a} a \iota_a^{-3} \times \mathbb{P}^a\left[\hat{\mathcal{M}}^a(x_1^a,x_2^a,y^a)\right] \\
	& \qquad \qquad \qquad = c_{\iota} \int_{x_3}^{x_4} \hat{M}(x_1,x_2,y)\ud y \\
	& \qquad \qquad \qquad = c_{\iota}C_{13}\int_{x_3}^{x_4} |y-x_1|^{-1}|x_2-x_1|^{-1}|y-x_2|^{-1} \ud y \\
	& \qquad \qquad \qquad = \frac{c_{\iota}C_{13}}{(x_2-x_1)^2} \times\log \frac{(x_4-x_2)(x_3-x_1)}{(x_3-x_2)(x_4-x_1)},
\end{align*}
which gives~\eqref{eqn::log_formula2} with $C_{M}:=c_{\iota}C_{13}$, where $c_{\iota}$ and $C_{13}$ are constants in Remark~\ref{rem::sharp_arm} and Lemma~\ref{lem::L3}, respectively.

Finally,~\eqref{eqn::log_asy2} and \eqref{eqn::log_cov2} follow immediately from~\eqref{eqn::log_formula2}, which completes the proof. 
	\end{proof}

\end{document}